\documentclass[a4paper,12pt]{article}
\usepackage{tensor}

\usepackage[utf8]{inputenc}
\usepackage[english]{babel}

\usepackage{amsmath,amsthm}

\usepackage{amsfonts}
\usepackage{amssymb}

\usepackage{array}

\usepackage{mathrsfs}
\usepackage{hyperref}

\usepackage[pdftex]{graphicx}
\usepackage{xcolor}

\theoremstyle{plain}
\newtheorem{theorem}{Theorem}

\theoremstyle{definition}

\theoremstyle{remark}
\newtheorem{remark}[theorem]{Remark}

\begin{document}

\markboth{R. Niardi}
{Gauge Field Theories and Propagators in Curved Space-Time}

\date{\today}

\author{Roberto Niardi ORCID: 0000-0001-9216-3322\thanks{E-mail: roberto93niardi@gmail.com} \\
	Dipartimento di Fisica \lq\lq Ettore Pancini'', \\
	Università degli Studi di Napoli \lq\lq Federico II '', Italy
}

\title{GAUGE FIELD THEORIES AND PROPAGATORS IN CURVED SPACE-TIME}

\maketitle

\begin{abstract}
In this paper DeWitt's formalism for field theories is presented; it provides a framework in which the quantization of fields possessing infinite dimensional invariance groups may be carried out in a manifestly covariant (non-Hamiltonian) fashion, even in curved space-time.
Another important virtue of DeWitt's approach is that it emphasizes the common features of apparently very different theories such as Yang-Mills theories and General Relativity; moreover, it makes it possible to classify all gauge theories in three categories characterized in a purely geometrical way, i.e., by the algebra which the generators of the gauge group obey; the geometry of such theories is the fundamental reason underlying the emergence of ghost fields in the corresponding quantum theories, too. These \lq\lq tricky extra particles'', as Feynman called them in 1964, contribute to a physical observable such as the stress-energy tensor, which can be expressed in terms of Feynman's Green function itself.
Therefore, an entire section is devoted to the study of the Green functions of the neutron scalar meson: in flat space-time, the choice of a particular Green's function is the choice of an integration contour in the \lq\lq momentum'' space; in curved space-time the momentum space is no longer available, and the definition of the different Green functions requires a careful discussion itself. After the necessary introduction of bitensors, world function and parallel displacement tensor, an expansion for the Feynman propagator in curved space-time is obtained.
Most calculations are explicitly shown.
\end{abstract}

\section{Introduction}
The study of quantum g\-a\-u\-g\-e f\-i\-e\-l\-d t\-h\-e\-o\-r\-i\-e\-s a\-n\-d g\-r\-a\-v\-i\-t\-a\-t\-i\-o\-n i\-s b\-o\-t\-h a\-n i\-n\-t\-e\-l\-l\-e\-c\-t\-u\-a\-l p\-u\-r\-s\-u\-i\-t a\-n\-d a n\-e\-c\-e\-s\-s\-i\-t\-y c\-o\-n\-n\-e\-c\-t\-e\-d t\-o b\-l\-a\-c\-k h\-o\-l\-e t\-h\-e\-o\-r\-y a\-n\-d q\-u\-a\-n\-t\-u\-m c\-o\-s\-m\-o\-l\-o\-g\-i\-c\-a\-l m\-o\-d\-e\-l\-s. 
On one hand, electroweak and strong interactions are described by Yang-Mills theories, and are satisfactorily set in a quantum framework; on the other hand, although General Relativity is in many ways a gauge theory too, it stands apart from the other three forces of nature, and \textit{a quantum theory of gravity does not yet exist}, at least as a coherent discipline.
Nevertheless, it is possible to discuss the influence of the gravitational field on quantum phenomena: one can study the regime for quantum aspects of gravity in which the \textit{gravitational field} is described as a \textit{classical background} through Einstein's theory while \textit{matter fields are quantized}; this is reasonable as long as length and time scales of quantum processes of interest are greater than the Planck values ($l_{Planck} \equiv (G\hbar /c^3)^{1/2} \sim 1.616 \times 10^{-33}{\rm cm}$, $t_{Planck} \equiv (G\hbar /c^5)^{1/2} \sim 5.39 \times 10^{-44}{\rm s}$). Since Planck length is so small (twenty orders of magnitude below the size of an atomic nucleus), one can hope that such a \lq\lq semiclassical" approach has some predictive power. Therefore, one is naturally led to the subject of quantum field theory in a curved background spacetime. Its basic physical prediction is that \textit{strong gravitational fields can \lq\lq polarize" the vacuum and, when time-dependent, lead to pair creation}; moreover, in a curved space-time, notions of \lq\lq vacuum" and \lq\lq particles" need a deeper discussion than in the flat case. These two fundamental results are strongly linked to the most important predictions of the theory, i.e., Hawking and Unruh effects (see \cite{hawking1974black}, \cite{hawking1975particle} and \cite{unruh1976notes}): according to the Hawking effect, a classical, spherically symmetric \textit{black hole} of mass $M$ has the \textit{same spectrum of emission of a black body} having the temperature $T_{Hawking} \equiv \frac{1}{8\pi M}$; according to the Unruh effect, \textit{from the point of view of an accelerating observer}, \textit{empty space contains a gas of particles at a temperature proportional to the acceleration}. For a detailed treatise on these subjects, see also \cite{birrell1984quantum}, \cite{fulling1989aspects} and \cite{parker2009quantum}.

Th\-is \- pa\-pe\-r \- is \- s\-tr\-uc\-tu\-re\-d \- in \- six se\-ct\-io\-ns, i\-nc\-lu\-di\-n\-g t\-he i\-nt\-ro\-du\-ct\-io\-n; 
in Sec. 2, gauge field theories are introduced in DeWitt's formalism \cite{de1983houches}, \cite{dewitt1988dynamical},  \cite{dewitt1975quantum1}, \cite{dewitt2003global}, and the \textit{set of all fields} is presented as an \textit{infinite-dimensional manifold}, on which an action functional is defined; then gauge transformations are viewed as flows which leave the action functional unchanged; whenever the \textit{gauge group realization} is \textit{linear} (this is the case for Yang-Mills theories and General Relativity; see \cite{palais1957imbedding} and \cite{mostow1957equivariant} for precious results in the finite-dimensional case), \textit{manifest covariance} is ensured, i.e., the group transformation laws for the various symbols that appear in the theory may be inferred simply from the position and nature of their indices, and both sides of any equation transform similarly. Last, the theory of small disturbances is discussed, Green's functions are introduced in a general fashion and Peierls brackets are defined, in light of their importance in the quantization procedure.

In Sec. 3, \textit{quantization of non-gauge field theories} is discussed in the framework of DeWitt's formalism: problems with the heuristic quantization rules are stressed, and Schwinger's variational principle is introduced as a way to get around them; then the operator dynamical equations are presented, with the necessary introduction of the measure functional, which, at the simplest level, may be thought of as correcting the lack of self-adjointness of the \lq\lq time-ordered'' version of the classical dynamical equations; nevertheless, it plays a far deeper role, closely linked to the Wick rotation for the evaluation of divergent integrals. Last, the path integral for non-gauge theories is derived.

In Sec. 4, the treatment is extended to gauge theories: some insight is given about the \textit{geometrical structure of the infinite-dimensional manifold of field theories} belonging to the same class as Yang-Mills and General Relativity; under the group action, this space is separated into \textit{orbits}; one could say that it is in the \textit{space of orbits} that the real physics of the system takes place; then one can choose on the manifold \textit{a new set of coordinates} made of two parts: the first one labelling the orbit, and the second one labelling a particular field configuration belonging to the specified orbit; when one writes down the functional integral and reverts to the original coordinates, one finds out that a new term appears: it is the \textit{ghost contribution}, which involves two ghost fields; they are fermionic for a bosonic theory and bosonic for a fermionic theory. Therefore, an \lq\lq extended'' space may be introduced, where ghost fields appear too and a new action functional, containing also ghost terms, may be defined: then one is fascinated to find out that the gauge transformations for the original theory correspond to a set of rigid invariance transformations for the theory on the extended space, i.e., the \textit{BRST transformations}. 
It is important to stress that, altough many manipulations contained in this paper are highly formal, many rigorous results concerning functional integration and path integral can be found in \cite{simon1979functional}, \cite{glimm2012quantum}, \cite{cartier2006functional}.

Section 5 is dedicated to the study of the \textit{Green functions} of free, massive, scalar field, first in flat, then in curved space-time; it is shown that all Green functions can be derived from Feynman's Green function: in flat space-time, one can pass to the \lq\lq momentum'' space and easily verify that the choice of the Green function is the choice of an integration contour which passes around two poles; \textit{in curved space-time}, where \textit{the momentum space is no longer available}, one can nevertheless define the Feynman's Green function: it is the one that obeys the same variational law as finite square matrices and is symmetric. In sight of the curved space-time treatment, it is of fundamental importance to derive, even in the flat case, \textit{a formula for the Feynman propagator involving space-time coordinates only}; therefore, after the introduction of some necessary mathematical tools such as bitensors and some geometric quantities such as the world function (see \cite{hawking1973large}, \cite{milnor1969morse}, \cite{relativity1960general}), \textit{an expansion for the Feynman Green function in curved space-time} is obtained, \textit{valid for small values of the geodetic distance}. It is important to observe that the expansion obtained ceases to exist when the massless limit is taken: the massless case has to be treated with alternative methods.

In the conclusions some applications are outlined, with their related literature.

\section{Field Theories in DeWitt's Formulation}
\subsection{A few words on space-time (I)}
Basic to the whole of quantum field theory is the assumption that space-time, which we shall denote by $M$ (for manifold), has the topological structure 
\begin{equation}
M = \mathcal{R} \times \Sigma
\end{equation}
where $\mathcal{R}$ is the real line and $\Sigma$ is some connected three-dimensional manifold, compact or non-compact. In particular, space-time will be assumed to be endowed with a hyperbolic metric $g$ which admits a foliation of space-time into spacelike sections, each being a complete Cauchy hypersurface (i.e., a spacelike surface which intersects every non spacelike curve \textit{exactly once}) and a topological copy of $\Sigma$. Being characterized by a manifold and a tensor field defined on the manifold, it is more accurate to say that a \lq\lq space-time'' is an equivalence class of pairs $(M,g)$: the equivalence relation is the following:
\begin{equation}
(M,g) \sim (N,h) \longleftrightarrow \exists\psi \in Diff(M) | \; \; N=\psi(M), \; g=\psi^* h,
\end{equation}
where $\psi^*$ is the pullback map associated to $\psi$.
Throughout this work, the following sign convention will be assumed for the signature of the metric tensor: $(-.+,+,...)$.
\subsection{Space of histories and functional differentiation}
In this section, DeWitt's formalism for field theories will be introduced, both bosonic and fermionic ones.

Denote by $\Phi$ the set, or space, of all possible field histories; it will be useful to view $\Phi$ as an infinite-dimensional manifold; in this work the coordinates $\phi^{i}$ will be assumed to be real-valued, whether $c$-type or $a$-type (see \cite{berazin2012method} and \cite{dewitt1992supermanifolds}).  The concept of differentiation on $\Phi$, on which the idea of tangent space at $\phi \in \Phi$ is based, can be introduced through \textit{functional derivative}:

Let $F: \Phi \ni \phi \mapsto F[\phi] \in \Lambda_{\infty}$, where $\Lambda_{\infty}$ is the \textit{supernumber algebra}; $F$ is called a supernumber-valued scalar field or \textit{functional} on $\Phi$, and its value at a point $\phi$ of $\Phi$ is denoted by $F[\phi]$. Let $\delta \phi$ be an infinitesimal variation in $\phi$; it can be represented by a set of functions $\delta \phi^{i}$ on the manifold $M$, where, at each point $x \in M$, the $\delta \phi^{i}(x)$ are components, in the appropriate chart of $\Phi$, of an infinitesimal vector in its tangent space at the point having coordinates $\phi^{i}(x)$. Let the $\delta \phi^{i}(x)$ be $C^{\infty}$ and have compact support in $M$, and let $\delta F[\phi]$ denote the change in value that $F[\phi]$ undergoes in shifting from $\phi$ to $\phi + \delta \phi$. If, for all $\phi \in \Phi$ and for all $C^{\infty}$ variations $\delta \phi$ of compact support, $\delta F[\phi]$ can be written in the form\footnote{The summation convention over repeated indices is assumed throughout this work.}
\begin{equation}
\delta F[\phi] = \int_{M} \delta \phi^{i}(x) \; _{i(x),}F[\phi] \; d^{n}x = \int_{M}   F_{,i(x)}[\phi] \; \delta\phi^{i}(x) \; d^{n}x, \label{funct_var}
\end{equation}
where $_{i(x),}F[\phi]$, $F_{,i(x)}[\phi]$ in the integrands are independent of $\delta \phi^{i}$ and depend at most on $\phi$, then $F$ is called a \textit{differentiable functional} on $\Phi$, and $_{i(x),}F[\phi]$, $ F_{,i(x)}[\phi]$ are called \textit{left} and \textit{right} functional derivatives of $F$, respectively: 
\begin{eqnarray}
_{i(x),}F[\phi] &\equiv & \dfrac{\overrightarrow{\delta}}{\delta\phi(x)^i}F[\phi], \label{left_funct_deriv}\\
F_{,i(x)}[\phi] &\equiv & F[\phi]\dfrac{\overleftarrow{\delta}}{\delta\phi(x)^i}. \label{right_funct_deriv}
\end{eqnarray}
Differentiation will be indicated by a comma followed by one or more indices: Greek indices will denote differentiation with respect to the chart coordinates $x^{\mu}$ in $M$, while Latin indices will denote differentiation with respect to the field coordinates.

In a repeated functional derivative it does not matter whether the left differentiations or the right differentiations are performed first, but the order of the induced indices on either side is important. That is, although left differentiations commute with right differentiations, left differentiations do not generally commute with each other, nor do right differentiations. The laws for interchanging are
\begin{equation}
_{ij',}F= (-1)^{ij'}\,_{j'i,}F \; \; \; \; \; F_{,ij'}= (-1)^{ij'}F_{,j'i},
\end{equation}
where the convention is here adopted that an index or symbol appearing in an exponent of $(-1)$ is to be understood as assuming the value $0$ or $1$ according as the associated quantity is $c$-type or $a$-type. 

If the functional $F$ is \textit{pure}, i.e., either $c$-number-valued or $a$-number-valued, then \eqref{funct_var} implies that its left and right functional derivatives are related by 
\begin{equation}
_{i,}F= (-1)^{i(F+1)}F_{,i}. \label{left_right_der}
\end{equation}
In the previous equation too, the symbol $F$ in the exponent of $(-1)$ assumes the value $0$ if $F$ is a $c$-type quantity, the value $1$ if $F$ is a $a$-type quantity.

When indices appear in exponents of $-1$ a special rule must be introduced regarding the summation convention: although an index appearing in an exponent of $-1$ may participate in the summation induced by its appearance twice elsewhere in a term of a given expression, it may not itself induce a summation. 
\subsection{Condensed and supercondensed notations}
In developing the general formalism of field theory we shall find it often convenient to lump the symbol $x$ with the generic index $i$ and to make the latter do double duty as a discrete label for the field components and as a continuous label for the points 
of space-time. With this notation, the symbol $\tensor{\delta}{^i _{j'}}$ should be understood as a combined $\delta$-distribution Kronecker delta, while Kronecker deltas shall be $\tensor{\underline{\delta}}{^i _j}$ for the sake of clarity:
\begin{equation}
\phi^i \frac{\overleftarrow{\delta}}{\delta \phi^{j'}} \equiv \phi^i(x) \frac{\overleftarrow{\delta}}{\delta \phi^{j}(x')} = \tensor{\delta}{^i _{j'}} \equiv \tensor{\underline{\delta}}{^i _j} \; \delta(x,x')
\end{equation} 
Therefore, it seems natural to establish a new convention: the summation over repeated field indices includes (by virtue of their role as continuous labels) integration over $M$. Hence \eqref{funct_var} takes the form:
\begin{equation}
\delta F[\phi]=\delta \phi^{i} \; _{i,}F[\phi]=F_{,i}[\phi] \; \delta\phi^{i}.
\end{equation}
Sometimes, even this \textit{condensed} notation is cumbersome, and a \textit{supercondensed} notation is used: the indices themselves are suppressed and the following replacement is made:
\begin{equation}
_{i_{r}...i_{1},}F_{j_{1}...j_{s}} \to \; _{r}F_{s}.
\end{equation}
\begin{remark}
	The condensed and supercondensed notations must be used with care because the associative law of multiplication does not always hold. For example, the value of an expression such as $\chi^{i} \, _{i,}F_{,j}\psi^{j}$ may depend on which summation-integration is performed first. They give the same results only in certain cases\footnote{For example, if $F$ is given by the expression $F[\phi]\equiv \frac{1}{2}\int_{M}K_{ij}^{\mu}(x)\phi^{i}(x)\phi^{j}_{,\mu}(x)d^{n}x$ then the reader may easily verify that if the $i$ summation-integration is performed first one gets a result that differs by an amount $-\int_{M}(K_{ij}^{\mu}\phi^{i}\phi^{j})_{,\mu}d^{n}x$ from that obtained when the $j$ summation-integration is performed first. In order to get one result from the other one has to carry out an integration by parts, and this is legitimate only in certain cases, for example if the intersection of the supports of $\chi^{i}$ and $\psi^{j}$ is compact in $M$.}. When the law does not hold, ambiguities in condensed expressions will be removed by the use of parentheses or arrows.
\end{remark}
\subsection{A few words on space-time (II)}
Use of the condensed notation underscores the following point: \textit{The manifold $M$ of space-time, independently of any physical fields that may be imposed on it, is an index set}. Its points are labels that may be lumped together with the indices for field components. 

When $M$ is viewed in this way the notion that \textit{alternative topologies} for space-time may be alternative dynamical possibilities for a given universe makes no sense. Changing the topology of $M$ corresponds to changing the index set, and one cannot change the index set of a theory in midstream. A different index set means a different theory. 

Transitions from one topology to another \textit{could} be followed if space-time were embedded in a higher dimensional manifold endowed with physical properties. But then space-time and its contents would not be all there is; the \lq\lq universe'' would be something bigger. Since nobody has yet developed a successful embedding theory of space-time we shall assume that space-time is the universe and leave its topology fixed.

\subsection{Action functional and dynamical equations}
Throughout this work, the following (fundamental) principle will be postulated: the nature and dynamical properties of a classical dynamical system are completely determined by specifying an \textit{action functional} $S$ for it. 

The action functional is a differentiable real-$c$-number-valued scalar field on the space of histories $\Phi$, i.e., a functionally differentiable mapping $S: \Phi \ni \phi \mapsto S[\phi] \in \mathcal{R}_{c}$. 

The choice of action functional for a given system is not unique but depends on the choice of dynamical variables $\phi^{i}$ used to describe the system and on the boundary conditions that one imposes on the $\phi^{i}$ at the time limits and at spatial infinity. However, all the possible action functionals for a given system must yield equivalent families of \textit{dynamical histories}. A dynamical history is any stationary point of $S$, i.e., any point $\phi$ of $\Phi$ that satisfies 
\begin{equation}
_{i,}S[\phi]= 0 \; \; \; \; \;     \text{or, equivalently}   \; \; \; \; \;  S_{,i}[\phi] =0. \label{dynam_eq}
\end{equation}
The set of all stationary points is called the \textit{dynamical subspace} of $\Phi$, or the \textit{dynamical shell}, and all histories 
satisfying \eqref{dynam_eq} are said to be \textit{on shell}. 

Equations \eqref{dynam_eq} are called the \textit{dynamical equations} of the system. They will be assumed to be \textit{local} in time, i.e., involving no time integrals and not more than a finite number of time derivatives. In a relativistic theory this implies that they must also be local in space. This greatly limits the possible choices for $S$. Throughout this work $S$ will have the general form 
\begin{equation}
S[\phi]= \int_{M} L(\phi^{i},\phi^{i}_{,\mu},x) \, d^{n}x + \text{boundary terms}. \label{action_lagr}
\end{equation}
The integrand of this expression, known as the \textit{Lagrange function} or \textit{Lagrangian}, is a scalar density of unit weight. If the gravitational field is not numbered among the $\phi^{i}$, then $L$ usually has an explicit dependence on a fixed background metric. Expression \eqref{action_lagr} immediately yields
\begin{equation}
0= \, _{i,}S[\phi] \equiv \dfrac{\overrightarrow{\delta}}{\delta\phi^i}L[\phi] - \left(\dfrac{\overrightarrow{\delta}}{\delta\phi^i_{,\mu}}L[\phi]\right)_{,\mu},
\end{equation}
in which the boundary terms do not appear. 

It is worth remarking already at this point that the condition of locality, which is imposed on the dynamical equations largely in order to have easy control over causality, is by no means the only condition that is imposed in practice. Even when the action functional has the structure \eqref{action_lagr} there are additional criteria, of a physical nature, that greatly restrict the Lagrange function itself. For example $L$ must satisfy the constraints imposed by relativistic invariance, either special or general; it should lead to an energy that is bounded from below and the stationary points of the action should be non-trivial.   
\subsection{Invariance transformations}
For many of the most interesting dynamical systems there exists, on the space of histories $\Phi$, a set of \textit{flows} that leave the value of the action invariant. That is, there exists a set of \textit{nowhere vanishing} vector fields $Q_{\alpha}$ on $\Phi$ such that 
\begin{equation}
SQ_{\alpha} \equiv 0. \label{inv_flows}
\end{equation}
Here the vector fields are written as operators acting from the right. They can be expressed in terms of the basis vectors $\dfrac{\overleftarrow{\delta}}{\delta\phi^{i}}$ associated to the coordinate patch on $\Phi$ defined by the $\phi^{i}$:
\begin{equation}
Q_{\alpha}= \dfrac{\overleftarrow{\delta}}{\delta\phi^{i}} \, ^{i}Q_{\alpha}.
\end{equation}
In terms of the components $^{i}Q_{\alpha}$, \eqref{inv_flows} becomes
\begin{equation}
S_{,i} \, ^{i}Q_{\alpha} \equiv 0. \label{inv_flows2}
\end{equation}
Alternative forms are 
\begin{equation}
_{\alpha}QS \equiv 0,   \; \; \; \; \;     \text{or, equivalently}   \; \; \; \; \; _{\alpha}Q^{\sim i}\,_{i,}S \equiv 0,
\end{equation}
where $_{\alpha}Q$ acts from the left, and \lq\lq $\sim$'' denotes the supertranspose: 
\begin{equation}
_{\alpha}Q = (-1)^{\alpha}Q_{\alpha},  \; \; \;  \; \; \;  \,_{\alpha}Q^{\sim i} = (-1)^{\alpha(i+1)} \,^{i}Q_{\alpha}.
\end{equation}
It will be noted that allowance has been made in the previous equation for the possibility that some of the $Q_{\alpha}$ may be $a$-type. The index $\alpha$ is said to be $c$-type or $a$-type according as the vector field that it designates is $c$-type or $a$-type. It should also be noted that eq. \eqref{inv_flows} will generally lead to difficulties at the boundary of the action integral \eqref{action_lagr} unless, for each $\alpha$, the $Q_{\alpha}$ have compact support in $M$. If the $Q_{\alpha}$ must be independent of both boundary conditions and any special coordinate frame in space-time, then they can only be $\delta$-distributions or derivatives of $\delta$-distributions times local functions of the fields and their derivatives. This means that the index $\alpha$, like the index $i$, must include a space-time point and hence range over a continuous infinity of values. 

Because of the invariance equation, the value of the action remains invariant under infinitesimal changes in the dynamical variables of the form
\begin{equation}
\delta\phi^{i}= \,^{i}Q_{\alpha}\delta\xi^{\alpha}. \label{gauge_field}
\end{equation} 
The infinitesimal parameters $\delta\xi^{\alpha}$ of these transformations are $C^{\infty}$ functions over space-time, $c$-number-valued or $a$-number-valued according as the index $\alpha$ is $c$-type or $a$-type. The $\delta\xi^{\alpha}$ will be assumed to be real (i.e., taking their values in either $\mathcal{R}_{c}$ or $\mathcal{R}_{a}$), and since the dynamical variables $\phi^{i}$ are real-valued this implies that the vector fields $Q_{\alpha}$ are real or imaginary according as $\alpha$ is $c$-type or $a$-type. The $\delta\xi^{\alpha}$ are additionally required to have compact support in space-time or else to satisfy such conditions at the time boundaries and at spatial infinity as are needed in order that the integrations in
\begin{equation}
0 \equiv \delta S = S_{,i}\delta\phi^{i}=S_{,i}\,^{i}Q_{\alpha}\delta\xi^{\alpha}
\end{equation} 
be performable in any order.

In the following sections, two important examples of such theories will be presented: \textit{Yang-Mills} theories and \textit{General Relativity}.

\subsection{Yang-Mills theories (I)}
The dynamical object in the $N$-dimensional Yang-Mills theory is a Lie-algebra-valued $1$-form, which can be expressed, in a suitable chart, as follows:
\begin{eqnarray}
A &\equiv & A_{\mu}^{\alpha} \; T_{\alpha}\otimes dx^{\mu} \\
&\equiv & A_{\mu} dx^{\mu},
\end{eqnarray}
where the $T_{\alpha}$ are a basis for the Lie algebra of $SU(N)$, which will be denoted by $su(N)$; for $N \geq 2$, it is the algebra of square anti-hermitian traceless matrices with Lie bracket the commutator; this algebra can be endowed with the Euclidean metric\footnote{this metric is nonsingular if and only if the group is semisimple; additionally, whenever the group is compact, as in this case, then it is positive definite.} 
\begin{equation}
\gamma_{\alpha \beta} \equiv -tr(T_{\alpha}T_{\beta})
\end{equation}
which will be used to lower/raise Lie algebra indices.

Therefore the following Lie-algebra-valued $2$-form is defined:
\begin{eqnarray}
F &\equiv & \frac{1}{2}F_{\mu\nu}^{\alpha} \; T_{\alpha}\otimes \left(dx^{\mu} \wedge dx^{\nu}\right) \\
&\equiv & F_{\mu\nu}^{\alpha} \; T_{\alpha}\otimes dx^{\mu} \otimes dx^{\nu} \\
&\equiv & \left( A_{\nu;\mu}^{\alpha} - A_{\mu;\nu}^{\alpha} +\tensor{f}{^{\alpha}_{\beta\gamma}}A_{\mu}^{\beta}A_{\nu}^{\gamma}\right) T_{\alpha}\otimes dx^{\mu} \otimes dx^{\nu},
\end{eqnarray}
where the semicolon denotes covariant differentiation associated with the Levi-Civita connection $\nabla$, and the $\tensor{f}{^{\alpha}_{\beta\gamma}}$ are the \textit{structure constants} of $su(N)$ associated to the basis of the $T_{\alpha}$, i.e., they satisfy the equation 
\begin{equation}
[T_{\alpha},T_{\beta}] = T_{\gamma}\tensor{f}{^{\gamma}_{\alpha\beta}}.
\end{equation}
The dynamical equations follow from the action functional
\begin{equation}
S_{YM}[A^{\alpha}_{\mu}] \equiv -\frac{1}{4} \int_{M} \sqrt{| \mathrm{g}|} F_{\mu\nu}^{\alpha}F^{\mu\nu}_{\alpha} \, d^{n}x,
\end{equation}
where $\mathrm{g}$ is defined to be ${\rm det}(g_{\mu\nu})$.

It can be readily seen that such an action is invariant under the transformation
\begin{equation}
A_{\mu}(x) \mapsto \tensor*[^{U}]{A}{_{\mu}}(x)\equiv U^{\dagger}(x) A_{\mu}(x) U(x) +  U^{\dagger}(x) U_{,\mu}(x) \label{YM_gauge}
\end{equation}
for every $U: M \ni x \mapsto U(x) \in SU(N)$. Consider an element $T \equiv \omega_{\alpha}T^{\alpha} \in su(N)$; as is well known, \textit{exponentiation} yields an element in $SU(N)$; let now $\omega^{\alpha}(x)$ be a set of real funcions on $M$; then, for every $x \in M$, $U(x) \equiv e^{\omega^{\alpha}(x)T_{\alpha}}$ $\in SU(N)$; if the $\omega^{\alpha}(x)$ are infinitesimal, then $U(x)$ is close to the identical transformation, and \eqref{YM_gauge} reads
\begin{eqnarray}
A^{\gamma}_{\mu} & \mapsto & \tensor*[^U]{A}{_{\mu}^{\gamma}} \equiv A^{\gamma}_{\mu}+ A^{\alpha}_{\mu}\omega^{\beta}\tensor{f}{^{\gamma}_{\alpha\beta}}+\omega^{\gamma}_{,\mu}, \\
\delta A^{\gamma}_{\mu} & \equiv & \tensor*[^U]{A}{_{\mu}^{\gamma}}-A^{\gamma}_{\mu} =A^{\alpha}_{\mu}\omega^{\beta}\tensor{f}{^{\gamma}_{\alpha\beta}}+\omega^{\gamma}_{,\mu}. \label{YM_inf_gauge}
\end{eqnarray}
With DeWitt's notation, eq. \eqref{YM_inf_gauge} can be written
\begin{eqnarray}
\delta A^{\gamma}_{\mu} &=& \tensor*[^{\gamma}_{\mu}]{Q}{_\alpha} \delta \xi^{\alpha}, \\
\tensor*[^{\gamma}_{\mu}]{Q}{_\alpha}(x,x') &\equiv & \delta(x,x') \left( A^{\beta}_{\mu}(x)\tensor{f}{^{\gamma}_{\beta\alpha}} + \underline{\delta}^{\gamma}_{\alpha}\partial_{\mu'} \right)  \\
&=&\delta(x,x') A^{\beta}_{\mu}(x)\tensor{f}{^{\gamma}_{\beta\alpha}} - \underline{\delta}^{\gamma}_{\alpha}\delta(x,x')_{,\mu '}. \label{YM_Q}
\end{eqnarray}
\subsection{General Relativity (I)}
The dynamical object is the Lorentzian metric tensor defined on a manifold $M$; it can be expressed, in a suitable chart, as follows:
\begin{equation}
g=g_{\mu\nu} dx^{\mu} \otimes dx^{\nu}.
\end{equation}
It is important to observe that, if a theory describes nature in terms of a manifold $M$ and tensor fields $T^{(i)}$ defined on the manifold, then if $\phi: M \to N$ is a diffeomorphism, the solutions $(M,T^{(i)})$ and $(N,\phi^{*}T^{(i)})$ have physically identical properties. Any physically meaningful statement about $(M,T^{(i)})$ will hold with equal validity for $(N,\phi^{*}T^{(i)})$. On the other hand, if $(N,\phi^{*}T^{(i)})$ is not related to $(M,T^{(i)})$ by a diffeomorphism, then $(N,\phi^{*}T^{(i)})$ will be physically distinguishable from $(M,T^{(i)})$.

Thus, the \textit{diffeomorphisms} comprise the gauge freedom of any theory formulated in terms of tensor fields on a manifold. In particular, diffeomorphisms comprise the gauge freedom of general relativity.

Now consider the case where $N=M$; let $\phi: M \ni x \mapsto \phi(x) \in M$ be a diffeomorphism; its pullback acts on tensors of type $(0,2)$ as follows:
\begin{eqnarray}
& &\phi^{*}: T_{\phi(x)}^{*}M \otimes T_{\phi(x)}^{*}M \ni g \big|_{\phi(x)}  \mapsto  \phi^{*}g \big|_{\phi(x)} \in T_{x}^{*}M \otimes T_{x}^{*}M, \\
& &\phi^{*}g \big|_{\phi(x)} = g_{\mu\nu} \big|_{\phi(x)}\dfrac{\partial \phi^{\mu}}{\partial x^{\rho}}\bigg|_{x}\dfrac{\partial \phi^{\nu}}{\partial x^{\sigma}}\bigg|_{x}dx^{\rho}\big|_{x} \otimes dx^{\sigma}\big|_{x}. \label{pullback}
\end{eqnarray}
If $\phi$ is close to the identical diffeomorphism, then it can be always seen as an element of the \textit{flow} $\sigma^{\xi}(t,\cdot)\equiv \sigma_{t}$ associated to a \textit{vector field} $\xi$ defined on $M$; hence, for an \lq\lq infinitesimal'' diffeomorphism, eq. \eqref{pullback} reads:
\begin{eqnarray}
\phi^{*}g \big|_{\phi(x)}&=&\sigma^{\xi *}_{\epsilon}g \big|_{\sigma_{\epsilon}^{\xi}(x)}=g_{\mu\nu} \big|_{\sigma^{\xi}_{\epsilon}(x)}\dfrac{\partial \tensor{\sigma}{^{\xi}_{\epsilon}^{\mu}}}{\partial x^{\rho}}\bigg|_{x}\dfrac{\partial \tensor{\sigma}{^{\xi}_{\epsilon}^{\nu}}}{\partial x^{\sigma}}\bigg|_{x}dx^{\rho}\big|_{x} \otimes dx^{\sigma}\big|_{x} \nonumber \\
&=&g\big|_{x}+\epsilon \left(\xi^{\alpha}\partial_{\alpha}g_{\rho\sigma}+g_{\rho\mu}\partial_{\sigma}\xi^{\mu}+g_{\mu\sigma}\partial_{\rho}\xi^{\mu} \right)\big|_{x} dx^{\rho}\big|_{x} \otimes dx^{\sigma}\big|_{x} \nonumber\\
&=& g\big|_{x}+\epsilon \left(\xi_{\rho;\sigma}+\xi_{\sigma;\rho}\right)\big|_{x}dx^{\rho}\big|_{x} \otimes dx^{\sigma}\big|_{x} \nonumber\\
&=& g\big|_{x} + \epsilon \mathcal{L}_{\xi}g \big|_{x}. \label{diff_algebra}
\end{eqnarray}
Equation \eqref{diff_algebra} shows a remarkable result: the Lie algebra of the diffeomorphism group of $M$, $Diff(M)$, is the space of vector fields on $M$ with Lie bracket Lie derivative.

Therefore, General Relativity is invariant under the trasformation
\begin{eqnarray}
g &\mapsto & g + \epsilon \mathcal{L}_{\xi}g, \\
\delta g &=& \epsilon \mathcal{L}_{\xi}g. \label{gauge_gravity}
\end{eqnarray}
for every vector field $\xi$ on $M$.

With DeWitt's notation, eq. \eqref{gauge_gravity} can be written
\begin{eqnarray}
\delta g_{\mu\nu}&=& \tensor*[_\mu_\nu]{Q}{_\rho}\delta \xi^{\rho}, \nonumber \\
\tensor*[_\mu_\nu]{Q}{_\rho}(x,x')&=& \delta(x,x')\left( g_{\mu\rho}\nabla{_\nu}+g_{\nu\rho}\nabla_{\mu}\right)\big|_{x'}\\
&=&\partial_{\rho}g_{\mu\nu}\big|_{x}\delta(x,x')-g_{\mu\rho}\big|_{x}\delta(x,x')_{,\nu}-g_{\nu\rho}\big|_{x}\delta(x,x')_{,\mu}. \label{GR_Q}
\end{eqnarray}
\subsection{Commutator of invariance transformations}
Let $B$ be a functional on $\Phi$; by applying two invariance transformations $Q_{\alpha},Q_{\beta}$ to $G$, one arrives at
\begin{eqnarray}
B\left(Q_{\alpha}Q_{\beta}\right)&\equiv &\left(BQ_{\alpha}\right)Q_{\beta}  \nonumber\\
&=&\left(B_{,i}\tensor[^i]{Q}{_\alpha}\right)Q_{\beta}  \nonumber\\
&=& B_{,i}\tensor[^i]{Q}{_\alpha_{,j}}\tensor[^j]{Q}{_\beta}+(-1)^{j\alpha+ij}B_{,ij}\tensor[^i]{Q}{_\alpha}\tensor[^j]{Q}{_\beta}
\end{eqnarray} 
while
\begin{eqnarray}
\left(BQ_{\beta}\right)Q_{\alpha}&=&B_{,i}\tensor[^i]{Q}{_\beta_{,j}}\tensor[^j]{Q}{_\alpha}+(-1)^{j\beta+ij}B_{,ij}\tensor[^i]{Q}{_\beta}\tensor[^j]{Q}{_\alpha}  \nonumber \\
&=&B_{,i}\tensor[^i]{Q}{_\beta_{,j}}\tensor[^j]{Q}{_\alpha}+(-1)^{i\beta+ji}B_{,ji}\tensor[^j]{Q}{_\beta}\tensor[^i]{Q}{_\alpha}  \nonumber\\
&=&B_{,i}\tensor[^i]{Q}{_\beta_{,j}}\tensor[^j]{Q}{_\alpha}+(-1)^{i\beta}B_{,ij}\tensor[^j]{Q}{_\beta}\tensor[^i]{Q}{_\alpha}  \nonumber\\
&=&B_{,i}\tensor[^i]{Q}{_\beta_{,j}}\tensor[^j]{Q}{_\alpha}+(-1)^{i\beta+i\beta+ij+\alpha\beta+\alpha j}B_{,ij}\tensor[^i]{Q}{_\alpha}\tensor[^j]{Q}{_\beta}  \nonumber\\
&=&B_{,i}\tensor[^i]{Q}{_\beta_{,j}}\tensor[^j]{Q}{_\alpha}+(-1)^{ij+\alpha\beta+\alpha j}B_{,ij}\tensor[^i]{Q}{_\alpha}\tensor[^j]{Q}{_\beta}.
\end{eqnarray} 
Hence one obtains
\begin{eqnarray}
B\left(Q_{\alpha}Q_{\beta}-(-1)^{\alpha\beta}Q_{\beta}Q_{\alpha}\right)&=& B_{,i}\tensor[^i]{Q}{_\alpha_{,j}}\tensor[^j]{Q}{_\beta}-(-1)^{\alpha\beta}B_{,i}\tensor[^i]{Q}{_\beta_{,j}}\tensor[^j]{Q}{_\alpha} \nonumber \\
&=&B_{,i}\left(\tensor[^i]{Q}{_\alpha_{,j}}\tensor[^j]{Q}{_\beta}-(-1)^{\alpha\beta}\tensor[^i]{Q}{_\beta_{,j}}\tensor[^j]{Q}{_\alpha}\right).
\end{eqnarray}
Thus, given two fields $Q_{\alpha},Q_{\beta}$, their \textit{supercommutator} or \textit{super Lie bracket} $[Q_{\alpha},Q_{\beta}]$ is itself a vector field:
\begin{eqnarray}
[Q_{\alpha},Q_{\beta}] &\equiv & Q_{\alpha}Q_{\beta}-(-1)^{\alpha\beta}Q_{\beta}Q_{\alpha},  \label{supercomm} \\ 
\tensor[^i]{[Q_{\alpha},Q_{\beta}]}{}&=&\tensor[^i]{Q}{_\alpha_{,j}}\tensor[^j]{Q}{_\beta}-(-1)^{\alpha\beta}\tensor[^i]{Q}{_\beta_{,j}}\tensor[^j]{Q}{_\alpha}.
\end{eqnarray}
In the particular case where $B$ is the action functional, it is immediately obvious that 
\begin{equation}
S[Q_{\alpha},Q_{\beta}]=0. \label{comm_inv}
\end{equation}
Hence, the commutator of two invariance transformations is an invariance transformation itself. 

It must be pointed out at once that, for \textit{every} dynamical system there exist, on the space of histories $\Phi$, vector fields that, like the $Q_{\alpha}$, yield zero when acting on the action, i.e., vector fields $V$ of the form 
\begin{equation}
V^{i}=S_{,j}\tensor*[^j]{T}{^i}, \label{skew_fields}
\end{equation}
where $T$ is any antisupersymmetric tensor field:\footnote{The components of $T$ should also have the necessary support or rate-of-fall-off properties in 
	space-time for the implicit summation integration in \eqref{skew_fields} to converge. }
\begin{equation}
\tensor*[^j]{T}{^i}=-(-1)^{ij} \; \tensor*[^i]{T}{^j}.
\end{equation}
Such vector fields, however, vanish on the dynamical shell and are not true flows. They will be called \textit{skew fields}.

It will be assumed that all true flows can be expressed, at each point of $\Phi$, as linear combinations of the $Q_{\alpha}$'s and skew fields at that point, i.e., that the $Q_{\alpha}$'s form a pointwise complete set of flows \textit{modulo} skew fields. Pointwise completeness of the $Q_{\alpha}$'s and eq. \eqref{comm_inv} imply that the supercommutator in \eqref{supercomm} must have the general structure 
\begin{equation}
[Q_{\alpha},Q_{\beta}]=Q_{\gamma} \; \tensor{c}{^\gamma _\alpha _\beta} + T_{\alpha \beta} \; \tensor[_1]{S}{}, \label{flow_bracket}
\end{equation}
or, in component form:
\begin{equation}
\tensor*[^i]{[Q_{\alpha},Q_{\beta}]}{} = \tensor[^i]{Q}{_\gamma} \; \tensor{c}{^\gamma_\alpha_\beta}+\tensor[^i]{T}{_\alpha_\beta^j} \;\tensor*[_{j,}]{S}{},
\end{equation}
where the $\tensor{c}{^\gamma_\alpha_\beta}$ are scalar fields on $\Phi$ and the $\tensor[]{T}{_\alpha_\beta}$ are tensor fields, having the symmetries:
\begin{eqnarray}
\tensor{c}{^\gamma_\alpha_\beta}&=&-(-1)^{\alpha\beta} \; \tensor{c}{^\gamma_\beta_\alpha} \\
\tensor[^i]{T}{_\alpha_\beta^j}&=&-(-1)^{\alpha\beta} \; \tensor[^i]{T}{_\beta_\alpha^j}=-(-1)^{ij+(\alpha+\beta)(i+j)} \; \tensor[^j]{T}{_\alpha_\beta^i}.
\end{eqnarray}
In addition to these symmetries the $\tensor{c}{^\gamma_\alpha_\beta}$ and $\tensor[]{T}{_\alpha_\beta}$ must satisfy functional differential conditions imposed by the Jacobi identity
\begin{equation}
[Q_{\alpha},[Q_{\beta},Q_{\gamma}]]\epsilon^{\gamma\beta\alpha}=0,
\end{equation} 
the $\epsilon^{\alpha\beta\gamma}$ being any coefficients completely antisupersymmetric in their indices:
\begin{equation}
\epsilon^{\alpha\beta\gamma}=-(-1)^{\alpha\beta}\epsilon^{\beta\alpha\gamma}=-(-1)^{\beta\gamma}\epsilon^{\alpha\gamma\beta}.
\end{equation}
\subsection{Gauge algebra, gauge groups and orbits}
One may easily verify that the super Lie bracket of any two skew fields is a skew field. 
By functionally differentiating \eqref{inv_flows2}, one obtains:
\begin{eqnarray}
0&=& (S_{,j} \; \tensor*[^j]{Q}{_\alpha})_{,i}  \nonumber \\
&=&S_{,j} \; \tensor*[^j]{Q}{_\alpha_{,i}}+(-1)^{ij+i\alpha} \; S_{,ji} \; \tensor*[^j]{Q}{_\alpha}  \nonumber \\
&=&S_{,j} \; \tensor*[^j]{Q}{_\alpha_{,i}}+(-1)^{ij+i\alpha+ij}\; S_{,ij} \; \tensor*[^j]{Q}{_\alpha}  \nonumber \\
&=& S_{,j} \; \tensor*[^j]{Q}{_\alpha_{,i}}+(-1)^{i\alpha} \; S_{,ij} \;\tensor*[^j]{Q}{_\alpha},  \nonumber \\
S_{,j} \; \tensor*[^j]{Q}{_\alpha_{,i}} &=&-(-1)^{i\alpha} \; S_{,ij} \; \tensor*[^j]{Q}{_\alpha}. \label{Ward_id}
\end{eqnarray}
By using this identity, one can easily verify that the super Lie Bracket of a $Q_{\alpha}$ with a skew field is again a skew field:
\begin{eqnarray}
S_{,i}  \; ^i[Q_{\alpha},S_{,j} \tensor[^j]{T}{^\bullet}] &=& S_{,i} \left( \tensor[^i]{Q}{_{\alpha}_{,k}} \; S_{,j} \; \tensor[^j]{T}{^k} - (-1)^{\alpha T} \left(S_{,j} \; \tensor[^j]{T}{^i}\right)_{,k} \; \tensor[^k]{Q}{_\alpha} \right)  \nonumber\\
&=& S_{,i} \; \tensor[^i]{Q}{_{\alpha}_{,k}} \; S_{,j} \;  \tensor[^j]{T}{^k} - (-1)^{\alpha T} \; S_{,i}S_{,j} \; \tensor[^j]{T}{^i_{,k}} \; \tensor[^k]{Q}{_\alpha}  \nonumber\\ &&-(-1)^{\alpha T +ki+kT+kj} \; S_{,i}S_{,jk} \; \tensor[^j]{T}{^i} \; \tensor[^k]{Q}{_\alpha}  \nonumber\\
&=&S_{,i}\; \tensor[^i]{Q}{_{\alpha}_{,k}} \; S_{,j} \; \tensor[^j]{T}{^k}-(-1)^{\alpha T +ki+kT+kj} \; S_{,i}S_{,jk} \; \tensor[^j]{T}{^i} \; \tensor[^k]{Q}{_\alpha}  \nonumber\\
&=&S_{,i} \; \tensor[^i]{Q}{_{\alpha}_{,k}} \;S_{,j} \; \tensor[^j]{T}{^k}-(-1)^{\alpha T +\alpha i+\alpha T+\alpha j} \;S_{,i}S_{,jk} \; \tensor[^k]{Q}{_\alpha}\; \tensor[^j]{T}{^i}  \nonumber\\
&=& S_{,i} \; \tensor[^i]{Q}{_{\alpha}_{,k}} \; S_{,j} \; \tensor[^j]{T}{^k} +(-1)^{ \alpha i} \;S_{,i}S_{,k} \;\tensor[^k]{Q}{_\alpha_{,j}} \;\tensor[^j]{T}{^i}  \nonumber\\
&=& (-1)^{jk+ j\alpha+ji} \; S_{,i}S_{,j} \; \tensor[^i]{Q}{_{\alpha}_{,k}} \; \tensor[^j]{T}{^k}+(-1)^{ \alpha i} \; S_{,i}S_{,k} \; \tensor[^k]{Q}{_\alpha_{,j}} \; \tensor[^j]{T}{^i}  \nonumber\\
&=& (-1)^{ik+ i\alpha+ji} \; S_{,j}S_{,i} \; \tensor[^j]{Q}{_{\alpha}_{,k}} \; \tensor[^i]{T}{^k}+(-1)^{ \alpha i} \; S_{,i}S_{,j} \;\tensor[^j]{Q}{_\alpha_{,k}} \; \tensor[^k]{T}{^i}  \nonumber\\
&=& (-1)^{ik+ i\alpha} \; S_{,i}S_{,j} \; \tensor[^j]{Q}{_{\alpha}_{,k}} \; \tensor[^i]{T}{^k}+(-1)^{\alpha i}\; S_{,i}S_{,j} \; \tensor[^j]{Q}{_\alpha_{,k}} \; \tensor[^k]{T}{^i}  \nonumber\\
&=& (-1)^{i \alpha} \; \left[ S_{,i}S_{,j} \; \tensor[^j]{Q}{_{\alpha}_{,k}} \left( (-1)^{ik} \; \tensor[^i]{T}{^k} + \tensor[^k]{T}{^i} \right) \right]  \nonumber\\
&=& 0.
\end{eqnarray} 
From these facts, together with eq. \eqref{flow_bracket}, it follows that the set of all vector fields on $\Phi$ of the form
\begin{equation}
\tensor[^i]{Q}{_{\alpha}} \; \tensor{\xi}{^\alpha}+\tensor[^i]{T}{^j} \; \tensor[_{j,}]{S}{},
\end{equation}
the $\xi^{\alpha}$ being arbitrary ($\phi$-dependent) coefficients, $T$ being an arbitrary antisupersymmetric tensor field, form a closed algebra under the super Lie bracket operation. When true flows exist this algebra is called \textit{a gauge algebra}. 

The vector fields $Q_{\alpha}$ characterizing the flows on $\Phi$ are evidently not unique. They are defined only up to transformations of the form 
\begin{equation}
\tensor[^i]{\bar{Q}}{_\alpha}= \tensor[^i]{Q}{_\beta}\; \tensor{X}{^\beta_\alpha}+\tensor[^i]{T}{_\alpha^j} \; \tensor[_{j,}]{S}{}, \label{Qtransform}
\end{equation}
where the $\tensor{X}{^\beta_\alpha}$ are functionally differentiable scalar fields on $\Phi$ which, at each point of $\Phi$, form the elements of an invertible matrix, whose inverse is formed by functionally differentiable scalar fields, too, while $\tensor[^i]{T}{_\alpha^j}$ obey 
\begin{equation}
\tensor[^i]{T}{_\alpha^j}=-(-1)^{ij+(i+j)\alpha} \; \tensor[^j]{T}{_\alpha^i}.
\end{equation} 
It is easy to see that such transformations leave eq. \eqref{flow_bracket} unchanged. It is also easy to see that even when the $Q_{\alpha}$ are fixed, the $T_{\alpha\beta}$ in eq. \eqref{flow_bracket} are not unique but are determined only up to transformations of the form 
\begin{equation}
\tensor[^i]{\bar{T}}{_\alpha_\beta^j}=\tensor[^i]{T}{_\alpha_\beta^j}+\tensor[^i]{Q}{_\gamma}\tensor[^\gamma]{U}{_\alpha_\beta^\delta}\tensor[_\delta]{Q}{^\sim^j}, \label{Ttransform}
\end{equation}
where the coefficients $\tensor[^\gamma]{U}{_\alpha_\beta^\delta}$ satisfy 
\begin{equation}
\tensor[^\gamma]{U}{_\alpha_\beta^\delta}=-(-1)^{\alpha\beta}\;\tensor[^\gamma]{U}{_\beta_\alpha^\delta}=-(-1)^{\gamma\delta+(\gamma+\delta)(\alpha+\beta)}\;\tensor[^\delta]{U}{_\alpha_\beta^\gamma}.
\end{equation} 
By carrying out these transformations, one may often simplify the relations satisfied by the $Q_{\alpha}$. Three cases may be distinguished:
\subsection{Type-I}
The $Q_{\alpha}$ and the $T_{\alpha\beta}$ may be chosen in such a way that the latter vanish and the $\tensor{c}{^\gamma_\alpha_\beta}$ are $\phi$-independent; then eq. \eqref{flow_bracket} becomes
\begin{eqnarray}
[Q_{\alpha},Q_{\beta}]&=&Q_{\gamma} \tensor{c}{^\gamma _\alpha _\beta}, \label{case_i1} \\ 
\tensor{c}{^\gamma _\alpha _\beta_{,i}}&=&0, \label{case_i2}
\end{eqnarray}
and the Jacobi identity implies:
\begin{equation}
\tensor{c}{^\eta _\alpha _\delta}\tensor{c}{^\delta_\beta_\gamma}\tensor{\epsilon}{^\gamma^\beta^\alpha}=0. \label{case_i3}
\end{equation}
In this case the $\tensor{c}{^\gamma_\alpha_\beta}$ are the \textit{structure constants} of an infinite dimensional Lie group known as the \textit{gauge group} of the system. The gauge group, or more correctly, the \textit{proper gauge group} is defined as the set of transformations of $\Phi$ into itself obtained by exponentiating the transformation \eqref{gauge_field} with $\phi$-independent $\xi^{\alpha}$ and taking products of the resulting exponential maps. The proper gauge group is viewed as acting on $\Phi$, and its actions leave $S$ invariant. The \textit{full gauge group} is obtained by appending to the proper gauge group all other $\phi$-independent transformations that leave $S$ invariant and do not arise from global symmetries. Elements of the proper group are sometimes called \textit{little gauge transformations}, while elements of the full group outside the proper group are called \textit{big gauge transformations}. When big gauge transformations exist the gauge group has disconnected components. It should be remarked that for the systems encountered in practice a choice of flow vectors $Q_{\alpha}$ satisfying \eqref{case_i1} \eqref{case_i2} is usually given a priori, and it is not necessary to carry out transformations of the forms \eqref{Qtransform} \eqref{Ttransform} to find them. The closure property expressed by \eqref{case_i1}, which is stronger than eq. \eqref{flow_bracket}, implies that the gauge group decomposes $\Phi$ into subspaces to which the $Q_{\alpha}$ are tangent. These subspaces are known as orbits, and the point-wise linear independence of the $Q_{\alpha}$ implies that each orbit is a copy of the gauge group supermanifold. If the gauge group has disconnected components, then so does $\Phi$ itself. $\Phi$ may be viewed as a principal fibre bundle of which the orbits are the fibres. The base space of this bundle is a supermanifold of which the orbits may be regarded as the points. It is called the \textit{space of orbits}. The action functional is a scalar field on the space of orbits, and one might be tempted to say that it is in this space that the real physics of the system takes place. However, there may exist physical observables that remain invariant under little gauge transformations but not big ones, so a separate \lq\lq physical'' base space should in principle be assigned to each component of $\Phi$. But in practice the amounts by which physical observables change under a big gauge transformation are always dynamically inert. Therefore we shall from now on focus solely on the proper gauge group. 
\subsection{Type-II}
The $T_{\alpha\beta}$ can be made to vanish but the $\tensor{c}{^\gamma_\alpha_\beta}$ cannot be made $\phi$-independent globally on $\Phi$. Equation \eqref{case_i1} continues to hold, and the space of histories is again decomposed into orbits to which the $Q_{\alpha}$ are tangent, but the orbits are not group supermanifolds. If the $Q_{\alpha}$ are pointwise linearly independent, then the components of the orbits are all topologically identical, and each is a \textit{parallelizable} supermanifold. The space of histories may again be viewed as a fibre bundle, and the real physics of the system takes place in the space of orbit components. The Jacobi identity, in this case, implies 
\begin{equation}
(\tensor{c}{^\eta _\alpha _\delta}\tensor{c}{^\delta_\beta_\gamma}-\tensor{c}{^\eta _\alpha _\beta_{,i}}\tensor[^i]{Q}{_\gamma})\tensor{\epsilon}{^\gamma^\beta^\alpha}=0.
\end{equation}
\subsection{Type-III}
The $T_{\alpha\beta}$ cannot be made to vanish globally on $\Phi$. Flow vectors of the form $Q_{\alpha}\xi^{\alpha}$, where the $\xi^{\alpha}$ are $\phi$-dependent, do not by themselves form a closed system under the super Lie bracket operation, except on the dynamical shell. Only the dynamical shell, not the full space of histories $\Phi$, is decomposed into orbits. The space of histories cannot be viewed as a fibre bundle; only the dynamical shell can. This means that although the real physics takes place in the space of orbit components as usual, the dynamics cannot be derived from an action functional on this space. The full space $\Phi$ is needed. 
\subsection{Yang-Mills theories (II)}
In this section it will be shown that Yang-Mills theories are Type-I theories, and their structure functions will be calculated explicitly. By recalling eq. \eqref{YM_Q}, one obtains
\begin{eqnarray}
\tensor[^{\gamma}_{\mu}]{Q}{_{\alpha'}_{,}^{\nu''}_{\delta''}} &\equiv  & \tensor*[^{\gamma}_{\mu}]{Q}{_{\alpha'}} \dfrac{\overleftarrow{\delta}}{\delta A(x'')_{\nu}^{\delta}}  \nonumber \\
&=& \delta (x,x')\tensor*{\delta}{^{\nu ''}_\mu} \tensor*{\underline{\delta}}{^{\rho}_{\delta}} \tensor{f}{^\gamma _\rho _{\alpha '}}  \nonumber\\
&=& \delta (x,x') \tensor*{\delta}{^{\nu ''}_\mu} \tensor{f}{^\gamma _{\delta} _{\alpha '}}.
\end{eqnarray}
Hence
\begin{eqnarray}
& &\tensor[^{\gamma}_{\mu}]{Q}{_{\alpha'}_{,}^{\nu''}_{\delta''}} \; \; \tensor[^{\delta ''}_{\nu ''}]{Q}{_{\beta'''}} =  \nonumber\\
&=& \int_{M} dx'' \left(\delta (x,x')  \tensor*{\delta}{^{\nu ''}_\mu} \tensor{f}{^\gamma _{\delta} _{\alpha '}} \right) \left( \delta(x'',x''') A^{\rho''}_{\nu''}\tensor{f}{^{\delta}_{\rho''}_{\beta'''}} - \tensor*{\delta}{^{\delta ''}_{\beta'''}_{,\nu''}} \right)  \nonumber\\
&=& \delta (x,x')\delta (x,x''')A^{\rho}_{\mu} \tensor{f}{^\gamma _{\delta} _{\alpha '}}\tensor{f}{^\delta _{\rho} _{\beta'''}} - \delta (x,x')\tensor{f}{^\gamma _{\delta} _{\alpha '}}\tensor*{\delta}{^{\delta}_{\beta'''}_{,\mu}};
\end{eqnarray}
by swapping $(\alpha ',x')$ and $(\beta''',x''')$, one obtains:
\begin{eqnarray}
& &\tensor[^{\gamma}_{\mu}]{Q}{_{\beta'''}_{,}^{\nu''}_{\delta''}}\;\; \tensor[^{\delta ''}_{\nu ''}]{Q}{_{\alpha'}} =  \nonumber \\
&=& \delta (x,x')\delta (x,x''')A^{\rho}_{\mu} \; \tensor{f}{^\gamma _{\delta} _{\beta'''}}\; \tensor{f}{^\delta _{\rho} _{\alpha'}} - \delta (x,x''')\tensor{f}{^\gamma _{\delta} _{\beta'''}}\;\tensor*{\delta}{^{\delta}_{\alpha'}_{,\mu}},
\end{eqnarray}
therefore
\begin{eqnarray}
& &\tensor*[^\gamma _\mu]{[Q_{\alpha},Q_{\beta}]}{}=  \nonumber\\
&=&\tensor[^{\gamma}_{\mu}]{Q}{_{\alpha'}_{,}^{\nu''}_{\delta''}}\;\;\tensor[^{\delta ''}_{\nu ''}]{Q}{_{\beta'''}}-\tensor[^{\gamma}_{\mu}]{Q}{_{\beta'''}_{,}^{\nu''}_{\delta''}}\;\;\tensor[^{\delta ''}_{\nu ''}]{Q}{_{\alpha'}}  \nonumber\\
&=&\delta (x,x')\delta (x,x''')A^{\rho}_{\mu} \left( \tensor{f}{^\gamma _{\delta} _{\alpha '}}\tensor{f}{^\delta _{\rho} _{\beta'''}}- \tensor{f}{^\gamma _{\delta} _{\beta'''}}\tensor{f}{^\delta _{\rho} _{\alpha'}} \right)   \nonumber\\ &&- \left(\delta (x,x')\tensor{f}{^\gamma _{\delta} _{\alpha '}}\tensor*{\delta}{^{\delta}_{\beta'''}_{,\mu}} - \delta (x,x''')\tensor{f}{^\gamma _{\delta} _{\beta'''}}\tensor*{\delta}{^{\delta}_{\alpha'}_{,\mu}}\right).
\end{eqnarray}
The first term contains  $\tensor{f}{^\gamma _{\delta} _{\alpha '}}\tensor{f}{^\delta _{\rho} _{\beta'''}}- \tensor{f}{^\gamma _{\delta} _{\beta'''}}\tensor{f}{^\delta _{\rho} _{\alpha'}}$; by using the Jacobi identity for the structure constants and their antisimmetry in the lower indices, one obtains:
\begin{equation}
\tensor{f}{^\gamma _{\delta} _{\alpha '}}\tensor{f}{^\delta _{\rho} _{\beta'''}}- \tensor{f}{^\gamma _{\delta} _{\beta'''}}\tensor{f}{^\delta _{\rho} _{\alpha'}} = \tensor{f}{^\gamma _{\rho} _{\delta}}\tensor{f}{^\delta _{\alpha'}_{\beta'''}},
\end{equation}  
therefore the first term is
\begin{eqnarray}
& &\delta (x,x')\delta (x,x''')A^{\rho}_{\mu} \tensor{f}{^\gamma _{\rho}_{\delta}}\tensor{f}{^\delta _{\alpha'}_{\beta'''}}=  \nonumber\\
&=& \int_{M}dx'' \tensor{f}{^{\delta''}_{\alpha'}_{\beta'''}}\delta (x'',x')\delta (x'',x''') \left(\delta(x,x'')A^{\rho}_{\mu} \tensor{f}{^\gamma _{\rho} _{\delta''}}\right)
\end{eqnarray} 
the second term is, instead:
\begin{eqnarray}
& &- \left(\delta (x,x')\tensor{f}{^\gamma _{\delta} _{\alpha '}}\tensor*{\delta}{^{\delta}_{\beta'''}_{,\mu}} - \delta (x,x''')\tensor{f}{^\gamma _{\delta} _{\beta'''}}\tensor*{\delta}{^{\delta}_{\alpha'}_{,\mu}}\right) =  \nonumber\\
&=& - \left(\delta (x,x')\tensor{f}{^\gamma _{\beta'''} _{\alpha '}}\tensor*{\delta(x,x''')}{_{,\mu}} - \delta (x,x''')\tensor{f}{^\gamma _{\alpha'} _{\beta'''}}\tensor*{\delta(x,x')}{_{,\mu}} \right) \nonumber\\
&=&-\tensor{f}{^\gamma_{\beta'''}_{\alpha '}}\left( \delta (x,x')\tensor*{\delta(x,x''')}{_{,\mu}}+\tensor*{\delta(x,x')}{_{,\mu}}\delta (x,x''') \right) \nonumber \\
&=&\tensor{f}{^\gamma_{\alpha'}_{\beta'''}} \left( \delta (x,x')\tensor*{\delta(x,x''')}{}+\tensor*{\delta(x,x')}{}\delta (x,x''') \right)_{,\mu}  \nonumber\\
&=&\int_{M}dx'' \tensor{f}{^\gamma_{\alpha'}_{\beta'''}} \delta (x,x'')\left( \delta (x'',x')\tensor*{\delta(x'',x''')}{}\right)_{,\mu}  \nonumber\\
&=&-\int_{M}dx'' \tensor{f}{^\gamma_{\alpha'}_{\beta'''}} \delta (x,x'')_{,\mu} \left( \delta (x'',x')\tensor*{\delta(x'',x''')}{}\right)  \nonumber\\
&=&-\int_{M}dx'' \tensor{f}{^{\delta''}_{\alpha'}_{\beta'''}} \tensor*{\delta}{^{\gamma}_{\delta''}_{,\mu}} \delta (x'',x')\tensor*{\delta(x'',x''')}{} \nonumber \\
&=&\int_{M}dx'' \tensor{f}{^{\delta''}_{\alpha'}_{\beta'''}} \delta (x'',x')\tensor*{\delta(x'',x''')}{} \left(-\tensor*{\delta}{^{\gamma}_{\delta''}_{,\mu}} \right)
\end{eqnarray}
Putting it all together, one obtains, eventually:
\begin{eqnarray}
& &\tensor*[^\gamma _\mu]{[Q_{\alpha},Q_{\beta}]}{}=  \nonumber\\
&=&\int_{M}dx'' \tensor{f}{^{\delta''}_{\alpha'}_{\beta'''}} \delta (x'',x')\tensor*{\delta(x'',x''')}{} \left(\delta(x,x'')A^{\rho}_{\mu} \tensor{f}{^\gamma _{\rho} _{\delta''}}-\tensor*{\delta}{^{\gamma}_{\delta''}_{,\mu}} \right)  \nonumber\\
&=& \int_{M}dx'' \tensor*[^\gamma _\mu]{Q}{_{\delta''}} \tensor{f}{^{\delta''}_{\alpha'}_{\beta'''}} \delta (x'',x')\tensor*{\delta(x'',x''')}{} \label{YM_dim}
\end{eqnarray}
Equation \eqref{YM_dim} shows that Yang-Mills theories are Type-I theories, with structure constants
\begin{equation}
\tensor{c}{^{\delta''} _{\alpha'} _{\beta'''}} \equiv \tensor{f}{^\delta _{\alpha} _{\beta}} \delta (x'',x')\delta(x'',x''').
\end{equation}
\subsection{General Relativity (II)}
By recalling the commutation law for the Lie derivative, one obtains
\begin{equation}
[\mathcal{L}_{X},\mathcal{L}_{Y}]=\mathcal{L}_{[X,Y]};
\end{equation}
But $[X,Y]\big|_{x}=\left(X^{\nu}\partial_{\nu}Y^{\mu}-Y^{\nu}\partial_{\nu}X^{\mu}\right)\big|_{x}\partial_{\mu}\big|_{x} $ can be expressed as
\begin{eqnarray}
[X,Y]\big|_{x} &=& \int_{M} dx' \int_{M} dx'' \tensor{c}{^{\sigma} _{\mu'} _{\nu''}}X^{\mu'}Y^{\nu''},  \\
\tensor{c}{^{\sigma} _{\mu'} _{\nu''}}&=& \tensor*{\delta}{^\sigma_{\mu'}_{;\tau}}\delta^{\tau}_{\nu''}-\tensor*{\delta}{^\sigma_{\nu''}_{;\tau}}\delta^{\tau}_{\mu'}, \label{str_f_GR}
\end{eqnarray}
as is straightforward to verify. These equations show that General Relativity is a Type-I theory too, and its structure functions are given by \eqref{str_f_GR}.
\subsection{Physical Observables}
A change in the dynamical variables of the form \eqref{gauge_field} leaves the action functional invariant. Such changes therefore play no role in determining the dynamical shell. Moreover, they map the dynamical shell into itself, as may be seen by varying the dynamical equations and making use of eq. \eqref{Ward_id}:
\begin{eqnarray}
\delta \tensor*[_{j,}]{S}{}&=& \tensor*[_{j,}]{S}{_{,i}}\delta \phi^{i}  \nonumber\\
&=&\tensor*[_{j,}]{S}{_{,i}} \tensor*[^i]{Q}{_\alpha}\delta \xi^{\alpha}  \nonumber\\
&=& -(-1)^{j\alpha+j}\tensor*{S}{_{,i}} \tensor*[^i]{Q}{_\alpha_{,j}}\delta \xi^{\alpha},
\end{eqnarray}
therefore
\begin{equation}
\tensor*[_{j,}]{S}{}=0 \implies \delta \tensor*[_{j,}]{S}{}=0.
\end{equation}
Hence transformations generated by the $Q_{\alpha}$ are unphysical. No functional of the dynamical variables that is affected by them can be a physical quantity. Conversely, any functional that \textit{is} invariant under \eqref{gauge_field} will be called a \textit{physical observable}. In the classical 
theory this nomenclature constitutes an abuse of language because both $c$-type and $a$-type quantities can be invariant under \eqref{gauge_field}, and of course nobody can observe an $a$-number. However, the quantum counterpart of a real-valued classical observable, whether $c$-type or $a$-type, will, for any valid physical theory, be a self-adjoint linear operator in the super Hilbert space of the full quantum theory, having ordinary real numbers as eigenvalues. 

It is useful to distinguish two types of invariants under \eqref{gauge_field}: \textit{absolute invariants} and \textit{conditional invariants}. An absolute invariant $A$ is a functional of the $\phi^{i}$ that is invariant under \eqref{gauge_field} at all points of $\Phi$. It satisies
\begin{equation}
A Q_{\alpha} = A_{,i}\; \tensor[^i]{Q}{_\alpha} = 0   \; \; \; \; \; \forall\phi \in \Phi. \label{abs_inv}
\end{equation}
The action functional is always an absolute invariant. A conditional invariant $B$ is a funcional of the $\phi^{i}$ that is invariant under \eqref{gauge_field} on shell but not everywhere on $\Phi$. It typically satisfies
\begin{equation}
B Q_{\alpha} = B_{,i} \; \tensor[^i]{Q}{_\alpha} = S_{,i} \; \tensor[^i]{b}{_\alpha}   \; \; \; \; \; \forall\phi \in \Phi,
\end{equation}
where the $\tensor[^i]{b}{_\alpha}$ are certain $\phi^{i}$-dependent coefficients. A simple example of a conditional invariant is 
\begin{equation}
\bar{A}= A+ S_{,i} \;\tensor[^i]{a}{}, \label{obs_transf}
\end{equation}
where $A$ is an absolute invariant and $\tensor[^i]{a}{}$ are arbitrary  $\phi^{i}$-dependent coefficients.

A physical observable may be either an absolute invariant or a conditional invariant. In a physical situation (i.e., when $\phi$ is on shell) there is in fact no distinction between the two. As a functional of the $\phi^{i}$ a physical observable is really defined only modulo the dynamical equations, i.e., up to transformations of the form \eqref{obs_transf}.
\subsection{Gauge groups and manifest covariance}
\textit{Manifest covariance} refers to the following facts: the group transformation laws for the various symbols that appear in the theory may be inferred simply from the position and nature of their indices, and both sides of any equation transform similarly. This applies for Type-I theories when the group realization is linear, i.e.  
\begin{equation}
\tensor[^i]{Q}{_\alpha_{,jk}}=0. \label{linearity}
\end{equation}
In fact, by deriving \eqref{case_i1}, one obtains:
\begin{equation}
\tensor[^i]{Q}{_\alpha_{,j}}\;\tensor[^j]{Q}{_\beta_{,k}}-(-1)^{\alpha\beta}\;\tensor[^i]{Q}{_\beta_{,j}}\;\tensor[^j]{Q}{_\alpha_{,k}}= (-1)^{k(\alpha+\beta+\gamma)}\;\tensor[^i]{Q}{_\gamma_{,k}}\;\tensor{c}{^\gamma _\alpha _\beta}, \label{def_rep}
\end{equation}
which implies that the matrices $(\tensor[^i]{Q}{_\alpha_{,j}})$ (which, in view of eq. \eqref{linearity}, are $\phi$-independent) generate a representation of the Lie algebra of the gauge group and, by exponentiation, of the gauge group itself. Call this representation the \textit{defining representation} and call the contragradient representation (generated by the negative (super)transposes of the above matrices) the \textit{co-defining representation}. Similarly, eq. \eqref{case_i3} implies that $(\tensor{c}{^\gamma _\alpha _\beta})$ generate a representation of the Lie algebra of the gauge group too: call it \textit{adjoint representation}; call \textit{co-adjoint representation} the representation generated by the negative (super)transpose of the structure constants.

Given an absolute invariant $A$, one can take subsequent derivatives of eq. \eqref{abs_inv} and use \eqref{linearity}; the first two derivatives yield:
\begin{eqnarray}
(A_{,i} \tensor[^i]{Q}{_{\alpha}})_{,j}&=&0,  \nonumber\\
(-1)^{ij+j\alpha}A_{,ij} \tensor[^i]{Q}{_{\alpha}}&=& -A_{,i} \tensor[^i]{Q}{_{\alpha}_{,j}}; \\
& &  \nonumber\\
(A_{,i} \tensor[^i]{Q}{_{\alpha}})_{,jk}&=&0,  \nonumber\\
(A_{,i} \tensor[^i]{Q}{_{\alpha}_{,j}}+(-1)^{ij+j\alpha}A_{,ij} \tensor[^i]{Q}{_{\alpha}})_{,k}&=&0,  \nonumber\\
(-1)^{ij+j\alpha+\alpha k+ik}A_{,ijk} \tensor[^i]{Q}{_{\alpha}}&=&-(-1)^{jk+j\alpha+ij}A_{,ik} \tensor[^i]{Q}{_{\alpha}_{,j}}  \nonumber\\ &&-(-1)^{ij+j\alpha}A_{,ij} \tensor[^i]{Q}{_{\alpha}_{,k}}.
\end{eqnarray}
These identities relate functional derivatives of any absolute invariant of adjacent order; when the functional under observation is the action functional, these derivatives are called \textit{vertex functions}, and these identities are called \textit{bare Ward identities}. They imply the transformation laws
\begin{eqnarray}
\delta A_{,j} &\equiv & A_{,ji}\delta \phi^{i}  \nonumber\\
&=& A_{,ji}\tensor[^i]{Q}{_{\alpha}}\delta\xi^{\alpha}  \nonumber\\
&=& (-1)^{ij}A_{,ij}\tensor[^i]{Q}{_{\alpha}}\delta\xi^{\alpha}  \nonumber\\
&=& (-1)^{j\alpha}(-1)^{ij+j\alpha}A_{ji}\tensor[^j]{Q}{_{\alpha}}\delta\xi^{\alpha}  \nonumber \\
&=& (-1)^{j\alpha}\left[(A_{,i} \tensor[^i]{Q}{_{\alpha}})_{,j}-A_{,i} \tensor[^i]{Q}{_{\alpha}_{,j}}\right]\delta\xi^{\alpha}  \nonumber\\
&=& -(-1)^{j\alpha}A_{,i} \tensor[^i]{Q}{_{\alpha}_{,j}}\delta\xi^{\alpha}; \\
& &  \nonumber\\
\delta A_{,jk} &\equiv & A_{,jki}\delta \phi^{i}  \nonumber\\
&=& A_{,jki}\tensor[^i]{Q}{_{\alpha}}\delta\xi^{\alpha}  \nonumber\\ 
&=& (-1)^{ij+ik}A_{,ijk}\tensor[^i]{Q}{_{\alpha}}\delta\xi^{\alpha}  \nonumber\\
&=& (-1)^{j\alpha+k\alpha}(-1)^{ij+ik+j\alpha+k\alpha}A_{,ijk}\tensor[^i]{Q}{_\alpha}\delta\xi^{\alpha}  \nonumber\\
&=& (-1)^{j\alpha+k\alpha}[(A_{,i} \tensor[^i]{Q}{_{\alpha}})_{,jk}-(-1)^{jk+j\alpha+ij}A_{,ik} \tensor[^i]{Q}{_{\alpha}_{,j}}  \nonumber\\ &&-(-1)^{ij+j\alpha}A_{,ij} \tensor[^i]{Q}{_{\alpha}_{,k}} ]\delta\xi^{\alpha}  \nonumber\\
&=& -(-1)^{j\alpha+k\alpha}[(-1)^{jk+j\alpha+ij}A_{,ik} \tensor[^i]{Q}{_{\alpha}_{,j}}  \nonumber\\ &&+(-1)^{ij+j\alpha}A_{,ij} \tensor[^i]{Q}{_{\alpha}_{,k}}]\delta\xi^{\alpha}  \nonumber\\
&=& -(-1)^{k\alpha+kj+ij}A_{,ik} \tensor[^i]{Q}{_{\alpha}_{,j}}\delta\xi^{\alpha}-(-1)^{k\alpha}A_{,ji} \tensor[^i]{Q}{_{\alpha}_{,k}}\delta\xi^{\alpha};
\end{eqnarray}
analogous equations hold for higher order derivatives.

The above equations show that the functional derivatives of absolute invariants transform according to direct products of the codefining representation. Equation \eqref{case_i1} may itself be regarded as a transformation law:
\begin{eqnarray}
\delta \tensor[^i]{Q}{_{\alpha}}&\equiv & \tensor[^i]{Q}{_{\alpha}_{,j}}\; \tensor[^j]{Q}{_{\beta}}\;\delta\xi^{\beta}  \nonumber\\
&=&(\tensor[^i]{[Q_{\alpha},Q_{\beta}]}{} + (-1)^{\alpha\beta}\;\tensor[^i]{Q}{_{\alpha}_{,j}}\;\tensor[^j]{Q}{_{\beta}})\delta\xi^{\beta} \nonumber \\
&=&(\tensor[^i]{Q}{_{\gamma}}\;\tensor{c}{^\gamma _\alpha _\beta}+(-1)^{\alpha\beta}\;\tensor[^i]{Q}{_{\beta}_{,j}}\;\tensor[^j]{Q}{_{\alpha}})\delta\xi^{\beta}  \nonumber\\
&=&(-\tensor[^i]{Q}{_{\gamma}}\;\tensor{c}{^\gamma _\beta _\alpha}+(-1)^{\alpha\beta}\;\tensor[^i]{Q}{_{\beta}_{,j}}\;\tensor[^j]{Q}{_{\alpha}})\delta\xi^{\beta},
\end{eqnarray}
which says that $\tensor[^i]{Q}{_{\alpha}}$ transforms according to the direct product of the defining representation and the coadjoint representation.

Hence, when the group realization is linear, quite generally, field indices (Latin) and group indices (Greek) signal respectively the defining representation and the adjoint representation when they are in the upper position and the contragradient representations when they are in the lower position. 

One may then wonder whether the realization can be always made linear for Type-I theories: for Yang-Mills theories and General Relativity, this is possible, as has been shown in the previous sections, but the answer is not known in general. However, certain results in the theory of finite-dimensional compact Lie groups are suggestive in this connection. Palais \cite{palais1957imbedding} and Mostow \cite{mostow1957equivariant} showed that if a manifold is acted on by a compact Lie group with finitely many orbit types, then it can be embedded into some finite-dimensional linear, homogeneous, orthogonal representation. Moreover, results of this kind can usually be extended to the case of finite-dimensional semisimple Lie groups whether compact or not. 

If similar results could be extended to \textit{field realizations of gauge groups} (which are infinite-dimensional), then one could simply add enough extra fields to Type-I systems to make the realization linear. The extra fields could be made dynamically innocuous by inclusion of appropriate Lagrange-multiplier fields in the action. There is one difference between the finite-dimensional and field theoretical cases that apparently cannot be eliminated: in the case of fields the variables $\phi^i$ cannot always be chosen in such a way as to yield a realization that is simultaneously linear and homogeneous. In the case of the Yang-Mills field the infinitesimal gauge transformation law \eqref{YM_inf_gauge} includes an inhomogeneous term that cannot be removed by any choice of variables. In the case of the Maxwell field the inhomogeneous term is all there is.  
\subsection{Equation of small disturbances}
Let $\phi^i$ and $\phi^i+\delta\phi^i$ be two neighboring solutions of the dynamical equations \eqref{dynam_eq}:
\begin{eqnarray}
0&=&\tensor[_{i,}]{S}{}[\phi], \\
0&=&\tensor[_{i,}]{S}{}[\phi+\delta\phi]=\tensor[_{i,}]{S}{}[\phi]+\tensor[_{i,}]{S}{_{,j}}[\phi]\delta\phi^j + ... ,
\end{eqnarray}
where the dots stand for terms which are at least quadratic in $\delta \phi^i$.

Evidently, to first order in $\delta \phi^i$, we have:
\begin{equation}
\tensor[_{i,}]{S}{_{,j}}[\phi]\delta\phi^j =0.
\end{equation}
This is called \textit{homogeneous equation of small disturbances}. Its solutions are known as \textit{Jacobi fields} relative to the on-shell field $\phi$. In the following equations the argument $\phi$ will often be suppressed. In practice a small disturbance is produced by a weak external agent, which may be described by a small change in the functional form of the action. Let $A$ be a pure real-valued scalar field on $\Phi$ and $\epsilon$ be an infinitesimal real $c$-number or imaginary $a$-number according as $A$ is $c$-type or $a$-type; therefore $A\epsilon$ is a real valued $c$-type scalar field, and the following change in the functional form of the action is admissible:
\begin{equation}
S[\phi] \mapsto S[\phi] + A[\phi]\epsilon \label{action_var}
\end{equation}
Let $\delta\phi^i$ be a solution of 
\begin{equation}
\tensor[_{i,}]{S}{_{,j}}\delta\phi^j = - \tensor[_{i,}]{A \epsilon}{}. \label{inhom_sd}
\end{equation}
It is easy to see that, neglecting higher order terms, $\phi^i+\delta\phi^i$ satisfies the dynamical equations of the system $S + A\epsilon$ if and only if $\phi^i$ satisfies those of the system $S$:
\begin{eqnarray}
\tensor[_{i,}]{(S+A\epsilon)}{}[\phi+\delta\phi]&=& \tensor[_{i,}]{S}{}[\phi+\delta\phi]+\tensor[_{i,}]{A}{}[\phi+\delta\phi]\epsilon   \nonumber\\
&=&\tensor[_{i,}]{S}{}[\phi]+\tensor[_{i,}]{S}{_{,j}}[\phi] \; \delta\phi^j+\tensor[_{i,}]{A}{}[\phi]\;\epsilon +\tensor[_{i,}]{A}{_{,j}}[\phi]\;\delta\phi^j\;\epsilon + ...  \nonumber\\
&=&\tensor[_{i,}]{S}{}+(\tensor[_{i,}]{S}{_{,j}}\delta\phi^j+\tensor[_{i,}]{A}{}\epsilon) + ... .
\end{eqnarray} 
Equation \eqref{inhom_sd} is called \textit{inhomogeneous equation of small disturbances}. Its general solution is obtained by adding to a particular solution an arbitrary Jacobi field.

When the dynamical equation \eqref{dynam_eq} is satisfied, eq. \eqref{Ward_id} reads:
\begin{equation}
\tensor[_\alpha]{Q}{^{\sim}^i}\tensor[_{i,}]{S}{_{,j}}=0 \; \; \; (\tensor{S}{_{,j}}=0) \label{Ward_id_shell}
\end{equation}
When applied to \eqref{inhom_sd}, by virtue of the arbitrariness of $\epsilon$, this equation implies 
\begin{equation}
\tensor[_\alpha]{Q}{^{\sim}^i}\tensor[_{i,}]{A}{}=0 \; \; \; (\tensor{S}{_{,j}}=0).
\end{equation}
Therefore eq. \eqref{inhom_sd} is seen to be inconsistent unless $A$ is a conditional invariant, i.e., a physical observable. If this is not the case, then the solution of $\tensor[_{i,}]{(}{}S+A\epsilon)=0$ cannot differ from $\phi$ by infinitesimal amounts. Evidently small changes in the action functional will produce small changes in the on-shell dynamical variables only if they leave intact the flow invariances of the theory. 
\subsection{Supplementary conditions}
Equation \eqref{Ward_id_shell} implies that, on the dynamical shell
\begin{equation}
\tensor[_{i,}]{S}{_{,j}}\tensor[^j]{Q}{_{\alpha}}\delta\xi^{\alpha}=0 \; \; \; (\tensor{S}{_{,j}}=0)
\end{equation}
for every $\delta\xi^{\alpha}$ of compact support in space-time: this implies that $\tensor[_1]{S}{_1}$ is not an invertible operator; hence, it has no Green's functions. 

When $\tensor[^i]{Q}{_{\alpha}}\delta\xi^{\alpha}$ is added to a solution of eq. \eqref{inhom_sd}, the result is another solution: however, they are physically identical, since they differ merely by an invariance transformation \eqref{gauge_field}. It is convenient to remove this redundancy by imposing a differential \textit{supplementary condition} on the small disturbances $\delta\phi^i$, of the form\footnote{Position and nature of the indices of the auxiliary distributions introduced throughout this work is not accidental: when dealing with Type-I theories with linear gauge group realization, they show how these distributions must transform under gauge transformations. }
\begin{equation}
\tensor[_\alpha]{P}{_i}\delta\phi^{i}=0. \label{supp_cond} 
\end{equation}
The supplementary condition is effective as long as the operator
\begin{equation}
\tensor[_\alpha]{\mathcal{F}}{_\beta}\equiv \tensor[_\alpha]{P}{_j}\tensor[^j]{Q}{_{\beta}} \label{Fcal}
\end{equation}
is nonsingular and has Green's functions; in fact, given a Jacobi field $\delta\phi^{i}$, all the physically identical solutions can be written as:
\begin{equation}
\delta\phi^{i}+\tensor[^i]{Q}{_{\alpha}}\delta\xi^{\alpha};
\end{equation}
by imposing the supplementary condition \eqref{supp_cond}, one obtains
\begin{eqnarray}
\tensor[_\alpha]{P}{_i}(\delta\phi^{i}+\tensor[^i]{Q}{_{\beta}}\delta\xi^{\beta})&=&0  \nonumber\\
\tensor[_\alpha]{P}{_i}\delta\phi^{i}+\tensor[_\alpha]{P}{_i}\tensor[^i]{Q}{_{\beta}}\delta\xi^{\beta}&=&0  \nonumber\\
\tensor[_\alpha]{P}{_i}\delta\phi^{i}+\tensor[_\alpha]{\mathcal{F}}{_\beta}\delta\xi^{\beta}&=&0.
\end{eqnarray}
Being $\tensor[_\alpha]{\mathcal{F}}{_\beta}$ invertible, this equation determines $\delta\xi^{\alpha}$ and, therefore, the solution $\delta\phi^{i}+\tensor[^i]{Q}{_{\alpha}}\delta\xi^{\alpha}$.

Let $\eta$ be a local, continuous, nonsingular, supersymmetric matrix whose elements are:
\begin{equation}
\eta^{\alpha\beta}=(-1)^{\alpha+\beta}\eta^{\beta\alpha};
\end{equation} 
introduce the following differental operator:
\begin{equation}
\tensor[_i]{F}{_j}\equiv \tensor[_{i,}]{S}{_{,j}} + \tensor[_i]{P}{^{\sim}_\alpha}\eta^{\alpha \beta}\tensor[_\beta]{P}{_j}, \label{F_op}
\end{equation}
where $\tensor[_i]{P}{^{\sim}_\alpha}$ is the supertranspose of $\tensor[_\alpha]{P}{_i}$:
\begin{equation}
\tensor[_i]{P}{^{\sim}_\alpha}=(-1)^{i+\alpha+i\alpha}\;\tensor[_\alpha]{P}{_i}.
\end{equation}
It is easy to see that $\tensor[_i]{F}{_j}$ has the same supersymmetry properties as $\tensor[_{i,}]{S}{_{,j}}$, i.e. it is supersymmetric:
\begin{equation}
\tensor[_i]{F}{_j}=(-1)^{i+j+ij}\tensor[_j]{F}{_i}. \label{supersymm_F}
\end{equation}
From now on, we will assume that the kernel of the linear differential operator $\tensor[_{i,}]{S}{_{,j}}$ consists of the fields of the form $\tensor[^i]{Q}{_{\alpha}}\delta\xi^{\alpha}$, with $\delta\xi^{\alpha}$ of compact support in space-time: this is true in all the practical cases. Hence, the folloing holds:
\begin{theorem}
	If $\tensor[_\alpha]{\mathcal{F}}{_\beta}$ is nonsingular, then $\tensor[_i]{F}{_j}$ is nonsingular too.
\end{theorem}
\begin{proof}
	Suppose that there is a set of functions $X^{j}$ of compact support such that
	\begin{equation}
	\tensor[_i]{F}{_j}X^{j}=0. \label{first_eq}
	\end{equation}
	By suppressing all indices and recalling the definition \eqref{F_op}, one obtains:
	\begin{equation}
	(\tensor[_{1}]{S}{_{1}}+\tensor{P}{^{\sim}}\eta\tensor[]{P}{})X=0.
	\end{equation}
	By taking the supertranspose and recalling the supersymmetry properties stated above, one can write:
	\begin{equation}
	X^{\sim}(\tensor[_{1}]{S}{_{1}}+\tensor{P}{^{\sim}}\eta\tensor[]{P}{})=0.
	\end{equation}
	By applying $Q_{\alpha}$ from the right and using \eqref{Ward_id_shell}, which holds on shell, one arrives at:
	\begin{eqnarray}
	X^{\sim}(\tensor[_{1}]{S}{_{1}}+\tensor{P}{^{\sim}}\eta\tensor[]{P}{})Q_{\alpha}&=&0, \\
	X^{\sim}\tensor{P}{^{\sim}}\eta\tensor[]{P}{}Q_{\alpha}&=&0.
	\end{eqnarray}
	By restoring indices and recalling the definition \eqref{Fcal}, one can write:
	\begin{equation}
	X^{\sim j}\tensor[_j]{P}{^{\sim}_\alpha}\eta^{\alpha\beta}\tensor[_\beta]{\mathcal{F}}{_\gamma}=0.
	\end{equation}
	Being $\eta^{ \alpha\beta}$, $\tensor[_\beta]{\mathcal{F}}{_\gamma}$ nonsingular, the previous equation implies
	\begin{equation}
	X^{\sim j}\tensor[_j]{P}{^{\sim}_\alpha}=0,
	\end{equation}
	or equivalently
	\begin{eqnarray}
	\tensor[_\alpha]{P}{_i}X^i&=&0, \\
	PX&=&0. \label{sec_eq}
	\end{eqnarray}
	Hence, from the first equation \eqref{first_eq}, it follows 
	\begin{eqnarray}
	FX&=&(\tensor[_{1}]{S}{_{1}}+\tensor{P}{^{\sim}}\eta\tensor[]{P}{})X  \nonumber\\
	&=&\tensor[_{1}]{S}{_{1}}X  \nonumber \\
	&=&0.
	\end{eqnarray}
	But the kernel of $\tensor[_1]{S}{_1}$ consists of the fields of the form $\tensor[^i]{Q}{_{\alpha}}\delta\xi^{\alpha}$, with $\delta\xi^{\alpha}$ of compact support in space-time, then $X$ must be of the form $Q\xi$; therefore one can write again eq. \eqref{sec_eq} and recall the definition \eqref{Fcal}:
	\begin{eqnarray}
	PX&=&0,  \nonumber\\
	PQ\xi&=&0,  \nonumber\\
	\mathcal{F}\xi&=&0.
	\end{eqnarray}
	Since $\mathcal{F}$ is non singular, this equation implies $\xi=0$ and, as a consequence, $X=0$.
	Therefore the kernel of the operator $\tensor[_i]{F}{_j}$ consists of the null field $X^{i}=0$ only: then $\tensor[_i]{F}{_j}$ is nonsingular.
	
	This ends the proof.
\end{proof}
\subsection{Retarded and advanced Green's function of $F$}
When the supplementary condition \eqref{supp_cond} is satisfied, eq. \eqref{inhom_sd} may be replaced by 
\begin{equation}
\tensor[_{i}]{F}{_{j}}\;\delta\phi^j = - \tensor[_{i,}]{A \epsilon}{}. \label{inhom_sd_sc}
\end{equation}
In fact, if both \eqref{supp_cond} and \eqref{inhom_sd} are satisfied, it is obvious that \eqref{inhom_sd_sc} is satisfied too; on the other hand, if \eqref{inhom_sd_sc} is satisfied, by applying $\tensor[_\alpha]{Q}{^\sim^i}$ from the left one obtains:
\begin{eqnarray}
\tensor[_\alpha]{Q}{^\sim^i}\;\tensor[_{i}]{F}{_{j}}\;\delta\phi^j&=&\tensor[_\alpha]{Q}{^\sim^i}(- \tensor[_{i,}]{A \epsilon}{})  \nonumber\\
\tensor[_\alpha]{Q}{^\sim^i}(\tensor[_{i,}]{S}{_{,j}} + \tensor[_i]{P}{^{\sim}_\alpha}\;\eta^{\alpha \beta}\;\tensor[_\beta]{P}{_j})\delta\phi^j&=&\tensor[_\alpha]{Q}{^\sim^i}(- \tensor[_{i,}]{A \epsilon}{})  \nonumber\\
\tensor[_\alpha]{Q}{^\sim^i}\;\tensor[_i]{P}{^{\sim}_\alpha}\;\eta^{\alpha \beta}\;\tensor[_\beta]{P}{_j}\;\delta\phi^j&=&0  \nonumber\\
\tensor[_\alpha]{\mathcal{F}}{^\sim_\beta}\;\eta^{\alpha \beta}\;\tensor[_\beta]{P}{_j}\;\delta\phi^j&=&0,
\end{eqnarray}
where \eqref{Ward_id_shell} have been used, plus the fact that $A$ is a conditional invariant; Being $\eta^{\alpha\beta}$, $\tensor[_\beta]{\mathcal{F}}{_\gamma}$ nonsingular, the previous equation implies 
\begin{equation}
\tensor[_\beta]{P}{_j}\;\delta\phi^j=0,
\end{equation}
i.e., the supplementary condition is satisfied. Hence, the following holds:
\begin{eqnarray}
\tensor[_{i}]{F}{_{j}}\;\delta\phi^j&=&- \tensor[_{i,}]{A \epsilon}{}  \nonumber\\
\tensor[_{i,}]{S}{_{,j}}\;\delta\phi^j + \tensor[_i]{P}{^{\sim}_\alpha}\;\eta^{\alpha \beta}\;\tensor[_\beta]{P}{_j}\;\delta\phi^j&=&- \tensor[_{i,}]{A \epsilon}{}  \nonumber\\
\tensor[_{i,}]{S}{_{,j}}\;\delta\phi^j&=&- \tensor[_{i,}]{A \epsilon}{}
\end{eqnarray}
i.e., the inhomogeneous equation of small disturbances is satisfied too.

Since $F$ is a nonsingular operator this equation has unique solutions for given boundary conditions. These solutions can be expressed in terms of Green's 
functions. 

We shall consider in this section only retarded and advanced boundary conditions. Denote by $\delta^{-}\phi^i$ and $\delta^{+}\phi^i$ respectively the corresponding solutions. Then 
\begin{equation}
\delta^{\pm}\phi^i=G^{\pm ij}\tensor[_{j,}]{A \epsilon}{} \label{jac_Gr_sol}
\end{equation}
where $G^{- ij}$ and $G^{+ ij}$ are the retarded and advanced Green's functions of $\tensor[_{i}]{F}{_{j}}$, respectively :
\begin{eqnarray}
\tensor[_{i}]{F}{_{k}}G^{\pm kj}&=& -\tensor[_i]{\delta}{^j}, \label{def_Gr}\\ 
G^{- ij}=0 \; \; \; \text{if} \; \; \; i<j, \; \; \; \;&& \; \; \; \; G^{+ ij}=0 \; \; \;\text{if} \; \; \; i>j, \label{ret_adv_Gr}
\end{eqnarray}
where \lq\lq $i<j$'' means \lq\lq the time associated with the index $i$ lies in the past of the time associated with the index $j$'' and \lq\lq$i>j$'' means \lq\lq the time associated with the index $i$ lies in the future of the time associated with the index $j$''. Consequently the kinematical conditions \eqref{ret_adv_Gr} imply that $ G^{- ij} (G^{+ ij} )$ is nonvanishing only when the space-time point associated with $i$ lies on or inside the future (past) light cone emanating from the space-time point associated with $j$. 

A minus sign appears on the right of eq. \eqref{def_Gr} for historical reasons, and the symbol $\tensor[_i]{\delta}{^j}$ represents a combined Kronecker delta $\delta$-distribution. In the supercondensed notation \eqref{def_Gr} is written 
\begin{equation}
FG^{\pm}=-1. \label{def_Gr_superc}
\end{equation}
It should be remarked that the summation-integration involved on the right side of eq. \eqref{jac_Gr_sol} will generally not converge unless the functional form of $A$ is such that the functions $\tensor[_{j,}]{A}{}$ do not increase in magnitude too rapidly toward the past or future. A sufficient condition for convergence, of course, is that supp $\tensor[_{j,}]{A}{}$ be limited in time, where \lq\lq supp $\tensor[_{j,}]{A}{}$'' denotes the union of the supports of all the $\tensor[_{j,}]{A}{}$. If \eqref{jac_Gr_sol} does not converge, then the solutions of $\tensor[_{i,}]{S}{_{,j}}\delta\phi^j + \tensor[_{i,}]{A \epsilon}{}=0$ do not lie close (in $\Phi$) to those of $\tensor[_{i,}]{S}{}= 0$ no matter how small $\epsilon$ may be chosen. 
\subsection{Equality of left and right Green's functions}
In eqs. \eqref{def_Gr}, \eqref{def_Gr_superc} the $G^{\pm}$ appear as \textit{right} Green's functions. They are also \textit{left} Green's functions. To prove this let us temporarily distinguish left Green's functions from right Green's functions by employing subscripts $L$ and $R$. Let $X^i$ be arbitrary functions of compact support in space-time and let 
\begin{equation}
Y^i \equiv (\tensor*[]{G}{_R^-^i^j}-\tensor*[]{G}{_L^-^i^j})\tensor[_j]{F}{_k}X^k.
\end{equation}
Since the $X^k$ have compact support it does not matter whether the $j$ summation-integration or the $k$ summation-integration is performed first in this expression. That is, $\tensor[_j]{F}{_k}$ may be regarded as acting either to the right or to the left. By performing the $j$ summation-integration first and using $\tensor*[]{G}{_L^-^i^j}\tensor[_j]{F}{_k}= -\tensor[^i]{\delta}{_k}$, one obtains
\begin{eqnarray}
Y^i&=&\tensor*[]{G}{_R^-^i^j} \; \tensor[_j]{F}{_k} \; X^k+\tensor[^i]{\delta}{_k} \; X^k \nonumber\\
&=&\tensor*[]{G}{_R^-^i^j} \; \tensor[_j]{F}{_k} \; X^k+X^i.
\end{eqnarray}
By applying $\tensor[_m]{F}{_i}$ from the the left, performing the $i$ summation-integration first and using $\tensor[_m]{F}{_i}\tensor*[]{G}{_R^-^i^j}= -\tensor[_m]{\delta}{^j}$, one obtains:
\begin{eqnarray}
\tensor[_m]{F}{_i}Y^i&=&\tensor[_m]{F}{_i}\tensor*[]{G}{_R^-^i^j}\tensor[_j]{F}{_k}X^k+\tensor[_m]{F}{_i}X^i  \nonumber\\
\tensor[_m]{F}{_i}Y^i&=&-\tensor[_m]{\delta}{^j}\tensor[_j]{F}{_k}X^k+\tensor[_m]{F}{_i}X^i  \nonumber \\
\tensor[_m]{F}{_i}Y^i&=&-\tensor[_m]{F}{_k}X^k+\tensor[_m]{F}{_i}X^i  \nonumber \\
\tensor[_m]{F}{_i}Y^i&=&0.
\end{eqnarray}
But $Y^i$ vanishes if $i<{\rm supp}\; X \equiv \cup_{j} {\rm supp} \; X^j$. This means that the boundary data for the above equation vanish to the past of ${\rm supp}\; X$ , and hence $Y^i$ must vanish \textit{everywhere}. Since the $X^i$ are arbitrary it follows that 
\begin{eqnarray}
0&=&  (\tensor*[]{G}{_R^-^i^j}-\tensor*[]{G}{_L^-^i^j})\tensor[_j]{F}{_k} \nonumber \\
&=& \tensor*[]{G}{_R^-^i^j}\;\tensor[_j]{F}{_k}+\tensor[^i]{\delta}{_k},  \nonumber \\
\tensor*[]{G}{_R^-^i^j}\;\tensor[_j]{F}{_k}&=&-\tensor[^i]{\delta}{_k}.
\end{eqnarray}
But this is just the condition that $\tensor*[]{G}{_R^-^i^j}$ be a left Green's function. Therefore $\tensor*[]{G}{_R^-^i^j}=\tensor*[]{G}{_L^-^i^j}$. In a similar manner one may show that $\tensor*[]{G}{_R^+^i^j}=\tensor*[]{G}{_L^+^i^j}$. Thus eqs. \eqref{def_Gr}, \eqref{def_Gr_superc} imply
\begin{equation}
G^{\pm ik} \; \tensor[_{k}]{F}{_{j}}= -\tensor[^i]{\delta}{_j}
\end{equation}
or, in supercondensed notation
\begin{equation}
G^{\pm}\overleftarrow{F}= -1.
\end{equation}
It is important to stress that this proof holds regardless of the symmetry of $F$.
\subsection{Reciprocity relations}
The actual supersymmetry of F gives rise to simple relations between the retarded and advanced Green's functions. Consider the expression 
\begin{equation}
(-1)^{ki}\tensor*[]{G}{^-^k^i}\tensor[_{k}]{F}{_{l}}\tensor*[]{G}{^+^l^j}.
\end{equation}
Because of the kinematical conditions \eqref{ret_adv_Gr}, the intersection of the supports of $\tensor*[]{G}{^-^k^i}$ and $\tensor*[]{G}{^+^l^j}$, with $i$ and $j$ held fixed, is compact, since it is the intersection of a forward light cone with a backward light cone. Therefore it makes no difference whether the $k$ summation-integration or the $l$ summation-integration is performed first. Using this fact, together with the supersymmetry law \eqref{supersymm_F}, one obtains:
\begin{eqnarray}
0&=&(-1)^{ki}\tensor*[]{G}{^-^k^i}(\tensor[_{k}]{F}{_{l}}-(-1)^{k+l+kl}\tensor[_{l}]{F}{_{k}})\tensor*[]{G}{^+^l^j} \nonumber\\
&=&-(-1)^{ki}\tensor*[]{G}{^-^k^i}\tensor[_{k}]{\delta}{^{j}}-(-1)^{ki+k+l+kl+il+ik+kl+k}\tensor[_{l}]{F}{_{k}}\tensor*[]{G}{^-^k^i}\tensor*[]{G}{^+^l^j}\nonumber \\
&=&-(-1)^{ji}\tensor*[]{G}{^-^j^i}+(-1)^{l+il}\tensor[_l]{\delta}{^i}\tensor*[]{G}{^+^l^j}\nonumber \\
&=&-(-1)^{ji}\tensor*[]{G}{^-^j^i}+\tensor*[]{G}{^+^i^j}.
\end{eqnarray}
Therefore
\begin{equation}
(-1)^{ji}\tensor*[]{G}{^-^j^i}=\tensor*[]{G}{^+^i^j},
\end{equation}
or, equivalently
\begin{equation}
\tensor*[]{G}{^\pm^i^j}=(-1)^{ij}\tensor*[]{G}{^\mp^j^i}. \label{recipr_rel}
\end{equation}
Equations \eqref{recipr_rel} are called \textit{reciprocity relations for the Green's functions}. In the supercondensed notations, they take the form:
\begin{equation}
G^\pm = G^{ \mp\sim}.
\end{equation}
\subsection{R\-e\-l\-a\-t\-i\-o\-n b\-e\-t\-w\-e\-e\-n G\-r\-e\-e\-n's fu\-nc\-ti\-on\-s of $\mathcal{F}$ a\-nd $F$}
In this section an important relation between the Green's Functions of $\mathcal{F}$ and $F$ will be derived; attention will be confined, for now, to the retarded and advanced Green's functions, those of $\mathcal{F}$ being denoted by $\mathcal{G}^-$ and $\mathcal{G}^+$, respectively: 
\begin{equation}
\tensor[_{_\alpha}]{\mathcal{F}}{_{\gamma}}\mathcal{G}^{\pm \gamma\beta}=-\tensor[_\alpha]{\delta}{^\beta}.
\end{equation}
Also in this case, $\mathcal{G}^\pm$ are both right and left Green's functions, and obey similar reciprocity relations to the ones shown for $G$. 

By writing down the definition of $F$ \eqref{F_op} and using \eqref{Ward_id_shell}, that is valid on shell, one obtains:
\begin{eqnarray}
\tensor[_j]{F}{_k}\tensor[^k]{Q}{_\beta}&=& \tensor[_j]{S}{_k}\tensor[^k]{Q}{_\beta}+\tensor[_j]{P}{^\sim_\alpha}\eta^{\alpha\gamma}\tensor[_\gamma]{P}{_k}\tensor[^k]{Q}{_\beta} \nonumber\\
&=& 0 + \tensor[_j]{P}{^\sim_\alpha}\eta^{\alpha\gamma}\tensor[_\gamma]{\mathcal{F}}{_\beta}.
\end{eqnarray}
Multiplying this equation on the left by $G^\pm$ and on the right by $\mathcal{G}^\pm$, and noting that the intersection of the supports of these extra factors (with the outer suppressed indices held fixed) is compact so that $F$ and $\mathcal{F}$ may act in either direction, one gets 
\begin{eqnarray}
G^{\pm ij}\tensor[_j]{F}{_k}\tensor[^k]{Q}{_\beta}\mathcal{G}^{\pm\beta\theta}&=&G^{\pm ij}\tensor[_j]{P}{^\sim_\alpha}\eta^{\alpha\gamma}\tensor[_\gamma]{\mathcal{F}}{_\beta}\mathcal{G}^{\pm\beta\theta}, \nonumber\\
\tensor[^i]{\delta}{_k}\tensor[^k]{Q}{_\beta}\mathcal{G}^{\pm\beta\theta}&=&G^{\pm ij}\tensor[_j]{P}{^\sim_\alpha}\eta^{\alpha\gamma}\tensor[_\gamma]{\delta}{^\theta},\nonumber \\
\tensor[^i]{Q}{_\beta}\mathcal{G}^{\pm\beta\theta}&=&G^{\pm ij}\tensor[_j]{P}{^\sim_\alpha}\eta^{\alpha\theta},
\end{eqnarray} 
or, equivalently:
\begin{equation}
Q\mathcal{G}^\pm=G^\pm P^\sim \eta. \label{F_Fcal_Green}
\end{equation}
Now, if $\eta$ is chosen to be \textit{ultralocal}, i.e., in $\eta$ no undifferentiated $\delta$ distributions appear, then its negative inverse $\lambda$ is unique and supersymmetric; on the other hand, if $\eta$ is not ultralocal, then its negative inverse is not unique: they are Green's functions, and they will be assumed to obey the same kinematical relations as $G^\pm$, $\mathcal{G}^\pm$; its elements will be indicated as
\begin{equation}
\tensor[_\alpha]{\lambda}{_\beta}=(-1)^{\alpha+\beta+\alpha\beta}\tensor[_\beta]{\lambda}{_\alpha}. 
\end{equation}
Hence eq. \eqref{F_Fcal_Green} may be written
\begin{equation}
-Q\mathcal{G}^\pm \lambda = G^\pm P^\sim, \label{F_Fcal_Green2}
\end{equation}
or, taking the supertranspose and using \eqref{recipr_rel} and the symmetry properties:
\begin{equation}
-\lambda \mathcal{G}^{\mp\sim} Q^\sim = P G^\pm. \label{F_Fcal_Green3}
\end{equation}
Given this equation, one can prove in another way that \eqref{jac_Gr_sol} is the solution for the inhomogeneous equation of small disturbances which obeys the supplementary conditions:
\begin{eqnarray}
\tensor[_\alpha]{P}{_i}\delta^\pm \phi^{i} &=& \tensor[_\alpha]{P}{_i}G^{\pm ij}\tensor[_{j,}]{A \epsilon}{} \nonumber\\
&=&-\tensor[_\alpha]{\lambda}{_\beta}\mathcal{G}^{\pm \beta\gamma} \tensor[_\gamma]{Q}{^\sim ^j}\tensor[_{j,}]{A \epsilon}{} \nonumber\\
&=&0, \label{suppl_cond_alt}
\end{eqnarray}
where the last line follows from the fact that $A$ is a physical observable.
\subsection{Landau Green's functions}
Given $\phi \in \Phi_{0}$, the entire tangent space at $\phi$, $T_{\phi}\Phi$, is spanned by vectors of the form $G^{\pm ij}\tensor[_{j,}]{D}{}$, with $D$ of compact support. Equation \eqref{suppl_cond_alt} shows that, in order to get the subspace of $T_{\phi}\Phi$ obeying the supplementary conditions, only $D$ which are physical observables have to be considered.

Another possible choice to obtain the same result without imposing conditions on $D$ is to \lq\lq modify'' the Green's function: the task is complete if one finds an object $B^{\pm ij}$ such that
\begin{equation}
\begin{cases}
&P(G^{\pm}+B^{\pm})=0. \\
&\tensor[_1]{S}{_1}B^{\pm}=0
\end{cases}
\end{equation}
But, by using \eqref{F_Fcal_Green3}, the first equation reads
\begin{equation}
PB^{\pm}=\lambda \mathcal{G}^{\mp\sim} Q^\sim.
\end{equation}
Using $PQ\mathcal{G}=\mathcal{FG}=-1$ and \eqref{Ward_id_shell}, it is easy to see that 
\begin{equation}
B^{\pm} = -Q \mathcal{G}^{\pm}\lambda\mathcal{G}^{\mp\sim}Q^{\sim}
\end{equation}
is a solution which satisfies the second equation of the system, too.
Therefore one is led to define \textit{Landau Green's functions} $G_{\infty}^{\pm}$
\begin{eqnarray}
G_{\infty}^{\pm} &\equiv & G^{\pm}+B^{\pm} \nonumber\\
&=&G^{\pm}- Q \mathcal{G}^{\pm}\lambda\mathcal{G}^{\mp\sim}Q^{\sim},
\end{eqnarray}
or, with restored indices:
\begin{eqnarray}
G_{\infty}^{\pm ij} &\equiv & G^{\pm ij}+B^{\pm ij} \nonumber\\
&=& G^{\pm ij} - \tensor[^i]{Q}{_\alpha}\mathcal{G}^{\pm \alpha\beta}\tensor[_\beta]{\lambda}{_\gamma}\mathcal{G}^{\mp\sim\gamma\delta}\tensor[_\delta]{Q}{^\sim ^j}.
\end{eqnarray}
The Landau Green's functions are defined only on shell and, when applied to a physical observable, $G^{\pm}$ and $G_{\infty}^{\pm}$ yield trivially the same results.

It is important to observe that $G_{\infty}^{\pm}$ is no longer a negative inverse for $F$; in fact
\begin{eqnarray}
G_{\infty}^{\pm}F &=& (G^{\pm}- Q \mathcal{G}^{\pm}\lambda\mathcal{G}^{\mp\sim}Q^{\sim})F \nonumber\\
&=& -1 - Q \mathcal{G}^{\pm}\lambda\mathcal{G}^{\mp\sim}Q^{\sim}F \nonumber\\
&=& -1 - Q \mathcal{G}^{\pm}\lambda\mathcal{G}^{\mp\sim}Q^{\sim}(\tensor[_1]{S}{_1}+P^\sim \eta P)\nonumber \\
&=& -1 - Q \mathcal{G}^{\pm}\lambda\mathcal{G}^{\mp\sim}Q^{\sim}P^\sim \eta P \nonumber\\
&=& -1 - Q \mathcal{G}^{\pm}\lambda\mathcal{G}^{\mp\sim} \mathcal{F}^\sim \eta P\nonumber \\
&=& -1 + Q \mathcal{G}^{\pm}\lambda \eta P \nonumber\\
&=& -1 - Q \mathcal{G}^{\pm} P.
\end{eqnarray}
Call this operator $\Pi^\pm \equiv -G_{\infty}^{\pm}F$; it can be easily seen that it is a projection operator whose kernel is the subspace of $T_{\phi}\Phi$ that is tangent to the invariance flows $Q_{\alpha}$:
\begin{eqnarray}
\Pi^{\pm 2} &=&\Pi^\pm \Pi^\pm \nonumber\\
&=& (1 + Q \mathcal{G}^{\pm} P) (1 + Q \mathcal{G}^{\pm} P)\nonumber \\
&=& 1 + 2 Q \mathcal{G}^{\pm} P + Q \mathcal{G}^{\pm} PQ \mathcal{G}^{\pm} P\nonumber \\
&=& 1 + 2 Q \mathcal{G}^{\pm} P + Q \mathcal{G}^{\pm} \mathcal{F} \mathcal{G}^{\pm} P \nonumber\\
&=& 1 + 2 Q \mathcal{G}^{\pm} P - Q \mathcal{G}^{\pm} P \nonumber\\
&=& 1 + Q \mathcal{G}^{\pm} P \nonumber\\ 
&=& \Pi^{\pm}, \\
& & \nonumber \\
\Pi^\pm Q &=& (1 + Q \mathcal{G}^{\pm} P)Q \nonumber\\
&=&Q+ Q \mathcal{G}^{\pm} PQ \nonumber\\
&=&Q+ Q \mathcal{G}^{\pm} \mathcal{F} \nonumber\\
&=&Q- Q \nonumber \\
&=&0,
\end{eqnarray}
and, obviously
\begin{eqnarray}
P\Pi^{\pm}&=& P(1 + Q \mathcal{G}^{\pm} P) \nonumber\\
&=&P + PQ \mathcal{G}^{\pm} P \nonumber\\
&=&P + \mathcal{F} \mathcal{G}^{\pm} P \nonumber\\
&=&P-P \nonumber\\
&=&0.
\end{eqnarray}
Finally, noticing that $\Pi^\pm \equiv -G_{\infty}^{\pm}F= -G_{\infty}^{\pm}\tensor[_1]{S}{_1} =-G^{\pm}\tensor[_1]{S}{_1}$, and recalling that the restriction of a projection operator on its image is the identity operator, it can be said that the Landau Green's functions are the negative inverses of the restriction on $Ran(\Pi^\pm)$ of the singular operator $\tensor[_1]{S}{_1}$.
\subsection{Disturbances in physical observables}
Let $B$ be a physical observable; call $\delta^{\pm}B$ the changes in value of $B$ under the disturbance \eqref{jac_Gr_sol} caused by the change in the action functional \eqref{action_var}. Then 
\begin{eqnarray}
\delta^{\pm}B &=& B_{,i}\delta^{\pm}\phi^i  \label{dist_ph_obs} \nonumber\\ 
&=&B_{,i}(G^{\pm ij}\tensor[_{j,}]{A}{}) \epsilon{} \nonumber \\
&=&(-1)^{AB}(A_{,j}G^{\mp ji})\tensor[_{i,}]{B}{}\epsilon,
\end{eqnarray}
in which eq. \eqref{left_right_der} and the reciprocity relations \eqref{recipr_rel} have been used in obtaining the final expression. Parentheses have been inserted because it is not guaranteed that if the summation-integration over $i$ is performed before the summation-integration over $j$ the same result will be obtained. We shall assume that the functions $B_{,i}$ do not increase in magnitude too rapidly for convergence either in the past or in the future. In fact we shall assume that these functions are well enough behaved that the parentheses may be removed. The following are some sufficient (although not necessary) conditions for the parentheses to be absent: 
\begin{itemize}
	\item[1.] With the retarded solution $\delta^{-}\phi$, if $(\cup_{j} {\rm \; supp}\; \tensor[_{j,}]{A}{}) \cap (\cup_{i} {\rm \; supp}\; B_{,i})$ is compact and there exist spacelike hypersurfaces $\Sigma_{+}$,$\Sigma_{-}$ such that $$(\cup_{i} {\rm \; supp}\; B_{,i})<\Sigma_{+} {\rm \; and \; } (\cup_{j} {\rm \; supp}\; \tensor[_{j,}]{A}{})>\Sigma_{-}.$$ In this case $\delta^{-}B$ vanishes unless $\Sigma_{+}>\Sigma_{-}$.
	\item[2.] With the advanced solution $\delta^{+}\phi$, if $(\cup_{j} {\rm \; supp} \;\tensor[_{j,}]{A}{}) \cap (\cup_{i} {\rm \; supp}\; B_{,i})$ is compact and there exist spacelike hypersurfaces $\Sigma_{+}$,$\Sigma_{-}$ such that $$(\cup_{j} {\rm \; supp} \;\tensor[_{j,}]{A}{})<\Sigma_{+} {\rm \;and \;} (\cup_{i} {\rm \; supp}\; B_{,i})>\Sigma_{-}.$$ In this case $\delta^{+}B$ vanishes unless $\Sigma_{+}>\Sigma_{-}$.
	\item[3.] With either solution, if $(\cup_{j} {\rm \; supp}\; \tensor[_{j,}]{A}{})$ and $(\cup_{i} {\rm \; supp}\; B_{,i})$ are both compact.
\end{itemize}
\begin{remark}
	When, as now, we are working with $\phi$ on shell, a question arises regarding the meaning of the expressions ${\rm supp} \;\tensor[_{j,}]{A}{}$ and ${\rm supp}\; B_{,i}$. When $\phi$ is on shell the functional form of a physical observable is defined only \textit{modulo} the dynamical equations, and hence the expressions ${\rm supp}\;\tensor[_{j,}]{A}{}$ and ${\rm supp}\; B_{,i}$ would seem to be ambiguous. The following clarification is necessary: every physical observable has an expression in terms of the fields $\phi^i$ the functional form of which is independent of that of $S$. \textit{This} is the form that is to be understood in the expressions ${\rm supp}\;\tensor[_{j,}]{A}{}$ and ${\rm supp} \;B_{,i}$. This form remains invariant under the change \eqref{action_var} in the action. Only its value changes, because the values of the dynamical variables themselves change. 
\end{remark}
\subsection{The reciprocity relation for physical observables}
It will be useful to introduce the notation 
\begin{equation}
D_{A}^{\pm}B \equiv \tensor[]{A}{_{,i}}G^{\mp ij}\tensor[_{j,}]{B}{} . \label{DAB_def}
\end{equation}
Equation \eqref{dist_ph_obs} may then be written:
\begin{eqnarray}
\delta^{\pm}B &=& D_{A\epsilon}^{\pm}B \nonumber\\
&=&\tensor[]{(A\epsilon)}{_{,i}}G^{\mp ij}\tensor[_{j,}]{B}{},
\end{eqnarray}
as may be seen by noting that $A$ and $\epsilon$ have the same type. 

Colloquially, $D_{A}^{-}B$ may be called the \lq\lq retarded effect of $A$ on $B$'' and $D_{A}^{+}B$ the \lq\lq advanced effect of $A$ on $B$''. It is a consequence of the reciprocity relations \eqref{recipr_rel} that 
\begin{eqnarray}
D_{A}^{\pm}B &\equiv & \tensor[]{A}{_{,i}}G^{\mp ij}\tensor[_{j,}]{B}{}\nonumber \\
&=&(-1)^{AB} B_{,i}G^{\pm ij}\tensor[_{j,}]{A}{} \nonumber\\
&=&(-1)^{AB}D_{B}^{\mp}A. \label{rec_rel_ph_obs}
\end{eqnarray}
In words: \textit{The retarded effect of} $A$ \textit{on} $B$ \textit{equals} $(-1)^{AB}$ \textit{times the advanced effect of} $B$ \textit{on} $A$ \textit{(and vice versa)}. This is known as the reciprocity relation for physical observables. 
\subsection{Off shell relations}
So far we have used the distributions $\tensor[_\alpha]{P}{_i}$, $\tensor[]{\eta}{^\alpha^\beta}$ only on shell. However, they, like the $\tensor[^i]{Q}{_\alpha}$, $\tensor[_i]{S}{_j}$, etc. have specific functional forms (as functionals of $\phi$) and are defined also off shell. Thus the operators $\tensor[_i]{F}{_j}$ and $\tensor[_\alpha]{\mathcal{F}}{_\beta}$ and Green's functions $G^{\pm ij}$ and $\mathcal{G}^{\pm  \alpha\beta}$ are defined off shell as well. Off shell, eq. \eqref{Ward_id_shell} no longer necessarily holds and the relations \eqref{F_Fcal_Green} \eqref{F_Fcal_Green2} fail to be satisfied generically. However, the operators $\tensor[_i]{F}{_j}$ and $\tensor[_\alpha]{\mathcal{F}}{_\beta}$ continue to be nonsingular and the Green's functions $G^{\pm ij}$ and $\mathcal{G}^{\pm  \alpha\beta}$ continue to exist, at least in an open neighborhood of the dynamical shell.

We need $G^{\pm ij}$ and $\mathcal{G}^{\pm \alpha\beta}$ off shell because we need to be able to calculate their functional derivatives, in order to discuss the invariance properties of $D_{A}^{\pm}B$. We begin by considering an arbitrary infinitesimal variation $\delta F$ in the operator $F$. This variation may arise either by shifting the point $\phi$ in $\Phi$, or by varying the functional forms of $P$, $\eta$, or even the action $S$ (the vector fields $\tensor{Q}{_\alpha}$ will be left untouched). It leads to corresponding variations $\delta G^{\pm}$ in the advanced and retarded Green's functions. 

The $\delta G^{\pm}$ satisfy a differential equation that is obtained by varying eq. \eqref{def_Gr_superc}:
\begin{eqnarray}
\delta(FG^{\pm})&=&\delta(-1),  \nonumber \\
\delta F G^{\pm} + F \delta G^{\pm} &=& 0,  \nonumber\\
F \delta G^{\pm} &=& - \delta F G^{\pm}.
\end{eqnarray}
This equation has the following unique solution: 
\begin{equation}
\delta G^{\pm} = G^{\pm} \delta F  G^{\pm} \label{var_Gr_F}
\end{equation}
which is determined by the kinematical conditions \eqref{ret_adv_Gr} that the Green's functions satisfy. It will be noted that the intersection of the supports of the two factors $G^{\pm}$ (with outer suppressed indices fixed) in the previous equation is compact, so that the operator $\delta F$ in this equation may act in either direction. 

Equation \eqref{var_Gr_F} is just the equation one would get if $F$ were a finite square matrix and $G^{\pm}$ were its negative inverse. There are important differences, however, between the present case and the case of finite matrices. First, $F$ has many \lq\lq inverses'', or Green's functions, not just one. Second, most of its Green's functions do \textit{not} satisfy variational equations having the structure \eqref{var_Gr_F}. For example, the average
\begin{equation}
\bar{G} \equiv \tfrac{1}{2}(G^+ +  G^-)
\end{equation}
is a Green's function of $F$:
\begin{eqnarray}
F\bar{G}&=&\tfrac{1}{2}F(G^+ + G^- ) \nonumber\\
&=&\tfrac{1}{2}FG^+ + \tfrac{1}{2} FG^-\nonumber \\
&=&-\tfrac{1}{2} -\tfrac{1}{2} \nonumber\\
&=&-1,
\end{eqnarray}
but it satisfies 
\begin{equation}
\delta \bar{G}= \tfrac{1}{2} (G^{+} \delta F G^{+} + G^{-} \delta F  G^{-} ),
\end{equation}
which is not equal to $\bar{G} \delta F  \bar{G}$.

If a Green's function does satisfy 
\begin{equation}
\delta G = G \delta F  G
\end{equation}
it will be called a \textit{coherent} Green's function. In many of the equations of this work the operators $F$, $\mathcal{F}$, $\tensor[_1]{S}{_1}$, $\delta F$, etc. will appear sandwiched between factors, the intersections of the supports of which are not compact. These operators will nevertheless be able to act in either direction because the Green's functions in these factors are all in the same \textit{coherence class} and establish an over-all set of boundary conditions that are preserved regardless of the direction of action. 

Suppose $\phi$ is on shell and suppose the variation $\delta F$ arises from variations $\delta P$ and $\delta \eta$ in $P$ and $\eta$: then
\begin{equation}
\delta F = \delta P^\sim \eta P + P^\sim \delta\eta P + P^\sim \eta\delta P.
\end{equation}
Inserting this expression in \eqref{var_Gr_F} and making use of eqs. \eqref{F_Fcal_Green2} and \eqref{F_Fcal_Green3}, as well as
\begin{equation}
\delta \eta = \eta \;\delta \lambda  \; \eta,
\end{equation}
one obtains 
\begin{eqnarray}
\delta G^{\pm}&=&G^{\pm}(\delta P^\sim \eta P + P^\sim \delta\eta P + P^\sim \eta \delta P )G^{\pm}\nonumber \\
&=&G^{\pm}\delta P^\sim \eta P G^{\pm}     +       G^{\pm}P^\sim(  \eta \;\delta \lambda  \; \eta)PG^{\pm}      +     G^{\pm} P^\sim \eta \delta P G^{\pm}\nonumber \\
&=&-G^{\pm}\delta P^\sim \eta\lambda \mathcal{G}^{\mp\sim} Q^\sim     -    Q\mathcal{G}^\pm \lambda(  \eta \;\delta \lambda  \; \eta)(-\lambda \mathcal{G}^{\mp\sim} Q^\sim )     -    Q\mathcal{G}^\pm \lambda\eta\delta P G^{\pm} \nonumber \\
&=& G^{\pm}\delta P^\sim \mathcal{G}^{\mp\sim} Q^\sim   +    Q\mathcal{G}^\pm \delta \lambda \mathcal{G}^{\mp\sim} Q^\sim    +    Q\mathcal{G}^\pm \delta P G^{\pm}. \label{var_Gr_onshell}
\end{eqnarray}
If the variation $\delta F$ arises instead from a variation in $\phi$, then one is led to the formula 
\begin{eqnarray}
\tensor[]{\delta G}{^\pm ^{ij}} &=& \tensor*[]{ G}{^\pm ^{ij}_{,k}} \; \delta\phi^k \nonumber\\
&=&\tensor*[]{ G}{^\pm ^{il}} \; \tensor*[_l]{ \delta F}{_m} \; \tensor*[]{ G}{^\pm ^{mj}} \nonumber\\
&=&\tensor*[]{ G}{^\pm ^{il}} \; \tensor*[_l]{  F}{_m_{,k}} \; \delta\phi^k \;\tensor*[]{ G}{^\pm ^{mj}} \nonumber\\
&=&(-1)^{km + kj}\;\tensor*[]{ G}{^\pm ^{il}} \;\tensor*[_l]{  F}{_m_{,k}} \;\tensor*[]{ G}{^\pm ^{mj}} \; \delta\phi^k,
\end{eqnarray}
and then
\begin{equation}
\tensor*[]{ G}{^\pm ^{ij}_{,k}} = (-1)^{km + kj} \; \tensor*[]{ G}{^\pm ^{il}} \; \tensor*[_l]{ F}{_m_{,k}} \; \tensor*[]{ G}{^\pm ^{mj}}
\end{equation}
Of course this formula is valid off shell. Going on shell one can proceed just as in the derivation of eq. \eqref{var_Gr_onshell} and obtain:
\begin{eqnarray}
\tensor*[]{ G}{^\pm ^{ij} _{,k}} &=& (-1)^{km + kj} \; \tensor*[]{ G}{^\pm ^{il}}  (\tensor[_{l,}]{S}{_{,m}}+\tensor[_l]{P}{^\sim _\alpha} \;\tensor[]{\eta}{^\alpha^\beta}\;\tensor[_\beta]{P}{_m})_{,k}  \;  \tensor*[]{ G}{^\pm ^{mj}} \nonumber\\
&=&(-1)^{km + kj}\; \tensor*[]{ G}{^\pm ^{il}}(\tensor[_{l,}]{S}{_{,mk}}  \nonumber\\ &&+ \tensor[_l]{P}{^\sim _\alpha}\;\tensor[]{\eta}{^\alpha^\beta}\;\tensor[_\beta]{P}{_m_{,k}}  \nonumber\\ && +  (-1)^{km+k\beta}\; \tensor[_l]{P}{^\sim _\alpha}\;\tensor[]{\eta}{^\alpha^\beta_{,k}}\;\tensor[_\beta]{P}{_m}  \nonumber\\ && +  (-1)^{km +k\alpha}\; \tensor[_l]{P}{^\sim _\alpha_{,k}}\;\tensor[]{\eta}{^\alpha^\beta}\;\tensor[_\beta]{P}{_m})\; \tensor*[]{ G}{^\pm ^{mj}} \nonumber\\
&=& (-1)^{km + kj} \; [\tensor*[]{ G}{^\pm ^{il}}\;\tensor[_{l,}]{S}{_{,mk}}\; \tensor*[]{ G}{^\pm ^{mj}}  \nonumber\\
&& +\tensor[^i]{Q}{_\alpha} \;\tensor[]{\mathcal{G}}{^\pm^\alpha^\beta}\;\tensor[_\beta]{P}{_m_{,k}} \;\tensor*[]{ G}{^\pm ^{mj}}  \nonumber\\
&& -(-1)^{km+k\beta}\;\tensor[^i]{Q}{_\alpha}\; \tensor[]{\mathcal{G}}{^\pm^\alpha^\beta}\;\tensor[_\beta]{\eta}{_\gamma}\;\tensor[]{\eta}{^\gamma^\theta_{,k}} \;\tensor[_\theta]{\lambda}{_\iota}\;\tensor[]{\mathcal{G}}{^\mp^\sim^\iota^\kappa}\;\tensor[_\kappa]{Q}{^\sim^j}  \nonumber\\
&& +(-1)^{km +k\alpha}\;\tensor*[]{ G}{^\pm ^{il}}\;\tensor[_l]{P}{^\sim _\alpha_{,k}}\;\tensor[]{\mathcal{G}}{^\mp^\sim^\alpha^\beta}\;\tensor[_\beta]{Q}{^\sim^j}]\nonumber \\
&=& (-1)^{km + kj}  [\tensor*[]{ G}{^\pm ^{il}}\;\tensor[_{l,}]{S}{_{,mk}} \;\tensor*[]{ G}{^\pm ^{mj}}  \nonumber\\
&& +\tensor[^i]{Q}{_\alpha} \;\tensor[]{\mathcal{G}}{^\pm^\alpha^\beta}\;\tensor[_\beta]{P}{_m_{,k}} \;\tensor*[]{ G}{^\pm ^{mj}}  \nonumber\\
&& +(-1)^{km+k\beta}\;\tensor[^i]{Q}{_\alpha}\; \tensor[]{\mathcal{G}}{^\pm^\alpha^\beta}\;\tensor[_\beta]{\lambda}{_\gamma_{,k}}\;\tensor[]{\mathcal{G}}{^\mp^\sim^\gamma^\theta}\;\tensor[_\theta]{Q}{^\sim^j}  \nonumber \\
&& +(-1)^{km +k\alpha}\;\tensor*[]{ G}{^\pm ^{il}}\;\tensor[_l]{P}{^\sim _\alpha_{,k}}\;\tensor[]{\mathcal{G}}{^\mp^\sim^\alpha^\beta}\;\tensor[_\beta]{Q}{^\sim^j}]. \label{Gr_1_on_shell}
\end{eqnarray}
It will be noted that summation-integrations can be performed in any order in the previous equations because all the Green's functions are in the same coherence class. 
\subsection{Invariance properties of $D_{A}^{\pm}B$}
Expression \eqref{DAB_def} for $D_{A}^{\pm}B$ involves the Green's functions $G^{\pm}$, which depend on specific choice for the operators $P$ and $\eta$. Since $D_{A}^{\pm}B$ represent the physical effects of physical changes in the action they must be $P$- and $\eta$- independent. To verify this write the variation of \eqref{DAB_def}, under changes in $P$ and $\eta$, in the supercondensed notation:
\begin{equation}
\delta D_{A}^{\pm}B = \tensor[]{A}{_{,i}}\delta G^{\mp ij}\tensor[_{j,}]{B}{}.
\end{equation}
If $\phi$ is on shell, one can insert \eqref{var_Gr_onshell} into the right-hand side. The result is immediately seen to vanish,
\begin{equation}
\delta D_{A}^{\pm}B = 0,
\end{equation}
because of the on shell invariance conditions 
\begin{equation}
\tensor{A}{_1}Q=0 \; \; \; \; \; \; \; \; \tensor{Q}{^\sim}\tensor[_1]{B}{}=0 \label{A_B_obs_on_shell}
\end{equation}
satisfied by $A$ and $B$ as physical observables.

For each on shell $\phi$ the \textit{values} of the $D_{A}^{\pm}B$ are physical, i.e., $P$- and $\eta$- independent. However, as \textit{functionals of the} $\phi$ (off shell as well as on) the $D_{A}^{\pm}B$ turn out \textit{not} to be physical observables unless both $A$ and $B$ are absolute invariants. To see this introduce parameters $\xi^\alpha$ of compact support in $M$, and make use of the supercondensed notation. Then if one evaluates the following quantity on shell, one obtains:
\begin{eqnarray}
(D_{A}^{\pm}B)_{1}Q\xi &=& (\tensor{A}{_1} \; \tensor[]{G}{^\mp} \; \tensor[_1]{B}{})_{1}Q\xi  \nonumber\\
&=&\tensor{A}{_1} \; \tensor[]{G}{^\mp} \; \tensor[_1]{B}{_1}Q\xi +  \tensor{A}{_1} \; \tensor[]{G}{^\mp _1} \; Q\xi \; \tensor[_1]{B}{} + \xi^\sim Q^\sim \; \tensor[_1]{A}{_1}\; \tensor[]{G}{^\mp} \; \tensor[_1]{B}{} \nonumber\\
&=&\tensor{A}{_1} \; \tensor[]{G}{^\mp} \; \tensor[_1]{B}{_1}Q\xi +  \tensor{A}{_1} \; \tensor[]{G}{^\mp} \; \tensor[_1]{S}{_2} \; Q\xi \; \tensor[]{G}{^\mp} \; \tensor[_1]{B}{} + \xi^\sim Q^\sim \; \tensor[_1]{A}{_1}\; \tensor[]{G}{^\mp} \; \tensor[_1]{B}{}, \nonumber \\ \label{DAB_Qxi}
\end{eqnarray}
where eqs. \eqref{Gr_1_on_shell} and \eqref{A_B_obs_on_shell} have been used.

In order to evaluate $\tensor[_1]{B}{_1}Q\xi$ and $\xi^\sim Q^\sim \; \tensor[_1]{A}{_1}$, multiply 
\begin{equation}
\tensor{A}{_1}\;Q=\tensor[]{S}{_1} \; a \; \; \; \; \; \; \; \; \tensor[]{B}{_1} \; \tensor[]{Q}{}=\tensor[]{S}{_1} \; b
\end{equation}
by $\xi$ and functionally differentiate; the result is (on shell):
\begin{equation}
\tensor[_1]{A}{_1}\;Q\xi=-\tensor[_1]{(\xi^\sim Q^\sim)}{} \; \tensor[_1]{A}{}+\tensor[_1]{S}{_1} \; a\xi \; \; \; \; \; \; \; \; \tensor[_1]{B}{_1}\;Q\xi=-\tensor[_1]{(\xi^\sim Q^\sim)}{} \; \tensor[_1]{B}{}+\tensor[_1]{S}{_1} \; b\xi.  \label{A_B_on_shell2}
\end{equation}
Functionally differentiating twice $\tensor[]{S}{_1} \; Q\xi = 0$ and then going on shell, one obtains the third Ward identity:

First functional derivative:
\begin{align}
\tensor[_1]{S}{_1} \; Q\xi  = -\tensor[_1]{(\xi^\sim Q^\sim)}{} \; \tensor[_1]{S}{}
\end{align}

Second functional derivative:
\begin{align}
\tensor[_1]{S}{_2} \; Q\xi  +  \tensor[_1]{S}{_1} \; \tensor[]{(Q\xi)}{_1} =  -\tensor[_1]{(\xi^\sim Q^\sim)}{} \; \tensor[_1]{S}{_1} -\tensor[_1]{(\xi^\sim Q^\sim)}{_1} \; \tensor[_1]{S}{}
\end{align}

Going on shell:
\begin{align}
\tensor[_1]{S}{_2} \; Q\xi = -  \tensor[_1]{S}{_1} \; \tensor[]{(Q\xi)}{_1} -\tensor[_1]{(\xi^\sim Q^\sim)}{} \; \tensor[_1]{S}{_1}. \label{third_Ward_on_shell}
\end{align}
Hence, inserting eqs. \eqref{A_B_on_shell2} and \eqref{third_Ward_on_shell} in \eqref{DAB_Qxi}, one obtains:
\begin{eqnarray}
(D_{A}^{\pm}B)_{1}Q\xi &=& \tensor{A}{_1} \; \tensor[]{G}{^\mp} \; (-\tensor[_1]{(\xi^\sim Q^\sim)}{} \; \tensor[_1]{B}{}+\tensor[_1]{S}{_1} \; b\xi)  \nonumber\\ 
&& + \tensor{A}{_1} \; \tensor[]{G}{^\mp} \; (-  \tensor[_1]{S}{_1} \; \tensor[]{(Q\xi)}{_1} -\tensor[_1]{(\xi^\sim Q^\sim)}{} \; \tensor[_1]{S}{_1}) \tensor[]{G}{^\mp}  \; \tensor[_1]{B}{}  \nonumber \\
&& + (- \; \tensor{A}{_1} \; \tensor[]{(Q\xi)}{_1} + \tensor[]{( \xi^\sim a^\sim)}{} \; \tensor[_1]{S}{_1} )\; \tensor[]{G}{^\mp} \; \tensor[_1]{B}{} \nonumber \\
&=&-\tensor{A}{_1} \; \tensor[]{G}{^\mp} \; \tensor[_1]{(\xi^\sim Q^\sim)}{} \; \tensor[_1]{B}{} + \tensor{A}{_1} \; (- \tensor[]{\Pi}{^\mp}) \; b\xi  \nonumber\\ 
&& - \tensor{A}{_1}  \; (-\tensor[]{\Pi}{^\mp}) \; \tensor[]{(Q\xi)}{_1} \; \; \tensor[]{G}{^\mp} \; \tensor[_1]{B}{} -  \tensor{A}{_1} \; \tensor[]{G}{^\mp} \; \tensor[_1]{(\xi^\sim Q^\sim)}{} \;(-\tensor[]{\Pi}{^\pm ^\sim} ) \; \tensor[_1]{B}{}  \nonumber\\
&& - \; \tensor{A}{_1} \; \tensor[]{(Q\xi)}{_1} \; \tensor[]{G}{^\mp} \; \tensor[_1]{B}{} + \xi^\sim a^\sim \; (- \tensor[]{\Pi}{^\pm ^\sim}) \; \tensor[_1]{B}{} \nonumber\\
&=& -\tensor{A}{_1} \; b\xi - \xi^\sim a^\sim  \; \tensor[_1]{B}{},
\end{eqnarray}
in which the properties of the projection operator $\Pi^{\pm}$ and eq. \eqref{A_B_obs_on_shell} have been exploited in obtaining the final expression.

This remarkably simple expression has two important properties. First, it vanishes if $A$ and $B$ are absolute invariants, i.e., if $a$ and $b$ vanish.  Second, it is independent of the $\pm$ signs, being the same for both retarded and advanced disturbances. 

Another result that is independent of the $\pm$ signs is the following, which shows explicitly that $D_{A}^{\pm}B$ is not invariant under changes in $A$ and $B$ of the form \eqref{obs_transf}; on shell:
\begin{eqnarray}
D_{\bar{A}}^{\pm}\bar{B}-D_{A}^{\pm}B &=& \tensor{(A}{}+\tensor{a}{^l}\tensor[_{l,}]{S}{}\tensor[]{)}{_{,i}} \; G^{\mp ij} \; \tensor[_{j,}]{(B}{}+\tensor[]{S}{_{,k}}\tensor[^k]{b)}{} - \tensor[]{A}{_{,i}}G^{\mp ij}\tensor[_{j,}]{B}{} \nonumber \\
&=& (A+ a^\sim \; \tensor[_1]{S}{}\tensor{)}{_1} \; \tensor[]{G}{^\mp} \; \tensor[_1]{(}{} B + \tensor{S}{_1} \; b) - \tensor[]{A}{_1} \; \tensor[]{G}{^\mp} \; \tensor[_1]{B}{} \nonumber\\
&=& \tensor{A}{_1} \; \tensor[]{G}{^\mp} \; \tensor[_1]{(}{} \tensor{S}{_1} \; b)   \nonumber\\ 
&& + (a^\sim \; \tensor[_1]{S}{}\tensor{)}{_1} \; \tensor[]{G}{^\mp} \; \tensor[_1]{(}{} B )   \nonumber\\
&& + (a^\sim \; \tensor[_1]{S}{}\tensor{)}{_1} \; \tensor[]{G}{^\mp} \; \tensor[_1]{(}{} \tensor{S}{_1} \; b) \nonumber\\
&=& \tensor{A}{_1} \; \tensor[]{G}{^\mp} \; \tensor[_1]{S}{_1} \; b + \tensor{A}{_1} \; \tensor[]{G}{^\mp} \; \tensor{S}{_1} \; \tensor[_1]{b}{} \nonumber \\
&& + a \; \tensor[_1]{S}{_1} \; \tensor{G}{^\mp} \; \tensor[_1]{B}{}  + \tensor{a}{^\sim _1} \; \tensor[_1]{S}{} \; \tensor{G}{^\mp} \; \tensor[_1]{B}{} \nonumber\\
&& + a^\sim \; \tensor[_1]{S}{_1} \; \tensor[]{G}{^\mp} \;  \tensor[_1]{S}{_1} \; b  \nonumber\\ 
&& +  a^\sim \; \tensor[_1]{S}{_1} \; \tensor[]{G}{^\mp} \; \tensor{S}{_1} \; \tensor[_1]{b}{}  \nonumber\\
&& + \tensor{a}{^\sim _1} \tensor[_1]{S}{} \; \tensor[]{G}{^\mp} \;  \tensor[_1]{S}{_1} \; b  \nonumber\\
&& + \tensor{a}{^\sim _1} \tensor[_1]{S}{} \; \tensor[]{G}{^\mp} \;  \tensor{S}{_1} \; \tensor[_1]{b}{} \nonumber \\
&=& \tensor{A}{_1} \; (-\tensor{\Pi}{^\mp}) \; b  \nonumber\\
&& + a^\sim \; (-\tensor{\Pi}{^\pm ^\sim}) \; \tensor[_1]{B}{}  \nonumber\\
&& + a^\sim \; \tensor[_1]{S}{_1} \; (-\tensor{\Pi}{^\mp}) \; b \nonumber \\
&=& -\tensor{A}{_1} \; b - \tensor{a}{^\sim}  \; \tensor[_1]{B}{} - a^\sim \; \tensor[_1]{S}{_1} \; b, \label{DAB_bar}
\end{eqnarray}
where 
\begin{equation}
\bar{A}=A+ \tensor{a}{^i}\tensor[_{i,}]{S}{}\; \; \; \; \; \; \; \; \bar{B}= B+\tensor{S}{_{,k}}\tensor[^k]{b}{}. \label{bar_AB}
\end{equation}
Here $\tensor{a}{^i}$ and $\tensor[^k]{b}{}$ are assumed to have the properties (e.g., rapid fall-off in the past and future) that are necessary for the associative law of multiplication to hold. 
\subsection{Peierls bracket, supercommutator function}
Let $A$ and $B$ be two physical observables. Their \textit{Peierls Bracket} is defined to be
\begin{equation}
(A,B) \equiv D_{A}^{-}B - (-1)^{AB}D_{B}^{-}A.
\end{equation}
Using the definition \eqref{DAB_def} and the reciprocity relation for physical observables \eqref{rec_rel_ph_obs} one may re-express this bracket in the form
\begin{eqnarray}
(A,B) &=& D_{A}^{-}B - D_{A}^{+}B \nonumber \\
&=& \tensor{A}{_{,i}} \; \tensor{G}{^+ ^{ij}} \; \tensor[_{j,}]{B}{} - \tensor{A}{_{,i}} \; \tensor{G}{^- ^{ij}} \; \tensor[_{j,}]{B}{} \nonumber \\
&=& \tensor{A}{_{,i}} \; \tensor{\tilde{G}}{^{ij}} \; \tensor[_{j,}]{B}{},  \label{P_bracket}
\end{eqnarray}
where 
\begin{equation}
\tensor{\tilde{G}}{^{ij}} \equiv \tensor{G}{^+ ^{ij}} - \tensor{G}{^- ^{ij}}.
\end{equation}
In anticipation of its role in quantum theory $\tilde{G}$ wll be called the \textit{supercommutator function}; it has the symmetry properties
\begin{eqnarray}
\tensor{\tilde{G}}{^{ji}} &=&  \tensor{G}{^+ ^{ji}} - \tensor{G}{^- ^{ji}} \nonumber\\
&=&(-1)^{ij} (\tensor{G}{^- ^{ij}} -\tensor{G}{^+ ^{ji}}) \nonumber \\
&=&-(-1)^{ij} \tensor{\tilde{G}}{^{ij}}, \label{symm_supercomm}
\end{eqnarray}
or, in supercondensed notation
\begin{equation}
\tilde{G}^{\sim} = - \tilde{G}.
\end{equation}

Unlike $D_{A}^{\pm}B$ \textit{the Peierls bracket}, as a functional of $\phi$, \textit{is always a physical observable} no matter whether $A$ and $B$ are absolute invariants or conditional invariants. This follows from \eqref{DAB_Qxi}, which yields
\begin{eqnarray}
\tensor[]{(A,B)}{_1} Q\xi &=& \tensor{(D_{A}^{-}B)}{_1} Q\xi - \tensor{(D_{A}^{+}B)}{_1} Q\xi \nonumber \\
&=&0.
\end{eqnarray}
Moreover, we also have, using \eqref{DAB_bar}
\begin{eqnarray}
(\bar{A},\bar{B}) &=& D_{\bar{A}}^{-}\bar{B} - D_{\bar{A}}^{+}\bar{B} \nonumber \\
&=& (D_{A}^{-}B -\tensor{A}{_1} \; b - \tensor{a}{^\sim}  \; \tensor[_1]{B}{} - a^\sim \; \tensor[_1]{S}{_1} \; b)  \nonumber\\
&& - (D_{A}^{+}B-\tensor{A}{_1} \; b - \tensor{a}{^\sim}  \; \tensor[_1]{B}{} - a^\sim \; \tensor[_1]{S}{_1} \; b)  \nonumber\\
&=& D_{A}^{-}B - D_{A}^{+}B \nonumber\\
&=& (A,B), \label{bracket_bar_AB}
\end{eqnarray}
where $\bar{A}$, $\bar{B}$ are given by \eqref{bar_AB}. Equation \eqref{bracket_bar_AB} shows that \textit{it is immaterial whether the dynamical equations are used before or after computing the Peierls bracket. That is, use of Peierls bracket commutes with use of any on shell conditions or restrictions}.

If the dynamical system possesses no invariant flows, then the $\phi$ are themselves physical observables, and eq. \eqref{P_bracket} implies 
\begin{equation}
(\phi^i,\phi^j) = \tensor{\tilde{G}}{^{ij}} \label{P_bracket_phi)}
\end{equation}
When invariant flows are present the Peierls bracket of the $\phi$ is not defined. However, in computing the brackets of observables one may proceed \textit{as if} it were given by eq. \eqref{P_bracket_phi)}.
\subsection{The bracket identities}
Let ${A^{\alpha}}$ and ${B^{\alpha}}$ be any two families of physical observables, and let $U(A)$ and $V(B)$ be two any functions of these families. Equation \eqref{P_bracket} has the immediate corollary
\begin{equation}
(U(A),V(B))= U(A) \frac{\overleftarrow{\partial}}{\partial A^\alpha}(A^{\alpha},B^{\beta}) \frac{\overrightarrow{\partial}}{\partial B^\beta} V(B). \label{P_bracket_functions}
\end{equation}
The following properties are called \textit{simple} identities satisfied by the Peierls Bracket; their proof is straightforward:
\begin{eqnarray}
(A,B+C) &=& (A,B) + (A,C),  \\
(A,BC)  &=& (A,B)C + (-1)^{AB}B(A,C), \\
(A,B)   &=& -(-1)^{AB}(B,A).
\end{eqnarray}

\section{Functional Integration: QFT with no Invariance Flows}
\subsection{Problems with heuristic quantization rules}
Given a classical field theory, how does one pass from the classical theory to the quantum theory? Traditionally, one attempts to answer the first question by starting from a \textit{classical} (or \textit{super}classical) dynamical system and obtaining from it a corresponding quantum system. Each real dynamical variable $\phi$ is replaced by a self-adjoint linear operator $\hat{\phi}$ of the same type ($c$-type or $a$-type), and these operators are assumed to satisfy differential equations similar, if not identical, in form to the dynamical equations of the classical theory. The only differences are: (i) a particular choice of factor ordering may have to be made, (ii) renormalization constants may have to be inserted, and (iii) in order to maintain consistency one may have to add extra terms that do not appear in the (super)classical theory. 

In the same way, a functional $R[\phi]$ on $\Phi$ is replaced by an operator $\hat{R}\equiv R[\hat{\phi}]$. The super Hilbert or Fock space on which the operators act is not given a priori but is constructed in such a way as to yield a representation of the operator superalgebra satisfied by the $\hat{\phi}$. This superalgebra is always determined in some way by the \textit{heuristic quantization rule}
\begin{equation}
(\phi^j,\phi^k) \mapsto (\hat{\phi}^j,\hat{\phi}^k)\equiv-i[\hat{\phi}^j,\hat{\phi}^k]  \; \; \; \; (\hbar=1), \label{Peierls_comutator}
\end{equation} 
which tries to identify, up to a factor $i$, each Peierls bracket with a supercommutator. When the (super)classical action $S$ possesses invariant flows, the variables $\hat{\phi}^i$, and hence the supercommutator \eqref{Peierls_comutator}, are defined only \textit{modulo} invariance transformations. If there are no invariant flows one may in principle write
\begin{equation}
[\hat{\phi}^j,\hat{\phi}^k]=i \hat{\tilde{G}}^{jk} ,
\end{equation}
but there is a difficulty. When the dynamical equations are nonlinear the quantum supercommutator function $\hat{\tilde{G}}^{ij}$ is not just the identity operator times the classical $\tilde{G}^{ij}$ but depends on the $\hat{\phi}^i$. It is usually difficult if not impossible to give a simple factor-ordering prescription for passing from $\tilde{G}^{ij}$ to $\hat{\tilde{G}}^{ij}$. The difficulty is even greater with the more general quantization rule 
\begin{equation}
[\hat{A},\hat{B}]=i (\hat{A},\hat{B}) \overset{?}{=}  i\tensor{\hat{A}}{_{,j}} \; \tensor{\hat{\tilde{G}}}{^{jk}} \; \tensor[_{k,}]{\hat{B}}{},
\end{equation}
which is applicable in principle to all physical observables (i.e., flow invariants). Even though simple factor-ordering prescriptions may exist for defining $\hat{A}$ and $\hat{B}$ there will often be no simple prescription for passing from the classical $\tensor{A}{_{,i}} \; \tensor{\tilde{G}}{^{ij}} \; \tensor[_{j,}]{B}{}$ to its quantum analog.

A possible way out of these difficulties will be discussed in the following. 

\subsection{Transition amplitudes}
In classical physics the dynamical equations are of central importance because their solutions correspond directly to reality. In quantum physics the situation is different. Solutions of the dynamical equations represent the system only in a generic sense. Correspondence to reality can be set up only when the state vector has been specified.

Instead of making direct use of the operator dynamical equations one can express the dynamical content of the quantum theory in another form, which brings the state vector into the picture and which is often more useful in applications. Let $\hat{A}$, $\hat{B}$ be any two physical observables of a given system which satisfy
\begin{equation}
{\rm supp}\; \hat{A}_{,i} > {\rm supp}\; \hat{B}_{,i}.
\end{equation}
That is, $\hat{A}$ is constructed out of $\hat{\phi}$ taken from a region of space-time that lies to the future of the region from which the $\hat{\phi}$ making up $\hat{B}$ are taken. Let $|a\rangle$ and $|b\rangle$ be normalized physical eigenvectors of $\hat{A}$ and $\hat{B}$ respectively, corresponding to physical eigenvalues $a$ and $b$. The inner product $\langle a|b \rangle$ is often called a \textit{transition amplitude}. If the state vector of the system is $|b\rangle$, then $\langle a|b \rangle$ is the \textit{probability amplitude} for the system to be found in the state represented by $|a\rangle$, i.e., for the value $a$ to be obtained when $\hat{A}$ is measured. The probability itself is $|\langle a|b \rangle|^2$. 

\subsection{The Schwinger variational principle}
Suppose the action functional of the classical theory suffers an infinitesimal change $\delta S$; the quantum theory will change accordingly, and we shall postulate that the associated change $\delta \hat{S}$ in the quantum action $\hat{S}$ is self-adjoint. This produces a change in the quantum dynamical equations and hence a change in their solutions $\hat{\phi}^i$. Suppose the forms of $\hat{A}$ and $\hat{B}$ as functional of the $\hat{\phi}^i$ remain unchanged. As operators, $\hat{A}$ and $\hat{B}$ will nevertheless be changed because the  $\hat{\phi}^i$ have changed. Denote these changes by $\delta\hat{A}$ and $\delta\hat{B}$  respectively. The eigenvectors $|a\rangle$ and $|b\rangle$ too will suffer changes $\delta|a\rangle$ and $\delta|b\rangle$. The precise nature of these changes will depend on boundary conditions. 

Suppose $\delta \hat{S}$ satisfies the condition 
\begin{equation}
{\rm supp}\; \hat{A}_{,i} >\delta \hat{S}_{,i}> {\rm supp} \; \hat{B}_{,i}. \label{kin_rel_Schw}
\end{equation}
That is, suppose $\delta \hat{S}$ is constructed out of $\hat{\phi}$'s taken from a region of space-time that lies to the past of the region associated with $\hat{A}$ and to the future of the region associated with $\hat{B}$. Suppose furthermore that retarded boundary conditions are adopted. Then at times to the past of the region associated with $\delta \hat{S}$ the dynamical variables $\hat{\phi}^i$ will remain unchanged. This means that
\begin{equation}
\delta \hat{B}=0. \label{perturb_B_Schw}
\end{equation}
The observable $\hat{A}$, on the other hand, suffers a change which, with DeWitt's notation, can be written as
\begin{equation}
\delta \hat{A}= D_{\delta \hat{S}}^{-}\hat{A}.
\end{equation}
In view of the kinematical relation \eqref{kin_rel_Schw} one has also:
\begin{equation}
D_{\hat{A}}^{-}\delta \hat{S}=0,
\end{equation}
and hence
\begin{equation}
\delta \hat{A}=D_{\delta \hat{S}}^{-}\hat{A}-D_{\hat{A}}^{-}\delta \hat{S} \equiv (\delta \hat{S},\hat{A}).
\end{equation}
Therefore, imposing the heuristic quantization rule, one finds
\begin{equation}
\delta \hat{A}=-i[\delta \hat{S},\hat{A}]. \label{perturb_obs_Schw}
\end{equation}
Since the relation between Peierls brackets and supercommutators is only heuristic, the derivation of eq. \eqref{perturb_obs_Schw} is hardly rigorous. Indeed, if an arbitrary operator ordering is chosen for the dynamical equations, eq. \eqref{perturb_obs_Schw} need \textit{not} hold. However, there is an inevitability and elegance about this equation which suggests that one turn the problem around and demand that the dynamics be such that it \textit{does} hold, whatever operator ordering may be chosen for $\delta \hat{S}$ as a functional of the $\hat{\phi}$'s. Note that if it holds for $\hat{A} = \hat{\phi}^j$, with $j>\delta \hat{S}_{,i}$, \eqref{perturb_obs_Schw} reads
\begin{eqnarray}
\hat{\phi}^j+\delta \hat{\phi}^j&=&\hat{\phi}^j-i[\delta \hat{S},\hat{\phi}^j] \nonumber \\
&\equiv &-i\delta \hat{S}\hat{\phi}^j + i\hat{\phi}^j\delta \hat{S} \nonumber \\
&=& \hat{u}^{-1}\hat{\phi}^j\hat{u}, \label{unit_trans}
\end{eqnarray}
with
\begin{equation}
\hat{u} \equiv 1 + i\delta \hat{S},
\end{equation}
hence \eqref{unit_trans} is a unitary transformation, and \eqref{perturb_obs_Schw} holds for all $\hat{A}$ satisfying the kinematical inequality \eqref{kin_rel_Schw}:
\begin{eqnarray}
\hat{A}+\delta \hat{A}&=&A[\hat{\phi}+\delta\hat{\phi}] \nonumber \\
&=& A[\hat{u}^{-1}\hat{\phi}\hat{u}] \nonumber \\
&=& \hat{u}^{-1}A[\hat{\phi}]\hat{u} \nonumber \\
&=& \hat{A} -i[\delta \hat{S},\hat{A}]. 
\end{eqnarray}
In this work the previous equations will be postulated as rigorous statements of quantum dynamics; moreover, we shall try to costrain the structure $\delta \hat{S}$ so that corresponding to each classical theory (i.e., to each action functional $S$) there is a virtually unique quantum theory or at most a unique \textit{family} of quantum theories. 

\textit{Modulo} an ignorable phase change $i\delta\theta|a \rangle$, with $\delta\theta$ real $c$-number, the change \eqref{perturb_obs_Schw} induces a change in the eigenvector $|a \rangle$ given by 
\begin{equation}
|a \rangle + \delta |a \rangle = \hat{u}^{-1}|a \rangle, \label{perturb_autok_Schw}
\end{equation}
as is straightforward to verify:
\begin{eqnarray}
(\hat{A}+\delta \hat{A})(\hat{u}^{-1}|a \rangle) &=& (\hat{u}^{-1}A[\hat{\phi}]\hat{u})(\hat{u}^{-1}|a \rangle) \nonumber \\
&=& \hat{u}^{-1}A[\hat{\phi}]|a \rangle \nonumber \\
&=& a(\hat{u}^{-1}|a \rangle). 
\end{eqnarray}
Equation \eqref{perturb_autok_Schw} is equivalent to
\begin{equation}
\delta |a \rangle = -i\delta \hat{S}|a \rangle.
\end{equation}
Equation \eqref{perturb_B_Schw}, on the other hand, implies:
\begin{equation}
\delta |b \rangle = 0
\end{equation}
\textit{modulo} a similar ignorable phase change. Hence
\begin{eqnarray}
\delta \langle a|b \rangle &=& (\delta \langle a|)|b \rangle +  \langle a|(\delta |b \rangle) \nonumber\\
&=& i \langle a|\delta \hat{S}|b \rangle . \label{Schw_Princ}
\end{eqnarray}
Equation \eqref{Schw_Princ} is known as the \textit{Schwinger Variational Principle}. Although the Schwinger variational principle was \lq\lq derived'' through imposition of retarded boundary conditions, it is in fact independent of boundary conditions. For example, if advanced boundary conditions are imposed and use is made of the reciprocity relation \eqref{rec_rel_ph_obs}, then eqs. \eqref{perturb_obs_Schw} and \eqref{perturb_B_Schw} get replaced by 
\begin{eqnarray}
& &\delta \hat{A}=0, \\
& &\delta \hat{B}= D_{\delta \hat{S}}^{+}\hat{B} = D_{\hat{B}}^{-}\delta \hat{S} - D_{\delta \hat{S}}^{-}\hat{B} = (\hat{B},\delta \hat{S}) = -i[\hat{B},\delta \hat{S}],
\end{eqnarray}
which imply
\begin{equation}
\delta |a \rangle = 0, \; \; \; \; \; \; \delta |b \rangle = i\delta \hat{S}|b \rangle
\end{equation}
again leading to \eqref{Schw_Princ}.

Whether one imposes retarded or advanced boundary conditions, or something in between, the following statements are always true: 
\begin{itemize}
	\item[1.] The unperturbed dynamical equations continue to hold in the regions to the past and to the future of ${\rm supp }\; \delta \hat{S}_{,i}$.
	\item[2.] The $\hat{\phi} + \delta\hat{\phi}$ in these regions are related to the unperturbed $\hat{\phi}$ by unitary transformations. 
\end{itemize}
\begin{remark} 
	Being the variation \eqref{Schw_Princ} a unitary transformation the Schwinger principle is guaranteed to preserve both the probability interpretation of the quantum theory and the unit normalization of total probability. 
\end{remark}
\begin{remark} 
	The particular choice of physical observables $\hat{A}$ and $\hat{B}$ in the statement of the Schwinger principle is irrelevant. Only the condition \eqref{kin_rel_Schw} is important. Since the eigenvalues of more than one observable usually have to be specified in order to determine a quantum state uniquely, it will be convenient from now on to replace \eqref{Schw_Princ} by the more general statement 
	\begin{equation}
	\delta \langle {\rm out }|{\rm in } \rangle = i \langle  {\rm out }|\delta \hat{S}|{\rm in } \rangle,
	\end{equation}
	where $|{\rm in } \rangle$ and $|{\rm out } \rangle$ are state supervectors determined by some unspecified conditions on the dynamics in regions respectively to the past and to the future of the region in which one may wish to vary the action. 
\end{remark}
\subsection{External sources and chronological products}
When the action possesses no invariant flows, a particularly convenient way to vary $\hat{S}$ is to append to it a term of the form $J_{i}\hat{\phi}^i$, where the $J_{i}$ are pure supernumber-valued functions over space-time, $c$-type and real when $\hat{\phi}^i$ is $c$-type, $a$-type and imaginary when $\hat{\phi}^i$ is a-type. The $J_{i}$ are called \textit{external sources}.

Let the external sources suffer variations $\delta J_{i}$, whose supports are confined to the space-time region lying, in time, between the regions associated with the state supervectors $|{\rm in } \rangle$ and $|{\rm out } \rangle$. Then the transition amplitude  $\langle {\rm out }|{\rm in } \rangle$ suffers the change 
\begin{equation}
\delta \langle {\rm out }|{\rm in } \rangle = i \langle  {\rm out }|\delta J_{j}\hat{\phi}^j|{\rm in } \rangle,
\end{equation}
which implies
\begin{equation}
\frac{\overrightarrow{\delta}}{i\delta J_{j}} \langle {\rm out }|{\rm in } \rangle = (-1)^{jF} \langle  {\rm out }|\hat{\phi}^j|{\rm in } \rangle,
\end{equation}
where $F$ is the fermionic number of $|{\rm out } \rangle$, i.e., $F$ is $0$ or $1$ according as $|{\rm out } \rangle$ is $c$-type or $a$-type. Let ${|\phi\rangle}$ be a complete set of normalized physical eigenvectors of $\hat{\phi}^j$, corresponding to the eigenvalues $\phi^j$. Such eigenvectors exist since, when no invariant flows are present, the $\hat{\phi}^j$ are physical observables. The previous equation may then be rewritten in the form 
\begin{equation}
\frac{\overrightarrow{\delta}}{i\delta J_{j}} \langle {\rm out }|{\rm in } \rangle = (-1)^{jF} \sum \langle  {\rm out }|\phi \rangle \phi^j \langle \phi|{\rm in } \rangle ,\label{deltaJ1}
\end{equation}
where the summation is over all the $|\phi \rangle$. 
Now let $\delta J_{i}$ be a second variation in the sources, and suppose ${\rm supp }\; \delta J_{i}>j$. Then the factor $\langle \phi|{\rm in } \rangle$ in 	\eqref{deltaJ1} remains unchanged, and
\begin{eqnarray}
\delta \frac{\overrightarrow{\delta}}{i\delta J_{j}} \langle {\rm out }|{\rm in } \rangle &=& (-1)^{jF} \sum (\delta \langle  {\rm out }|\phi \rangle) \phi^j \langle \phi|{\rm in } \rangle \nonumber \\
&=& (-1)^{jF} \sum (\langle  {\rm out }|\delta J_{k}\hat{\phi}^k|\phi \rangle) \phi^j \langle \phi|{\rm in } \rangle \nonumber \\
&=& (-1)^{jF} \langle  {\rm out }|\delta J_{k}\hat{\phi}^k \hat{\phi}^j {\rm in } \rangle.
\end{eqnarray}
Therefore
\begin{equation}
\frac{\overrightarrow{\delta}}{i\delta J_{i}} \frac{\overrightarrow{\delta}}{i\delta J_{j}} \langle {\rm out }|{\rm in } \rangle = (-1)^{(i+j)F} \langle  {\rm out }|\hat{\phi}^i \hat{\phi}^j|{\rm in } \rangle.
\end{equation}
If, on the other hand, ${\rm supp }\; \delta J_{i}<j$, then 
\begin{eqnarray}
\delta \frac{\overrightarrow{\delta}}{i\delta J_{j}} \langle {\rm out }|{\rm in } \rangle &=& (-1)^{jF} \sum  \langle  {\rm out }|\phi \rangle \phi^j (\delta\langle \phi|{\rm in } \rangle) \nonumber \\
&=& i(-1)^{jF} \sum  \langle  {\rm out }|\phi \rangle \phi^j (\langle \phi|\delta J_{k}\hat{\phi}^k|{\rm in } \rangle) \nonumber \\
&=& i(-1)^{jF} \langle  {\rm out }|\hat{\phi}^j \delta J_{k}\hat{\phi}^k|{\rm in } \rangle) \nonumber \\
&=& i(-1)^{jF} \langle  {\rm out }|\hat{\phi}^j \delta J_{k}\hat{\phi}^k|{\rm in } \rangle),
\end{eqnarray}
and 
\begin{equation}
\frac{\overrightarrow{\delta}}{i\delta J_{i}} \frac{\overrightarrow{\delta}}{i\delta J_{j}} \langle {\rm out }|{\rm in } \rangle = (-1)^{(i+j)F+ij}\langle  {\rm out }|\hat{\phi}^j \hat{\phi}^i|{\rm in } \rangle.
\end{equation}
Continuing in this manner one obtains, quite generally, 
\begin{equation}
\frac{\overrightarrow{\delta}}{i\delta J_{i_1}}... \frac{\overrightarrow{\delta}}{i\delta J_{i_n}} \langle {\rm out }|{\rm in } \rangle = (-1)^{(i_1+...i_n)F} \langle {\rm out }|\mathcal{T}(\hat{\phi}^{i_1}...\hat{\phi}^{i_n})|{\rm in } \rangle \label{funct_der_n}
\end{equation}
where $\mathcal{T}$ is the \textit{chronological ordering operator}, which rearranges the factors $\hat{\phi}_{i_1}...\hat{\phi}_{i_n}$ so that the times associated with the indices appear in chronological sequence, increasing from right to left, and which inserts an additional factor $-1$ for each interchange of a pair of $a$-type indices that occurs in the carrying out of this rearrangement. 

In the above derivation of \eqref{funct_der_n} the times associated with the indices were assumed to be in a well defined chronological order. However, we shall ultimately need to give meaning to the \textit{chronological product} $\mathcal{T}(\hat{\phi}_{i_1}...\hat{\phi}_{i_n})$ for arbitrary relative orientations of the space-time points associated with the indices. When the points associated with an index pair $i, j$ are separated by a spacelike interval there is no ambiguity in the chronological product because the supercommutator function $\hat{\tilde{G}}^{ij}$ then vanishes. Problems arise in the limit when two points coincide: it will be seen later that these problems are all resolved by requiring the $\mathcal{T}$-operation to commute with both differentiation and integration with respect to space-time coordinates: this requirement is equivalent to first imposing the linearity condition 
\begin{equation}
\mathcal{T}((\alpha\hat{\phi}^{i}+\beta\hat{\phi}^{j})\hat{\phi}^{k}\hat{\phi}^{l}...) = \alpha\mathcal{T}(\alpha\hat{\phi}^{i}\hat{\phi}^{k}\hat{\phi}^{l}...)+\beta\mathcal{T}(\hat{\phi}^{j}\hat{\phi}^{k}\hat{\phi}^{l}...) \label{T_requirements}
\end{equation}
for all $\alpha, \beta \in \Lambda_{\infty}$ and then requiring the $\mathcal{T}$-operation to commute with certain operations involving passage to a limit. 
\begin{remark}
	The above requirements have the consequence that an expression like $\langle {\rm out }|\mathcal{T}(\hat{S}_{,i}\hat{\phi}^{i}\hat{\phi}^{j})|{\rm in } \rangle$ does not generally vanish despite the fact that $\hat{S}_{,i}$ is zero when the operators that compose it are ordered appropriately for the operator dynamical equations.
\end{remark}

Now let $A[\phi]$ be any functional of the classical $\phi$ for which ${\rm supp }\; A_{,i}$ lies between the \lq\lq in'' and \lq\lq out'' regions, and which possesses a functional Taylor expansion about $\phi=0$ with a nonzero radius of convergence: 
\begin{equation}
A[\phi] = A[0]+A_{1}[0]\phi+\frac{1}{2}A_{2}[0]\phi\phi+... \; .
\end{equation}
It then follows from eq. \eqref{funct_der_n} and the requirements illustrated by eq. \eqref{T_requirements} that 
\begin{eqnarray}
\langle {\rm out }|\mathcal{T}(A[\hat{\phi}])|{\rm in } \rangle = (-1)^{AF} A\left[\frac{\overrightarrow{\delta}}{i\delta J}\right]\langle {\rm out }|{\rm in } \rangle &,& \label{def_TA1}\\
A\left[\frac{\overrightarrow{\delta}}{i\delta J}\right]  \equiv  A[0]+A_{1}[0]\frac{\overrightarrow{\delta}}{i\delta J}+\frac{1}{2}A_{2}[0]\frac{\overrightarrow{\delta}}{i\delta J}\frac{\overrightarrow{\delta}}{i\delta J}+ ... \; &.& \label{def_TA2}
\end{eqnarray}
The previous equations provide a way of associating an operator, i.e., $\mathcal{T}(A[\hat{\phi}])$, with each classical functional $A[\phi]$ having appropriate properties. If $A[\phi]$ has only a finite radius of convergence, then, strictly speaking, $\mathcal{T}(A[\hat{\phi}])$ is not defined by eqs. \eqref{def_TA1} and \eqref{def_TA2}, but it can often be given a meaning, for given \lq\lq in'' and \lq\lq out'' states, by analytic continuation. One can also frequently give $\mathcal{T}(A[\hat{\phi}])$ a meaning, even when $A[\phi]$ is singular at $\phi =0$, by expanding about a different point and continuing analytically. Equation \eqref{def_TA1} finds a wide variety of applications in quantum field theory. 

\subsection{The operator dynamical equations and the measure functional}
The association established by eq. \eqref{def_TA1}, between a classical functional and a quantum operator, suggests that when $A[\phi]$ is a classical observable the operator $\mathcal{T}(A[\hat{\phi}])$ might be taken as its quantum counterpart. This suggestion is often valid when $\mathcal{T}(A[\hat{\phi}])$ is self-adjoint. However, \textit{the chronological product is not self-adjoint in general, even when $A[\phi]$ is real because the operation of taking the adjoint reverses the order of all factors, placing them in antichronological order}.

In order for $\mathcal{T}(A[\hat{\phi}])$ to be self-adjoint $A[\phi]$ must usually be \textit{local}, i.e., built out of $\phi$'s and their derivatives all taken at the same space-time point. But this is not sufficient to guarantee the self-adjointness of $\mathcal{T}(A[\hat{\phi}])$. For example, $\mathcal{T}(S_{,i}[\hat{\phi}])$ is not generally self-adjoint (or anti-self-adjoint if the index $i$ is $a$-type rather than $c$-type). Hence the operator dynamical equations of the quantum theory should not be taken in the form $\mathcal{T}(S_{,i}[\hat{\phi}])=0$ or, when external sources are present, in the form $\mathcal{T}(S_{,i}[\hat{\phi}])=-J_{i}$. What we shall assume instead is the validity of the following postulate:

	{\it there exists a functional $\mu[\phi]$, determined by the classical action $S[\phi]$, such that the operator dynamical equations take the form} 
	\begin{equation}
	\mathcal{T}\left( \lbrace S[\hat{\phi}]-i\log \mu[\hat{\phi}]\rbrace \frac{\overleftarrow{\delta}}{\delta \phi^k}\right)=-J_{k} . \label{post_eq}
	\end{equation}
The functional $\mu[\phi]$ is known as the \textit{measure functional}. At the simplest level it may be thought of as correcting for the lack of self-adjointness or anti-self-adjointness of $\mathcal{T}(S_{,i}[\hat{\phi}])$. But it plays a far deeper role than this: once it is chosen the quantum theory is completely determined up to a finite-parameter family. Establishment of a correspondence between each classical theory and a unique quantum theory (or family of quantum theories) is therefore achieved by making $\mu[\phi]$ depend in a definite way on $S[\phi]$. How this dependence is itself to be chosen is a major question, which has to be approached in steps, but there is an easy argument that leads quickly to at least an approximate answer. 

In the quantum theory, as in the classical theory, it is often convenient to separate the dynamical variables $\hat{\phi}^i$ into a background $\phi^i_0$ and a remainder $\hat{\phi}^i_1$. If the background $\phi^i_0$ is classical, i.e., it is a pure supernumber-valued function times the identity operator, then the $\hat{\phi}^i_1$ satisfy the same supercommutation relations as the $\hat{\phi}^i$:
\begin{equation}
[\hat{\phi}^j_1,\hat{\phi}^k_1] = i \hat{\tilde{G}}^{jk}.
\end{equation} 
In terms of the $\hat{\phi}^i_1$ the operator dynamical equations take the form 
\begin{eqnarray}
& &S_{,i}[\phi_0] + S_{,ik}[\phi_0]\hat{\phi}^k_1 + \frac{1}{2}S_{,ikl}[\phi_0]\hat{\phi}^l _1\hat{\phi}^k _1 + O(\hbar ^2)+O(\hat{\phi}_1 ^3)=-J_{i} \label{taylor_din_eq}
\end{eqnarray}
The terms of order $0$ and $1$ in $\hat{\phi}_1$ on the left-hand side are unambiguous. Terms of higher order are not unambiguous since we do not yet know the complete factor-oredering rules. The coefficients of the higher order terms will generally \textit{not} be just the classical coefficients $S_{,ijk...}$. However, they will differ from the classical coefficients by terms of order $\hbar ^2$, which arise when the $\hat{\phi}_1$ 's are ordered as in \eqref{taylor_din_eq} and which will be dropped in the present approximate analysis. 

Ta\-ki\-ng t\-he su\-perco\-mm\-uta\-tor of \eqref{taylor_din_eq} wi\-th $\hat{\phi}^j_1$, re\-me\-mber\-ing that the s\-ou\-rce\-s $J_{i}$ are sup\-ern\-umber-valued, and moving the index $i$ to the left, one finds 
\begin{eqnarray}
& & S_{,ik}[\phi_0][\hat{\phi}^k_1,\hat{\phi}^j_1] + \frac{1}{2}S_{,ikl}[\phi_0][\hat{\phi}^l _1\hat{\phi}^k _1,\hat{\phi}^j_1] + ... = 0, \nonumber \\
& &i S_{,ik}[\phi_0] \hat{\tilde{G}}^{kj} - (-1)^{j(k+l)}\frac{1}{2}S_{,ikl}[\phi_0][\hat{\phi}^j_1,\hat{\phi}^l _1\hat{\phi}^k _1] + .... = 0, \nonumber \\
& &i S_{,ik}[\phi_0] \hat{\tilde{G}}^{kj} - (-1)^{j(k+l)}\frac{1}{2}S_{,ikl}[\phi_0] ( [\hat{\phi}^j_1,\hat{\phi}^l _1]\hat{\phi}^k _1 \nonumber \\ & & \; \; \; + (-1)^{jl}\hat{\phi}^l _1[\hat{\phi}^j_1,\hat{\phi}^k _1] )+... = 0, \nonumber \\
& &i S_{,ik}[\phi_0] \hat{\tilde{G}}^{kj} - (-1)^{j(k+l)}\frac{1}{2}S_{,ikl}[\phi_0] \left( i \hat{\tilde{G}}^{jl}\hat{\phi}^k _1 + i(-1)^{jl}\hat{\phi}^l _1  \hat{\tilde{G}}^{jk} \right)+... = 0, \nonumber \\ 
& &S_{,ik}[\phi_0] \hat{\tilde{G}}^{kj} - (-1)^{j(k+l)}\frac{1}{2}S_{,ikl} [\phi_0]\left(\hat{\tilde{G}}^{jl}\hat{\phi}^k _1 + (-1)^{jl}\hat{\phi}^l _1  \hat{\tilde{G}}^{jk} \right)+... = 0 ,\nonumber \\
& &\tensor[_{i,}]{S}{_{,k}}[\phi_0]\hat{\tilde{G}}^{kj}- (-1)^{j(k+l)}\frac{1}{2}\tensor[_{i,}]{S}{_{,kl}}[\phi_0]\hat{\tilde{G}}^{jl}\hat{\phi}^k _1 \nonumber\\& &\; \; \;- (-1)^{j(k+l)+jl}\frac{1}{2}\tensor[_{i,}]{S}{_{,kl}}[\phi_0]\hat{\phi}^l _1  \hat{\tilde{G}}^{jk} +... = 0, \nonumber \\
& &\tensor[_{i,}]{S}{_{,k}}[\phi_0]\hat{\tilde{G}}^{kj}+ (-1)^{jk}\frac{1}{2}\tensor[_{i,}]{S}{_{,kl}}[\phi_0]\hat{\tilde{G}}^{lj}\hat{\phi}^k _1 + \frac{1}{2}\tensor[_{i,}]{S}{_{,kl}}[\phi_0]\hat{\phi}^l _1  \hat{\tilde{G}}^{kj} +... = 0. \nonumber \\
\end{eqnarray}
The quantum supercommutator function, like the classical supercommutator function, can be expressed as the difference between an advanced Green's function and a retarded Green's function, both now operator-valued: 
\begin{equation}
\hat{\tilde{G}}^{ij} = \hat{G}^{+ij}-\hat{G}^{-ij},
\end{equation}
where
\begin{eqnarray}
& &\tensor[_{i,}]{S}{_{,k}}[\hat{\phi}]\hat{G}^{\pm kj} = - \tensor[_i]{\delta}{^j}, \nonumber \\
& &\tensor[_{i,}]{S}{_{,k}}[\phi_0]\hat{G}^{\pm kj}+ (-1)^{jk}\frac{1}{2}\tensor[_{i,}]{S}{_{,kl}}[\phi_0]\hat{G}^{\pm lj}\hat{\phi}^k _1 +\nonumber \\ & &\; \; \;+ \frac{1}{2}\tensor[_{i,}]{S}{_{,kl}}[\phi_0]\hat{\phi}^l _1  \hat{G}^{\pm kj} +... = - \tensor[_i]{\delta}{^j} .
\end{eqnarray} 
The previous equation, as is straightforward to verify, can be solved by iteration, yielding:
\begin{eqnarray}
\hat{G}^{\pm ij} &=& G^{\pm ij}[\phi_0]+\frac{1}{2}(-1)^{jk} G^{\pm im}[\phi_0] \; \tensor[_{m,}]{S}{_{,kl}} \;  G^{\pm lj}[\phi_0]\hat{\phi}^k _1 \nonumber \\ & &+ \frac{1}{2}G^{\pm im}[\phi_0] \; \tensor[_{m,}]{S}{_{,kl}} \; \hat{\phi}^l _1  G^{\pm kj}[\phi_0] + ... \nonumber \\
&=&  G^{\pm ij}[\phi_0] +  G^{\pm ij}_{,k}[\phi_0]\hat{\phi}^k _1 + ... \; .
\end{eqnarray}
Again it is not possible to specify the forms of the higher terms in the expansion. The coefficients of these terms will generally \textit{not} be just functional derivatives $G^{\pm ij}_{,kl...}[\phi_0]$ of the classical Green's functions. 

Introduce now the space-time generalization of the step function:
\begin{equation}
\theta(i,j') \equiv 
\begin{cases}
1 & {\rm if} \; \;  x^{0}>x'^{0} \\
\frac{1}{2 } & {\rm if} \; \; x^{0}=x'^{0} \\
0  & {\rm if} \; \; x^{0}<x'^{0} 
\end{cases}
\end{equation}
where $x^{0}$ is a global timelike coordinate and the submanifolds $x^{0} = {\rm constant}$ are complete spacelike Cauchy hypersurfaces. If the points $x$ and $x'$ are not in the immediate vicinity of one another, then 
\begin{eqnarray}
\hat{\phi}^i \hat{\phi}^{j'} - \mathcal{T}(\hat{\phi}^i \hat{\phi}^{j'}) &=& \left[\theta(i,j')+\theta(j',i)\right]\hat{\phi}^i \hat{\phi}^{j'}-\theta(i,j')\hat{\phi}^i \hat{\phi}^{j'} \nonumber \\ & &-(-1)^{ij'}\theta(j',i)\hat{\phi}^{j'}\hat{\phi}^i \nonumber \\
&=& \theta(j',i)\left[\hat{\phi}^i,\hat{\phi}^{j'} \right] \nonumber \\
&=& i\theta(j',i)\hat{\tilde{G}}^{ij'} \nonumber \\
&=& i\hat{G}^{+ij'}.
\end{eqnarray} 
We shall assume that this equation in fact holds for \textit{all} $x$, $x'$, at least up to the order needed in our analysis of the expanded dynamical equations. 
Equation \eqref{taylor_din_eq} may then be written in the form 
\begin{eqnarray}
-J_{l} &=& S_{,l}[\phi_0] + S_{,lk}[\phi_0]\hat{\phi}^k_1 + \frac{1}{2}\tensor{S}{_{,ljk}}[\phi_0]\left[ \mathcal{T}(\hat{\phi}^k \hat{\phi}^{j})+i\hat{G}^{+kj}\right] + ... \nonumber \\
&=&  \mathcal{T}(S_{,l}[\hat{\phi}]+ \frac{1}{2}i\tensor{S}{_{,ljk}}[\hat{\phi}]G^{+kj}[\hat{\phi}]+...),
\end{eqnarray}
where the dots in the second line stands not only for the unwritten terms in the first line but also for the error in replacing $\tensor{S}{_{ijk}}[\phi_0]\hat{G}^{+kj}$ by $\mathcal{T}( \tensor{S}{_{,ijk}}[\hat{\phi}]G^{+kj}[\hat{\phi})$. The previous expression can be simplified by recalling that $\tensor[_1]{S}{_1}$ is a negative inverse of  $G^{+}$; therefore:
\begin{eqnarray}
\tensor{S}{_{,ijk}}\;\tensor{G}{^+^{kj}} &=& -\tensor{S}{_{,jk}}\;\tensor{G}{^+^{kj}_{,i}} +(\tensor{S}{_{,jk}}\;\tensor{G}{^+^{kj}})_{,i} \nonumber \\
&=& -\tensor{S}{_{,jk}}\;\tensor{G}{^+^{kj}_{,i}} +(-\tensor{\delta}{_{k}^{k}})_{,i} \nonumber \\
&=& -\tensor{S}{_{,jk}}\;\tensor{G}{^+^{kj}_{,i}} \nonumber \\
&=& -(-1)^j\tensor[_{j,}]{S}{_{,k}}\;\tensor{G}{^+^{kj}_{,i}} \nonumber \\
&=& \log |{\rm sdet} G^{+}|)_{,i}.
\end{eqnarray}
Thus
\begin{equation}
-J_{k} = \mathcal{T}(S_{,k}[\hat{\phi}]+ \frac{1}{2}i(\log|{\rm sdet} G^{+}[\hat{\phi}]|)_{,k} + ...). \label{appr_post_eq}
\end{equation}
Comparison of eqs. \eqref{post_eq} and \eqref{appr_post_eq} yields 
\begin{equation}
\mu[\phi] \approx {\rm const } \cdot |{\rm sdet} G^{+}[\phi]|^{-\frac{1}{2}}. \label{appr_meas}
\end{equation}

\subsection{Functional Fourier analysis. The Feynman functional integral}
The transition amplitude $\langle {\rm out }|{\rm in } \rangle$ is a functional of the external sources. Let us try to express it as a functional Fourier integral: 
\begin{eqnarray}
& &\langle {\rm out }|{\rm in } \rangle = \int X[\phi] e^{iJ\phi} [d\phi] = \int e^{iJ\phi} X[\phi] [d\phi], \\
& &[d\phi] \equiv \prod_{i} d\phi^i . \label{field_measure}
\end{eqnarray}
Here the $\phi^i$ are supernumber-valued variables of integration, and the product in eq. \eqref{field_measure} is a continuous infinite one. Both it and the integral itself are thus formal expressions. Integration over the $c$-number variables is to be understood as patterned on ordinary integration. Integration over the $a$-number variables is to be understood as an infinite limit of a multiple Berezin integral (see \cite{berazin2012method} and \cite{dewitt1992supermanifolds}). The integral is to be understood as taken over a certain subspace of the space of field histories $\Phi$, whose properties will be indicated below. 
Assuming the validity of integrating by parts, and making use of eqs. \eqref{def_TA1} and \eqref{post_eq}, one may write 
\begin{eqnarray}
& &\int X[\phi] \frac{\overleftarrow{\delta}}{i\delta\phi^k}e^{iJ\phi} [d\phi] \nonumber = \\ 
&=& - \int X (e^{iJ\phi} \frac{\overleftarrow{\delta}}{i\delta\phi^k}) [d\phi]  \nonumber\\
&=& - \int X (e^{iJ\phi} \frac{\overleftarrow{\delta}}{i\delta\phi^k}) [d\phi]  \nonumber\\
&=& - \int X e^{iJ\phi} J_{k} [d\phi]  \nonumber\\
&=& - \int X J_{k} e^{iJ\phi}  [d\phi]  \nonumber\\
&=& -(-1)^{kX} J_{k}\int X  e^{iJ\phi}  [d\phi]  \nonumber\\
&=& -(-1)^{kX} J_{k}\int X  e^{iJ\phi}  [d\phi]  \nonumber\\
&=& -(-1)^{kX} J_{k}\langle {\rm out }|{\rm in } \rangle \nonumber\\
&=& (-1)^{k(X+F)} \langle {\rm out }|- J_{k}|{\rm in } \rangle \nonumber\\ 
&=& (-1)^{k(X+F)} \langle {\rm out }|\mathcal{T}\left( \lbrace S[\hat{\phi}]-i\log\mu[\hat{\phi}]\rbrace \frac{\overleftarrow{\delta}}{\delta \phi^k}\right)|{\rm in } \rangle \nonumber\\ 
&=& (-1)^{kX}\lbrace S_{,k}\left[\overrightarrow{\delta}/i\delta J\right]-i\mu^{-1}\left[\overrightarrow{\delta}/i\delta J\right]\mu_{,k}\left[\overrightarrow{\delta}/i\delta J\right]\rbrace \langle {\rm out }|{\rm in } \rangle  \nonumber\\
&=& (-1)^{kX}\lbrace S_{,k}\left[\overrightarrow{\delta}/i\delta J\right]-i\mu^{-1}\left[\overrightarrow{\delta}/i\delta J\right]\mu_{,k}\left[\overrightarrow{\delta}/i\delta J\right]\rbrace \int e^{iJ\phi} X[\phi] [d\phi]  \nonumber\\
&=& (-1)^{kX}\int \lbrace S[\phi]-i(\log\mu[\phi])\rbrace\frac{\overleftarrow{\delta}}{\delta\phi^k}  e^{iJ\phi} X[\phi] [d\phi] .
\end{eqnarray} 
Because of the uniqueness of Fourier integral representations the integrands in the first and last lines must be equal: 
\begin{equation}
X[\phi] \frac{\overleftarrow{\delta}}{i\delta\phi^k} = (-1)^{kX}\lbrace S[\phi]-i(\log\mu[\phi])\rbrace\frac{\overleftarrow{\delta}}{\delta\phi^k} X[\phi].
\end{equation}
One possible solution of this equation is 
\begin{equation}
X[\phi] = N e^{iS[\phi]}\mu[\phi],
\end{equation}
where $N$ is a constant of integration. This leads to 
\begin{equation}
\langle {\rm out }|{\rm in } \rangle = N \int  e^{i(S[\phi]+J\phi)}\mu[\phi] [d\phi]. \label{path_int}
\end{equation}

\subsection{The Schwinger variational principle revisited}
Expression \eqref{path_int}, when combined with eq. \eqref{def_TA1}, yields immediately a functional integral expression for the \lq\lq in-out'' matrix elements of chronological products: 
\begin{equation}
\langle {\rm out }|\mathcal{T}(A[\hat{\phi}])|{\rm in } \rangle = (-1)^{A(F+N)}N \int  A[\phi]e^{i(S[\phi]+J\phi)}\mu[\phi] [d\phi] . \label{def_TA3}
\end{equation}

The previous equation can be used to obtain a partial check on the consistency of the Feynman integral with the Schwinger variational principle which was used to derive it in the first place. Under a variation $\delta S$ in the functional form of the action, such that ${\rm supp }\;\delta S_{,i}$ lies in the \lq\lq in between'' region, eqs. \eqref{path_int} and \eqref{def_TA3} yield 
\begin{eqnarray}
\delta \langle {\rm out }|{\rm in } \rangle &=& N \int \left( i\delta S[\phi] e^{i(S[\phi]+J\phi)}\mu[\phi] + e^{i(S[\phi]+J\phi)} \delta \mu[\phi] \right) [d\phi] \nonumber \\
&=& N \int \left( i\delta S[\phi] e^{i(S[\phi]+J\phi)}\mu[\phi] + e^{i(S[\phi]+J\phi)}\mu[\phi] \delta \log \mu[\phi] \right) [d\phi]  \nonumber \\
&=& iN \int \left( \delta S[\phi]  -i \delta \log\mu[\phi]  \right)e^{i(S[\phi]+J\phi)}\mu[\phi] \nonumber \\
&=& i \langle {\rm out }|\mathcal{T}( \delta S[\hat{\phi}]  -i \delta \log\mu[\hat{\phi}] )|{\rm in } \rangle .
\end{eqnarray}
But from \eqref{appr_meas} one has 
\begin{eqnarray}
\delta \log \mu[\phi] &=& \mu^{-1}[\phi] \delta \mu[\phi] \nonumber \\
&\approx & |{\rm sdet} G^{+}[\phi]|^{\frac{1}{2}} \delta |{\rm sdet} G^{+}[\phi]|^{-\frac{1}{2}}  \nonumber \\
&=& |{\rm sdet} G^{+}[\phi]|^{\frac{1}{2}} \left(-\frac{1}{2}\right)|{\rm sdet} G^{+}[\phi]|^{-\frac{3}{2}}  \delta |{\rm sdet} G^{+}[\phi]|  \nonumber \\
&=& |{\rm sdet} G^{+}[\phi]|^{\frac{1}{2}} \left(-\frac{1}{2}\right)|{\rm sdet} G^{+}[\phi]|^{-\frac{3}{2}} |{\rm sdet} G^{+}[\phi]| {\rm str }( G^{+}[\phi] \delta\tensor[_1]{S}{_1}[\phi])\nonumber \\
&=& -\frac{1}{2} {\rm str }( G^{+}[\phi] \delta\tensor[_1]{S}{_1}[\phi]).
\end{eqnarray}
Hence 
\begin{equation}
\delta \langle {\rm out }|{\rm in } \rangle =  i \langle {\rm out }|\mathcal{T}( \delta S[\hat{\phi}]  + \frac{i}{2} {\rm str }( G^{+}[\hat{\phi}] \delta\tensor[_1]{S}{_1}[\hat{\phi}])+... )|{\rm in } \rangle.
\end{equation}
By the same kind of rearrangement as was used in obtaining eq. \eqref{appr_post_eq} one easily sees that the chronological product is, at least approximately, just $\delta S[\hat{\phi}]$ in its self-adjoint operator form. 

\subsection{Expressions involving $S_{,i}$}
Let $A[\phi]$ be an arbitrary functional of $\phi$ such that the support of $A_{,j}$ lies between the \lq\lq in'' and \lq\lq out'' regions. Since the functional integral respects the procedure of integration by parts, the following identity holds: 
\begin{eqnarray}
0 &=& (-1)^{(A+i)(F+N)}N \int [d\phi]\left(A[\phi]\mu[\phi]e^{i(S[\phi]+J\phi)} \right) \frac{\overleftarrow{\delta}}{\delta\phi^k} \nonumber \\
&=& \langle {\rm out }|\mathcal{T}(A_{,k}[\hat{\phi}]+A[\hat{\phi}](\log\mu)_{,k}[\hat{\phi}]+iA[\hat{\phi}] S_{,k}[\hat{\phi}]   )    |{\rm in } \rangle .
\end{eqnarray}
Since the \lq\lq in'' and \lq\lq out'' state supervectors may be chosen arbitrarily this is in fact a statement about chronological products of operators involving $\hat{S}_{,i}$ when the sources vanish: 
\begin{equation}
\mathcal{T}\left(    (S[\hat{\phi}] -i \log\mu [\hat{\phi}]         )\frac{\overleftarrow{\delta}}{\delta \phi^k}  A[\hat{\phi}]  \right) = i(-1)^{iA}\mathcal{T}\left(A_{,k}[\hat{\phi}] \right) .
\end{equation}
This expression does not generally vanish despite eq. \eqref{post_eq}.

\section{Functional Integration in QFT with Invariance Flows}
\subsection{Structure of the space of field histories}
The task of extending the previous treatment to gauge theories, i.e., field theories whose action functional features invariance flows, requires a basic understanding of the properties of the space $\Phi$ of field histories, its structure and its geometry. We shall only focus on Type-I theories in which the coefficients $\tensor[]{c}{^\alpha _\beta _\gamma}$ are structure constants of an infinite dimensional Lie group, the gauge group, which we shall denote by $\mathcal{G}$. As already stated, $\Phi$ may be viewed as a principal bundle having $\mathcal{G}$ as its typical fibre. Real physics takes place in the base space of this bundle, i.e., the space of orbits (or fibres), denoted by $\Phi/\mathcal{G}$.

Since $\mathcal{G}$ is a group manifold it admits an invariant Riemannian or pseudo-Riemannian metric. This metric can be extended in an infinity of ways to a group (or flow) invariant metric on $\Phi$. But it turns out that if one requires the extended metric to be \textit{ultralocal} (a requirement that greatly simplifies the analysis of any formalism in which it is used) then, up to a scale factor, it is unique in the case of the Yang-Mills field and belongs to a one-parameter family in the case of gravity, these being the two primary gauge systems of interest. 

Let us denote this metric tensor by $\gamma$ and its components in the chart specified by the dynamical variables $\phi^i$ by $\tensor[_i]{\gamma}{_j}$. Let $x$ and $x'$ be the space-time points specified by the indices $i$ and $j$. Ultralocality of $\gamma$ is the condition that $\tensor[_i]{\gamma}{_j}$ be equal to the undifferentiated delta distributon $\delta(x, x')$ times a coefficient that involves no space-time derivatives of fields. Group invariance of $\gamma$ is the statement
\begin{equation}
\mathcal{L}_{Q_{\alpha}}\gamma=0 .
\end{equation} 

The $Q_{\alpha}$ are Killing vectors for the metric $\gamma$ and vertical vector fields for the principal bundle $\Phi$.
In the following, we shall rely on the following convention: in gauge theories the indices from the first part of the Greek alphabet are always $c$-type. Therefore 
they need never appear in exponents of $(-1)$. Indices of type $a$ from the middle of the Latin alphabet refer to fermion fields that may be coupled to the basic gauge fields. 

Choice of an invariant metric on $\Phi$ immediately singles out a natural family of connection 1-forms $\omega^\alpha$ on $\Phi$: 
\begin{equation}
\tensor[_i]{\omega}{^\alpha} \equiv \tensor[_i]{Q}{_\beta} \; \mathcal{N}^{\beta\alpha},
\end{equation}
where $\tensor[_i]{Q}{_\alpha} = \tensor[_i]{\gamma}{_j}  \; \tensor[^j]{Q}{_\alpha}$ and $\mathcal{N}^{\beta\alpha}$ is any coherent Green's function of the real self-adjoint (and hence symmetric) operator 
\begin{equation}
\mathcal{M}_{\alpha\beta} \equiv - \tensor[^i]{Q}{_\alpha} \; \tensor[_i]{\gamma}{_j} \; \tensor[^j]{Q}{_\beta},
\end{equation}
which, in all cases of interest, turns out to be globally nonsingular, i.e., over the whole of $\Phi$. In view of the relation $\mathcal{M}_{\alpha\beta}\mathcal{N}^{\beta\gamma} = -\tensor{\delta}{_{\alpha}^{\gamma}}$ one easily sees that 
\begin{eqnarray}
\tensor[^i]{Q}{_\alpha} \; \tensor[_i]{\omega}{^\beta} &=& \tensor[^i]{Q}{_\alpha} \; \tensor[_i]{Q}{_\gamma} \; \mathcal{N}^{\gamma\beta} \nonumber\\
&=& \tensor[^i]{Q}{_\alpha} \; \tensor[_i]{\gamma}{_j}  \; \tensor[^j]{Q}{_\gamma} \; \mathcal{N}^{\gamma\beta}  \nonumber \\
&=& -\mathcal{M}_{\alpha\gamma} \; \mathcal{N}^{\gamma\beta}  \nonumber\\
&=& \tensor{\delta}{_{\alpha}^{\beta}} ,
\end{eqnarray}
and that horizontal vectors on $\Phi$ are those that are perpendicular (under the metric $\gamma$) to the fibres. A horizontal vector may be obtained from any vector by application of the \textit{horizontal projection operator}:
\begin{equation}
\tensor{\Pi}{^i_j} \equiv \tensor{\delta}{^i_j}-\tensor[^i]{Q}{_\alpha}\; \tensor{\omega}{^\alpha_j}, \; \; \;  \tensor{\omega}{^\alpha_i} \equiv (-1)^i \tensor[_i]{\omega}{^\alpha} .
\end{equation}
The name is justified by the following properties:
\begin{eqnarray}
\tensor{\Pi}{^i_j} \; \tensor{\Pi}{^j_k} &=&(\tensor{\delta}{^i_j}-\tensor[^i]{Q}{_\alpha}\; \tensor{\omega}{^\alpha_j})(\tensor{\delta}{^j_k}-\tensor[^j]{Q}{_\alpha}\; \tensor{\omega}{^\alpha_k}) \nonumber \\
&=& \tensor{\delta}{^i_k} -\tensor[^i]{Q}{_\alpha}\; \tensor{\omega}{^\alpha_k}-\tensor[^i]{Q}{_\alpha}\; \tensor{\omega}{^\alpha_k}+\tensor[^i]{Q}{_\alpha}\; \tensor{\omega}{^\alpha_j} \;\tensor[^j]{Q}{_\beta}\; \tensor{\omega}{^\beta_k} \nonumber \\
&=& \tensor{\delta}{^i_k} -2\;\tensor[^i]{Q}{_\alpha}\; \tensor{\omega}{^\alpha_k} +(-1)^j \; \tensor[^i]{Q}{_\alpha}\; \tensor[_j]{Q}{_\gamma}\; \mathcal{N}^{\gamma\alpha}   \;\tensor[^j]{Q}{_\beta}\; \tensor{\omega}{^\beta_k} \nonumber \\
&=& \tensor{\delta}{^i_k} -2\;\tensor[^i]{Q}{_\alpha}\; \tensor{\omega}{^\alpha_k} +(-1)^j \; \tensor[^i]{Q}{_\alpha}\; \tensor[_j]{\gamma}{_l} \; \tensor[^l]{Q}{_\gamma}\;  \mathcal{N}^{\gamma\alpha}   \;\tensor[^j]{Q}{_\beta}\; \tensor{\omega}{^\beta_k} \nonumber \\
&=& \tensor{\delta}{^i_k} -2\;\tensor[^i]{Q}{_\alpha}\; \tensor{\omega}{^\alpha_k} +(-1)^{j+j} \; \tensor[^i]{Q}{_\alpha}\; \tensor[^j]{Q}{_\beta} \;\tensor[_j]{\gamma}{_l} \; \tensor[^l]{Q}{_\gamma}\;  \mathcal{N}^{\gamma\alpha}  \; \tensor{\omega}{^\beta_k} \nonumber \\
&=& \tensor{\delta}{^i_k} -2\;\tensor[^i]{Q}{_\alpha}\; \tensor{\omega}{^\alpha_k} - \tensor[^i]{Q}{_\alpha} \; \mathcal{M}_{\beta\gamma}\;  \mathcal{N}^{\gamma\alpha}  \; \tensor{\omega}{^\beta_k} \nonumber \\
&=& \tensor{\delta}{^i_k} -2\;\tensor[^i]{Q}{_\alpha}\; \tensor{\omega}{^\alpha_k} + \tensor[^i]{Q}{_\alpha} \; \tensor[]{\delta}{_\beta^\alpha}  \; \tensor{\omega}{^\beta_k} \nonumber \\
&=& \tensor{\delta}{^i_k} -2\;\tensor[^i]{Q}{_\alpha}\; \tensor{\omega}{^\alpha_k} + \tensor[^i]{Q}{_\alpha} \; \tensor{\omega}{^\alpha_k} \nonumber \\
&=& \tensor{\delta}{^i_k} -\tensor[^i]{Q}{_\alpha}\; \tensor{\omega}{^\alpha_k} \nonumber \\
&=& \tensor{\Pi}{^i_k}, \\
\nonumber 
\end{eqnarray}
\begin{eqnarray}
\tensor{\omega}{^\alpha_i}\; \tensor{\Pi}{^i_j} &=& \tensor{\omega}{^\alpha_i} \; (\tensor{\delta}{^i_j}-\tensor[^i]{Q}{_\beta}\; \tensor{\omega}{^\beta_j}) \nonumber \\
&=& \tensor{\omega}{^\alpha_j} - \tensor{\omega}{^\alpha_i} \; \tensor[^i]{Q}{_\beta}\; \tensor{\omega}{^\beta_j} \nonumber \\
&=& \tensor{\omega}{^\alpha_j} - (-1)^{i+j} \; \tensor[_i]{\omega}{^\alpha} \; \tensor[^i]{Q}{_\beta}\;  \tensor[_j]{\omega}{^\beta}  \nonumber \\
&=& \tensor{\omega}{^\alpha_j} - (-1)^{j} \; \tensor[^i]{Q}{_\beta} \; \tensor[_i]{\omega}{^\alpha} \; \tensor[_j]{\omega}{^\beta}  \nonumber \\
&=& \tensor{\omega}{^\alpha_j} - (-1)^{j} \; \tensor[]{\delta}{_\beta^\alpha} \; \tensor[_j]{\omega}{^\beta}  \nonumber \\
&=& \tensor{\omega}{^\alpha_j} - (-1)^{j} \; \tensor[_j]{\omega}{^\alpha}  \nonumber \\
&=& \tensor{\omega}{^\alpha_j} - \tensor[]{\omega}{^\alpha_j}  \nonumber \\
&=& 0, \\
\nonumber 
\end{eqnarray}
\begin{eqnarray}
\tensor{\Pi}{^i_j} \; \tensor[^j]{Q}{_\alpha} &=& (\tensor{\delta}{^i_j}-\tensor[^i]{Q}{_\beta}\; \tensor{\omega}{^\beta_j}) \; \tensor[^j]{Q}{_\alpha} \nonumber \\
&=& \tensor[^i]{Q}{_\alpha} - \tensor[^i]{Q}{_\beta}\; \tensor{\omega}{^\beta_j} \tensor[^j]{Q}{_\alpha} \nonumber \\
&=& \tensor[^i]{Q}{_\alpha} - (-1)^{j} \; \tensor[^i]{Q}{_\beta}\; \tensor[_j]{\omega}{^\beta} \; \tensor[^j]{Q}{_\alpha} \nonumber \\
&=& \tensor[^i]{Q}{_\alpha} - (-1)^{j+j} \; \tensor[^i]{Q}{_\beta}\; \tensor[^j]{Q}{_\beta} \; \tensor[_j]{\omega}{^\alpha}  \nonumber \\
&=& \tensor[^i]{Q}{_\alpha} -  \tensor[^i]{Q}{_\beta}\; \tensor[]{\delta}{_\alpha^\beta} \nonumber \\
&=& \tensor[^i]{Q}{_\alpha} -  \tensor[^i]{Q}{_\alpha} \nonumber \\
&=& 0 .
\end{eqnarray}
\subsection{Fibre-adapted coordinate patches}
When dealing with $\Phi$ it is convenient to consider a transformation 
\begin{equation}
\phi^i \rightarrow I^A,K^\alpha
\end{equation}
to a set of \textit{fibre-adapted coordinates} $I^A,K^\alpha$. Here the $I$ 's label the fibres (i.e., the points in $\Phi/\mathcal{G}$) and are gauge invariant (i.e., flow invariant): 
\begin{equation}
I^A \tensor{Q}{_\alpha} = \tensor{I}{^A_{,i}} \; \tensor[^i]{Q}{_\alpha} = 0 . \label{I_flow_inv}
\end{equation}
The $K$'s label the points within each fibre. Because there is no canonical way of associating points on one fibre with those on another, transformations between fibre-adapted coordinate patches have the general structure 
\begin{equation}
I^{'A}=I^{'A}[I],  \; \; \; K^{'\alpha}=K^{'\alpha}[I,K],
\end{equation}
which is still special enough so that the Jacobian of the transformation splits into factors: 
\begin{equation}
\frac{\delta(I',K')}{\delta(I,K)} = \frac{\delta(I')}{\delta(I)} \frac{\delta(K')}{\delta(K)}.
\end{equation}
One often makes specific choices for the $K$'s. One usually singles out a \textit{base point} $\phi_*$ in $\Phi$ and chooses the $K$'s to be local functionals of the $\phi$'s of such a form that the matrix 
\begin{equation}
\underline{\mathcal{M}}^{\alpha}_{\beta} \equiv K^\alpha Q_\beta = \tensor{K}{^\alpha_{,i}} \; \tensor[^i]{Q}{_\beta} \label{def_underl_M}
\end{equation}
i\-s a n\-ons\-ing\-ula\-r dif\-fer\-enti\-al ope\-rat\-or at and in a ne\-igh\-bor\-hoo\-d of $\phi_*$. Ty\-pica\-l con\-ven\-ient ch\-oice\-s for $\phi_*$ are $A^{\alpha}_{\mu *}(x)=0$ in pu\-re Ya\-ng-Mi\-lls th\-eor\-y and $g_{\mu\nu *}(x) =$ so\-me we\-ll st\-udie\-d ba\-ckg\-rou\-nd me\-tr\-ic (Mi\-nko\-wsk\-i, Fri\-edm\-ann-Le\-maître-Robertson-Walker, black hole, etc.) in pure gravity theory. 

One can also make specific choices for the $I$'s, but in general such invariants depend nonlocally on the $\phi$'s and are clumsy to work with; no specific choices will be made here, so in what follows the $I$'s will remain purely conceptual.

In the region of $\Phi$ where the operator $\underline{\mathcal{M}}$ is nonsingular it is easy to show that
\begin{equation}
\frac{\overleftarrow{\delta}}{\delta K^\alpha} = - Q_{\beta} \underline{\mathcal{N}}^\beta _\alpha, \label{K_field}
\end{equation} 
where $\underline{\mathcal{N}}$ is a Green's function of $\underline{\mathcal{M}}$. In fact, by multiplying \eqref{def_underl_M} first on the right by $\underline{\mathcal{N}}$ and then on the left by $\frac{\overleftarrow{\delta}}{\delta K}$, and using \eqref{I_flow_inv}, one obtains
\begin{eqnarray}
\underline{\mathcal{M}}^{\alpha}_{\beta} \underline{\mathcal{N}}^{\beta}_{\gamma} &=& \tensor{K}{^\alpha_{,i}} \tensor[^i]{Q}{_\beta} \; \underline{\mathcal{N}}^{\beta}_{\gamma} , \nonumber \\
- \tensor{\delta}{^\alpha _\gamma} &=& \tensor{K}{^\alpha_{,i}} \tensor[^i]{Q}{_\beta} \; \underline{\mathcal{N}}^{\beta}_{\gamma}, \nonumber \\
- \frac{\overleftarrow{\delta}}{\delta K^\alpha} \; \tensor{\delta}{^\alpha _\gamma} &=&  \frac{\overleftarrow{\delta}}{\delta K^\alpha} \; \tensor{K}{^\alpha_{,i}} \; \tensor[^i]{Q}{_\beta} \; \underline{\mathcal{N}}^{\beta}_{\gamma}, \nonumber \\
- \frac{\overleftarrow{\delta}}{\delta K^\gamma} &=& \frac{\overleftarrow{\delta}}{\delta K^\alpha} \; K^\alpha \frac{\overleftarrow{\delta}}{\delta \phi ^i} \;  \tensor[^i]{Q}{_\beta} \; \underline{\mathcal{N}}^{\beta}_{\gamma}, \nonumber \\
- \frac{\overleftarrow{\delta}}{\delta K^\gamma} &=& \left(\frac{\overleftarrow{\delta}}{\delta K^\alpha} \; K^\alpha \frac{\overleftarrow{\delta}}{\delta \phi ^i}+\frac{\overleftarrow{\delta}}{\delta I^A} \; I^A \frac{\overleftarrow{\delta}}{\delta \phi ^i} \right)\;  \tensor[^i]{Q}{_\beta} \; \underline{\mathcal{N}}^{\beta}_{\gamma}, \nonumber \\
- \frac{\overleftarrow{\delta}}{\delta K^\gamma} &=&  \frac{\overleftarrow{\delta}}{\delta \phi ^i} \;  \tensor[^i]{Q}{_\beta} \; \underline{\mathcal{N}}^{\beta}_{\gamma}, \nonumber \\
- \frac{\overleftarrow{\delta}}{\delta K^\gamma} &=& Q_{\beta} \underline{\mathcal{N}}^{\beta}_{\gamma}.
\end{eqnarray}
It is important to stress that when $\mathcal{G}$ is non-Abelian it is impossible for the $K^\alpha$ to be valid coordinates globally: in fact if they were, then the $\overleftarrow{\delta}/\delta K^\alpha$, which are vertical vector fields that commute with each other, would generate Abelian orbits (fibres). This means that $\underline{\mathcal{M}}$, unlike $\mathcal{M}$, cannot be nonsingular globally on $\Phi$.

\subsection{Functional integration for \lq\lq in-out'' amplitudes}
Since the physics of a gauge theory takes place in the base space $\Phi/\mathcal{G}$ it is natural to try to write a functional integral for \lq\lq in-out'' amplitudes in the form 
\begin{equation}
\langle {\rm out }|{\rm in } \rangle = \int \mu_I [I] [dI] e^{iS[I]}. \label{path_int_flows}
\end{equation}
All functional integrals encountered in previous sections were purely formal. Equation \eqref{path_int_flows} is even \textit{more} formal, for the following reasons: 
\begin{itemize}
	\item[1.] The labels $I^A$ are not chosen explicitly but used only conceptually. 
	\item[2.] Since all known usable explicit choices depend nonlocally on the $\phi$'s it is hard to know what one can mean by an advanced Green's function of the Jacobi field operator $\tensor[_{A,}]{S}{_{,B}}$ (or its superdeterminant) and hence how to determine the measure $\mu_I [I]$ even approximately. 
	\item[3.] It is also hard to know how to set boundary conditions.
\end{itemize}
To bring the local variables $\phi^i$ into the theory one must first introduce the remaining variables $K^\alpha$ of a fibre-adapted coordinate system and then transform to the $\phi$'s. Let $\Omega[I,K]$ be a real scalar function(al) on $\Phi$ such that the integral 
\begin{equation}
\Delta[I] \equiv \int e^{i\Omega[I,K]} \mu_K[I,K] [dK] \label{def_Delta}
\end{equation}
exists and is nonvanishing for all $I$, the measure  $\mu_K[I,K]$ being assumed to transform under changes (generally $I$-dependent) of the fibre-adapted coordinates $K^\alpha$ according to
\begin{equation}
\mu_{K'}[I,K'] = \mu_K[I,K] \frac{\delta K}{\delta K'}.
\end{equation}
Then $\Delta[I]$ is invariant under such coordinate changes and one may write:
\begin{equation}
\langle {\rm out }|{\rm in } \rangle = \int [dI] \int [dK]   e^{i(S[I]+\Omega[I,K])}\Delta[I]^{-1} \mu_{I,K} [I,K] , \label{path_int_flows1}
\end{equation}
where 
\begin{equation}
\mu_{I,K} [I,K] \equiv  \mu_I [I]\mu_K[I,K].
\end{equation}
In a similar way one may write the analog of eq. \eqref{def_TA1} in the forms 
\begin{eqnarray}
\langle {\rm out }|\mathcal{T}(A[I])| {\rm in } \rangle &=& \int \mu_I [I] [dI] A[I] e^{iS[I]} \\
&=& \int [dI] \int [dK]  \nonumber \\
& & \; A[I] e^{i(S[I]+\Omega[I,K])}\Delta[I]^{-1} \mu_{I,K} [I,K] . \label{path_int_flows2}
\end{eqnarray}

Under changes of fibre-adapted coordinates the measure $\mu_I [I]$ must obviously transform according to 
\begin{equation}
\mu_{I'} [I'] = \mu_I [I] \frac{\delta I}{\delta I'},
\end{equation}
and hence the total measure $\mu_{I,K} [I,K]$ transforms as it should: 
\begin{equation}
\mu_{I',K'} [I',K'] = \mu_{I,K} [I,K] \frac{\delta I}{\delta I'}  \frac{\delta K}{\delta K'} = \mu_{I,K} [I,K] \frac{\delta (I,K)}{\delta (I',K')} .
\end{equation}

To make the transformation from the $I^A, K^\alpha$ to the local coordinates $\phi^i$ one must include also the formal Jacobian 
\begin{equation}
J[\phi] \equiv \frac{\delta (I,K)}{\delta\phi} = {\rm sdet} \left(\begin{array}{c} I^A_{,i} \\ K^\alpha_{,i} \end{array} \right).
\end{equation}

Then the functional integrals \eqref{path_int_flows1} and \eqref{path_int_flows2} take the forms:
\begin{eqnarray}
\langle {\rm out }|{\rm in } \rangle &=& \int [d\phi]  e^{i(S[\phi]+\Omega[\phi])}\Delta[\phi]^{-1} J[\phi] \mu_{I,K} [\phi], \label{path_int_flows2bis} \\
\langle {\rm out }|\mathcal{T}(A[\phi])| {\rm in } \rangle &=&  \int [d\phi] A[\phi]e^{i(S[\phi]+\Omega[\phi])} \bullet \nonumber \\ & & \; \bullet \; \Delta[\phi]^{-1} J[\phi] \mu_{I,K} [\phi], \label{path_int_flows3}
\end{eqnarray}
in which we have abused notation somewhat by simply writing $\Delta[I]=\Delta[\phi]$, $S[I]=S[\phi]$, $\mu_{I,K} [I]=\mu_{I,K} [\phi]$, $\Omega[I,K]=\Omega[\phi]$ and $A[I]=A[\phi]$. The last abuse in fact allows a certain generalization of the formalism. In eqs. \eqref{def_TA1}, \eqref{def_TA2} the functional $A$ was an invariant, i.e., a physical observable. The integral \eqref{path_int_flows3} may be regarded as a generalized average which can give meaning to $\langle {\rm out }|\mathcal{T}(A[\phi])| {\rm in } \rangle$ even when $A$ is not gauge invariant. True physical amplitudes, of course, only involve $A$'s that are gauge invariant. Note that when (and only when) $A$ is gauge invariant the average \eqref{path_int_flows3} is completely independent of the choice of the functional $\Omega[\phi]$.

\subsection{Properties of the jacobian $J[\phi]$}
Suppose we carry out an infinitesimal transformation of the fibre-adapted coordinates $K^\alpha$:
\begin{equation}
K^{'\alpha}= K^{\alpha}+\delta K^{\alpha}[I,K] .
\end{equation} 
Formally this will produce the following change in the Jacobian $J[\phi]$: 
\begin{eqnarray}
\delta J[\phi] &=& \delta \; {\rm sdet } \left(\begin{array}{c} I^A_{,i} \\ K^\alpha_{,i} \end{array} \right) \nonumber \\
&=& J \; { \rm str } \left[\left(\begin{array}{c}  I^A_{,j} \\  K^\alpha_{,j} \end{array} \right)^{-1} \left(\begin{array}{c} \delta I^A_{,i} \\ \delta K^\alpha_{,i} \end{array} \right)\right]  \nonumber \\
&=& J \; { \rm str } \left[\left(\begin{array}{c} \phi^j_{,A}\\ \phi^j_{,\alpha} \end{array} \right) \left(\begin{array}{c} 0 \\ \delta K^\alpha_{,i} \end{array} \right)\right]  \nonumber \\
&=& (-1)^i J \phi^{i}_{,\alpha} \; \delta K^\alpha_{,i},
\end{eqnarray}
Here we encounter an immediate problem: the meaning to be given to the factor $\phi^{i}_{,\alpha}$. If we apply the operator \eqref{K_field} to the fields $\phi^{i}$ we get 
\begin{equation}
\phi^{i}_{,\alpha} = - \tensor[^i]{Q}{_\beta} \; \underline{\mathcal{N}}^{\beta}_{\alpha} ,
\end{equation}
and we have to decide which Green's function of $\underline{\mathcal{M}}$ to use. Different choices correspond to different possible interpretations of the Jacobian itself. Tentatively we choose $\underline{\mathcal{N}}$ and $J$ to be coherent with the boundary conditions appropriate to the functional integral. 

Now note that
\begin{eqnarray}
\delta \; \log J &=& J^{-1} \delta J \nonumber \\
&=& (-1)^i \phi^{i}_{,\alpha} \; \delta K^\alpha_{,i} \nonumber \\
&=& -(-1)^i \; \tensor[^i]{Q}{_\beta} \; \underline{\mathcal{N}}^{\beta}_{\alpha} \; \delta K^\alpha_{,i} \nonumber \\
&=& - \underline{\mathcal{N}}^{\beta}_{\alpha} \; \delta K^\alpha_{,i} \; \tensor[^i]{Q}{_\beta} \nonumber \\
&=& - \underline{\mathcal{N}}^{\beta}_{\alpha} \; \delta \underline{\mathcal{M}}^\alpha _\beta  \nonumber \\
&=& - \delta \; \log \;{\rm det } \underline{\mathcal{N}};
\end{eqnarray}
therefore 
\begin{equation}
\delta \;( J \; {\rm det } \underline{\mathcal{N}}) = 0,
\end{equation}
i.e.,the product $J \; {\rm det } \underline{\mathcal{N}}$ is independent of how the coordinates $K^\alpha$ are chosen. The same is also true of the products $J^\pm \; {\rm det } \underline{\mathcal{N}}^\pm$, where the $J^\pm$ are the Jacobians interpreted according to advanced or retarded boundary conditions. This does \textit{not} automatically mean that these products depend only on the $I$'s and are hence gauge invariant: in fact, these products are not scalar functionals, but scalar densities of unit weight; therefore:
\begin{eqnarray}
(J \; {\rm det } \underline{\mathcal{N}}) \overleftarrow{\mathcal{L}}_{Q_\alpha} &=& \;J \; {\rm det } \underline{\mathcal{N}} \; {\rm str }\tensor[^i]{Q}{_{\alpha}_{,k}} + (J \; {\rm det } \underline{\mathcal{N}}) Q_\alpha \nonumber \\
&=&  (-1)^i \; J \; {\rm det } \underline{\mathcal{N}} \;\tensor[^i]{Q}{_\alpha_{,i}}  + (J \; {\rm det } \underline{\mathcal{N}}\tensor{)}{_{,i}} \tensor[^i]{Q}{_\alpha} \nonumber \\
&=& (-1)^i \; J \; {\rm det } \underline{\mathcal{N}} \;\tensor[^i]{Q}{_\alpha_{,i}}   \nonumber \\
&& + (J \; {\rm det } \underline{\mathcal{N}}\tensor{)}{_{,A}}\;I^A _{,i}\; \tensor[^i]{Q}{_\alpha} \nonumber \\
&& + (J \; {\rm det } \underline{\mathcal{N}}\tensor{)}{_{,\beta}}\;K^\beta _{,i}\; \tensor[^i]{Q}{_\alpha} \nonumber \\
&=& (-1)^i \; J \; {\rm det } \underline{\mathcal{N}} \;\tensor[^i]{Q}{_\alpha_{,i}}  \nonumber \\ 
&& + (J \; {\rm det } \underline{\mathcal{N}}\tensor{)}{_{,A}}\;(I^A \; Q_{\alpha}) \nonumber \\
&& + (J \; {\rm det } \underline{\mathcal{N}}\tensor{)}{_{,\beta}}\;K^\beta _{,i}\; \tensor[^i]{Q}{_\alpha} \nonumber \\
&=& (-1)^i \; J \; {\rm det } \underline{\mathcal{N}} \;\tensor[^i]{Q}{_\alpha_{,i}} , \label{Lie_der_meas}
\end{eqnarray}
where the following facts have been used:
\begin{itemize}
	\item[1.] the $I$ 's are flow invariant, i.e., $I^A \; Q_\alpha = 0$,
	\item[2.] $J \; {\rm det } \underline{\mathcal{N}}$ is independent of the $K$ 's, i.e., $(J \; {\rm det } \underline{\mathcal{N}}\tensor{)}{_{,\beta}}=0$,
	\item[3.] $\overleftarrow{\delta}/\delta \phi^i = \overleftarrow{\delta}/\delta I^A \;I^A _{,i} + \overleftarrow{\delta}/\delta K^\beta \;K^\beta _{,i}$.
\end{itemize}

Taking the Lie derivative of $J \; {\rm det } \underline{\mathcal{N}}$ with respect to $Q_\alpha$ is seen to be the same as multiplying it by $(-1)^i\;\tensor[^i]{Q}{_\alpha_{,i}}$ , which is a constant. This constant, which essentially describes a rescaling of $J \; {\rm det } \underline{\mathcal{N}}$ as one moves up and down the fibres, depends in no way on the choice made for the $K$'s. For \textit{practical} purposes it may be taken as \textit{zero}, for 
two reasons: 
\begin{itemize}
	\item[1.] In Yang-Mills theory the zero value may follow formally from the compactness of the finite dimensional Lie group with which it is associated, which implies $\tensor{f}{^\gamma _\gamma _\alpha}=0$:
	\begin{eqnarray}
	\tensor[^\gamma _\mu]{Q}{_\alpha _' _{,}^\mu _\gamma } &=& \int  dx \;  \delta (x,x') \delta (x,x) \underline{\delta}^\mu _\mu \; \tensor{f}{^\gamma _\gamma _\alpha} \nonumber \\
	&=& N \delta (x',x') \tensor{f}{^\gamma _\gamma _\alpha} \nonumber \\
	&=& 0.
	\end{eqnarray}
	\item[2.] In both Yang-Mills and gravity theories it is also formally either a $\delta$-distribution, or the derivative of a $\delta$-distribution, with coincident arguments. In dimensional regularization such formal expressions vanish.
\end{itemize}
From now on therefore we set 
\begin{equation}
(-1)^i\;\tensor[^i]{Q}{_\alpha_{,i}} = 0,
\end{equation}
and similarly
\begin{equation}
\tensor[]{c}{^\beta _\alpha _\beta} = 0.
\end{equation}

\subsection{A special choice for $\Omega[\phi]$ and a new measure functional }
Although we know that the $K^\alpha$ cannot be global coordinates when the gauge group is non-Abelian, in the loop expansion we can pretend that they are. That is, we pretend that they are coordinates in a tangent space. A favorite choice for the functional $\Omega[I,K]$ is then 
\begin{equation}
\Omega = \frac{1}{2} \kappa_{\alpha\beta}K^\alpha K^\beta
\end{equation}
where $(\kappa_{\alpha\beta})$ is a symmetric ultralocal invertible continuous real matrix which can be chosen either to be constant or to depend on the base point $\phi _*$ in the neighborhood of which the operator $\underline{\mathcal{M}}^\alpha _\beta$ of \eqref{def_underl_M} is nonsingular. The $K$'s themselves may be chosen to vanish at this base point.

Since we are staying in a single chart it is simplest to choose 
\begin{equation}
\mu _{K}[I,K] = 1
\end{equation}
so that $\eqref{def_Delta}$ reduces to 
\begin{equation}
\Delta = {\rm const } \cdot ({\rm det } \kappa)^{-1/2}.
\end{equation}
Equations \eqref{path_int_flows2bis} and \eqref{path_int_flows3} then take the forms 
\begin{eqnarray}
\langle {\rm out }|{\rm in } \rangle &=& \int \mu[\phi] [d\phi]  e^{i(S[\phi]+\frac{1}{2} \kappa_{\alpha\beta}K^\alpha K^\beta)} ({\rm det } \underline{\mathcal{N}})^{-1}, \label{path_integral_4}\\
\langle {\rm out }|\mathcal{T}(A[\phi])| {\rm in } \rangle &=&  \int \mu[\phi] [d\phi] A[\phi]e^{i(S[\phi]+\frac{1}{2} \kappa_{\alpha\beta}K^\alpha K^\beta)}({\rm det } \underline{\mathcal{N}})^{-1}, \label{path_integral_5} 
\end{eqnarray}
where 
\begin{equation}
\mu[\phi] = {\rm const } \cdot \mu _I [\phi] ({\rm det } \kappa)^{1/2} J[\phi] {\rm det } \underline{\mathcal{N}}
\end{equation}
The previous expression may be regarded as a new measure functional, which is to be used when the integration is carried out over the whole space of histories $\Phi$ rather than just the base space $\Phi/\mathcal{G}$. By virtue of eq. \eqref{Lie_der_meas}, the constancy of $\kappa$, and the 
fact that $\mu _I [\phi]$ depends only on the $I$'s, it follows that this measure satisfies 
\begin{equation}
\mu \overleftarrow{\mathcal{L}}_{Q_\alpha} = 0.
\end{equation}

\subsection{Ghosts and BRST symmetry}
The functional integrals \eqref{path_integral_4} and \eqref{path_integral_5} do not differ greatly in form from expressions \eqref{path_int} and \eqref{def_TA3}. The chief difference is the presence of the $K$'s and $\kappa$'s, and the curious factor $({\rm det }\underline{\mathcal{N}})^{-1}$ which comes ultimately from the Jacobian $J$. By introducing the $a$-type \textit{ghost fields} $\chi_\alpha,\psi^\beta$, one obtains
\begin{equation}
({\rm det }\underline{\mathcal{N}})^{-1} = \int [d\chi] \int [d\psi] \; e^{i\chi_\alpha \underline{\mathcal{M}}^\alpha _\beta \psi^\beta}.
\end{equation}
Therefore eq. \eqref{path_integral_4} may be written 
\begin{equation}
\langle {\rm out }|{\rm in } \rangle = \int \mu[\phi] [d\phi] \int [d\chi] \int [d\psi] \; e^{i(S[\phi]+\frac{1}{2} \kappa_{\alpha\beta}K^\alpha K^\beta + \chi_\alpha \underline{\mathcal{M}}^\alpha _\beta \psi^\beta)}. \label{path_int_ghost}
\end{equation}
It is important to emphasize that the ghost fields arise entirely from the fiber-bundle structure of $\Phi$, from the Jacobian of the transformation from the fiber-adapted coordinates to the conventional local fields $\phi^i$. \footnote{These \lq\lq tricky extra particles'' were first introduced by R.P. Feynman (see \cite{feynman1963quantum}) as a way of compensating for the propagation of nonphysical modes in one-loop order.}

Hence, the theory may be seen as a field theory on an \lq\lq extended'' space of field histories $\bar{\Phi}$, where the ghost field $\chi_\alpha,\psi^\beta$ appears in addition to $\phi$'s; however, the new fields have to be considered non-physical, since their \lq\lq fermionic number'' is the opposite of the one suggested by their indices, i.e., they are fermionic fields altough $\alpha ,\beta$ are bosonic indices.
It is clear that the full argument of the exponential in \eqref{path_int_ghost} is no longer invariant under gauge transformations, because of the \lq\lq gauge-averaging'' term $\tfrac{1}{2} \kappa_{\alpha\beta}K^\alpha K^\beta$ and the ghost term $\chi_\alpha \underline{\mathcal{M}}^\alpha _\beta \psi^\beta$.

Nevertheless, both the full action functional 
\begin{equation}
\bar{S}[\phi,\chi,\psi] \equiv S[\phi]+\tfrac{1}{2} \kappa_{\alpha\beta}K^\alpha K^\beta + \chi_\alpha \underline{\mathcal{M}}^\alpha _\beta \psi^\beta
\label{full_act} \end{equation}
and the measure 
\begin{equation}
\bar{\mu}[\phi,\chi,\psi][d\phi][d\chi][d\psi] \equiv \mu[\phi] [d\phi] [d\chi] [d\psi]
\end{equation}
are invariant under a group of \textit{global transformations} whose infinitesimal form is
\begin{eqnarray}
\delta \phi^i &=& \tensor[^i]{Q}{_\alpha}\; \psi^\alpha \delta \lambda, \\
\delta \chi_\alpha &=& \kappa_{\alpha\beta} K^\beta \delta \lambda, \\
\delta \psi^\alpha &=& -\tfrac{1}{2} \tensor{c}{^\alpha _\beta _\gamma} \; \psi^\beta \psi^\gamma \delta \lambda,
\end{eqnarray}
where $\delta \lambda$ is an infinitesimal $a$-number. They are called Becchi-Rouet-Stora-Tyutin (BRST) transformations.

As is clear, for the fields $\phi^i$, they are gauge transformations with an $a$-type parameter; therefore, from the invariance of $\mu[\phi]$ under gauge transformations, $\bar{\mu}[\phi,\chi,\psi]$ is invariant under BRST transformations; one can show that $\bar{S}[\phi,\chi,\psi]$ is invariant too:
\begin{equation}
\delta S[\phi] = 0, \label{dim_BRST_1}
\end{equation}
\begin{eqnarray}
\delta (\tfrac{1}{2} \kappa_{\alpha\beta}K^\alpha K^\beta) &=& \tfrac{1}{2} \cdot 2 \kappa_{\alpha\beta}\delta K^\alpha K^\beta \nonumber \\
&=& \kappa_{\alpha\beta}( K^\alpha_{,i}\;\delta\phi^i) K^\beta \nonumber \\
&=& \kappa_{\alpha\beta}( K^\alpha_{,i}\;\tensor[^i]{Q}{_\gamma}\; \psi^\gamma \delta \lambda) K^\beta \nonumber \\
&=& \kappa_{\alpha\beta}( K^\alpha_{,i}\;\tensor[^i]{Q}{_\gamma}\; \psi^\gamma) K^\beta \delta \lambda \nonumber \\
&=& \kappa_{\alpha\beta} \underline{\mathcal{M}}^\alpha _\gamma \; \psi^\gamma K^\beta \delta \lambda , \label{dim_BRST_2}
\end{eqnarray} 
\begin{eqnarray}
\delta (\chi_\alpha \underline{\mathcal{M}}^\alpha _\beta \psi^\beta) &=& (\delta \chi_\alpha) \underline{\mathcal{M}}^\alpha _\beta \psi^\beta + \chi_\alpha \underline{\mathcal{M}}^\alpha _\beta (\delta \psi^\beta) + \chi_\alpha (\delta\underline{\mathcal{M}}^\alpha _\beta) \psi^\beta \nonumber \\
&=& \kappa_{\alpha\gamma} K^\gamma \delta \lambda \underline{\mathcal{M}}^\alpha _\beta \psi^\beta + \chi_\alpha \underline{\mathcal{M}}^\alpha _\beta ( -\tfrac{1}{2} \tensor{c}{^\beta _\alpha _\gamma} \; \psi^\alpha \psi^\gamma \delta \lambda)  \nonumber \\&&+ \chi_\alpha (\delta K^\alpha _{,i} \; \tensor[^i]{Q}{_\beta}) \psi^\beta + \chi_\alpha (K^\alpha _{,i} \; \delta \tensor[^i]{Q}{_\beta}) \psi^\beta \nonumber \\
&=& -\kappa_{\alpha\beta}   \underline{\mathcal{M}}^\alpha _\gamma \psi^\gamma K^\beta \delta \lambda -\tfrac{1}{2}\chi_\alpha \underline{\mathcal{M}}^\alpha _\beta \; \tensor{c}{^\beta _\zeta _\gamma} \; \psi^\zeta \psi^\gamma \delta \lambda \nonumber \\ &&+ \chi_\alpha ( K^\alpha _{,ij}\delta \phi^j \; \tensor[^i]{Q}{_\beta}) \psi^\beta +\chi_\alpha (K^\alpha _{,i} \;  \tensor[^i]{Q}{_\beta _{,j}}) \; \delta \phi^j \psi^\beta \nonumber \\
&=& -\kappa_{\alpha\beta}   \underline{\mathcal{M}}^\alpha _\gamma \psi^\gamma K^\beta \delta \lambda -\tfrac{1}{2}\chi_\alpha \underline{\mathcal{M}}^\alpha _\beta \; \tensor{c}{^\beta _\zeta _\gamma} \; \psi^\zeta \psi^\gamma \delta \lambda  \nonumber \\ && +\chi_\alpha ( K^\alpha _{,ij}\;\tensor[^j]{Q}{_\eta}\; \psi^\eta \delta \lambda \; \tensor[^i]{Q}{_\beta}) \psi^\beta + \chi_\alpha (K^\alpha _{,i} \;  \tensor[^i]{Q}{_\beta _{,j}}) \; \tensor[^j]{Q}{_\zeta}\;\psi^\zeta \delta \lambda  \psi^\beta  \nonumber \\
&=& -\kappa_{\alpha\beta}   \underline{\mathcal{M}}^\alpha _\gamma \psi^\gamma K^\beta \delta \lambda -\tfrac{1}{2}\chi_\alpha \underline{\mathcal{M}}^\alpha _\beta \; \tensor{c}{^\beta _\zeta _\gamma} \; \psi^\zeta \psi^\gamma \delta \lambda  \nonumber \\ && +\chi_\alpha  K^\alpha _{,ij}\;\tensor[^j]{Q}{_\eta} \; \tensor[^i]{Q}{_\beta} \; \psi^\beta \psi^\eta \delta \lambda  \nonumber \\ &&+ \chi_\alpha K^\alpha _{,i} \;  \tensor[^i]{Q}{_\beta _{,j}} \; \tensor[^j]{Q}{_\zeta}\;  \psi^\beta \psi^\zeta \delta \lambda . \label{dim_BRST_3}
\end{eqnarray}
In the last equation, consider the third term:
\begin{eqnarray}
\chi_\alpha  K^\alpha _{,ij}\;\tensor[^j]{Q}{_\eta} \; \tensor[^i]{Q}{_\beta} \; \psi^\beta \psi^\eta \delta \lambda &=& \chi_\alpha  K^\alpha _{,ji}\;\tensor[^i]{Q}{_\eta} \; \tensor[^j]{Q}{_\beta} \; \psi^\beta \psi^\eta \delta \lambda \nonumber \\
&=& (-1)^{ij} \; \chi_\alpha  K^\alpha _{,ij}\;\tensor[^i]{Q}{_\eta} \; \tensor[^j]{Q}{_\beta} \; \psi^\beta \psi^\eta \delta \lambda \nonumber \\
&=& \chi_\alpha  K^\alpha _{,ij}\; \tensor[^j]{Q}{_\beta} \;\tensor[^i]{Q}{_\eta} \; \psi^\beta \psi^\eta \delta \lambda \nonumber \\
&=& -\chi_\alpha  K^\alpha _{,ij}\; \tensor[^j]{Q}{_\beta} \;\tensor[^i]{Q}{_\eta} \;  \psi^\eta \psi^\beta \delta \lambda \nonumber \\
&=& -\chi_\alpha  K^\alpha _{,ij}\; \tensor[^j]{Q}{_\eta} \;\tensor[^i]{Q}{_\beta} \;  \psi^\beta \psi^\eta \delta \lambda,
\end{eqnarray}
therefore 
\begin{equation}
\chi_\alpha  K^\alpha _{,ij}\;\tensor[^j]{Q}{_\eta} \; \tensor[^i]{Q}{_\beta} \; \psi^\beta \psi^\eta \delta \lambda = 0.
\end{equation}
Consider again \eqref{dim_BRST_3}; adding the second and the fourth terms, one obtains 
\begin{eqnarray}
&&-\tfrac{1}{2}\chi_\alpha \underline{\mathcal{M}}^\alpha _\beta \; \tensor{c}{^\beta _\zeta _\gamma} \; \psi^\zeta \psi^\gamma \delta \lambda + \chi_\alpha K^\alpha _{,i} \;  \tensor[^i]{Q}{_\beta _{,j}} \; \tensor[^j]{Q}{_\zeta}\;  \psi^\beta \psi^\zeta \delta \lambda   \nonumber \\
&&= -\tfrac{1}{2}\chi_\alpha \; K^\alpha _{,i} \; \tensor[^i]{Q}{_\beta} \; \tensor{c}{^\beta _\zeta _\gamma} \; \psi^\zeta \psi^\gamma \delta \lambda + \chi_\alpha K^\alpha _{,i} \;  \tensor[^i]{Q}{_\zeta _{,j}} \; \tensor[^j]{Q}{_\gamma}\;  \psi^\zeta \psi^\gamma \delta \lambda  \nonumber \\
&&= -\tfrac{1}{2}\chi_\alpha \; K^\alpha _{,i} \; \tensor[^i]{Q}{_\beta} \; \tensor{c}{^\beta _\zeta _\gamma} \; \psi^\zeta \psi^\gamma \delta \lambda \nonumber \\
&& + \tfrac{1}{2} \chi_\alpha K^\alpha _{,i} \;  \tensor[^i]{Q}{_\zeta _{,j}} \; \tensor[^j]{Q}{_\gamma}\;  \psi^\zeta \psi^\gamma \delta \lambda - \tfrac{1}{2} \chi_\alpha K^\alpha _{,i} \;  \tensor[^i]{Q}{_\gamma _{,j}} \; \tensor[^j]{Q}{_\zeta}\;  \psi^\zeta \psi^\gamma \delta \lambda \nonumber \\
&&= -\tfrac{1}{2}\chi_\alpha \; K^\alpha _{,i} \; \tensor[^i]{Q}{_\beta} \; \tensor{c}{^\beta _\zeta _\gamma} \; \psi^\zeta \psi^\gamma \delta \lambda \nonumber \\
&& + \tfrac{1}{2} \chi_\alpha K^\alpha _{,i} (\tensor[^i]{Q}{_\zeta _{,j}} \; \tensor[^j]{Q}{_\gamma}- \tensor[^i]{Q}{_\gamma _{,j}} \; \tensor[^j]{Q}{_\zeta})  \psi^\zeta \psi^\gamma \delta \lambda \nonumber \\
&&=  -\tfrac{1}{2}\chi_\alpha \; K^\alpha _{,i} \; \tensor[^i]{Q}{_\beta} \; \tensor{c}{^\beta _\zeta _\gamma} \; \psi^\zeta \psi^\gamma \delta \lambda \nonumber \\ 
&& + \tfrac{1}{2} \chi_\alpha K^\alpha _{,i} \; \tensor[^i]{Q}{_\beta} \;\tensor{c}{^\beta _\zeta _\gamma}  \;\psi^\zeta \psi^\gamma \delta \lambda \nonumber \\
&& = 0.
\end{eqnarray}
Noting that the first term in \eqref{dim_BRST_3} is the opposite of the rhs of \eqref{dim_BRST_2}, it is proven that
\begin{equation}
\delta \bar{S} = \delta S + \delta (\tfrac{1}{2} \kappa_{\alpha\beta}K^\alpha K^\beta) + \delta (\chi_\alpha \underline{\mathcal{M}}^\alpha _\beta \psi^\beta) = 0.
\end{equation}
Moreover, BRST invariance is often a good substitute for the original gauge invariance. For example, if $A$ is a functional of $\phi$, but not of the ghost field, it is easy to see that $A$ is gauge invariant if and only if it is BRST invariant:
\begin{eqnarray}
\delta _{BRST} A &=&  \tensor{A}{_{,i}} \; \delta _{BRST} \phi^i  \nonumber \\
&=& \tensor{A}{_{,i}} \; \tensor[^i]{Q}{_\alpha} \; \psi^\alpha \delta \lambda; 
\end{eqnarray}
In this case, since the $\psi^\alpha$ are arbitrary functions on space-time (although not necessarily of compact support) BRST invariance of $A$ implies gauge invariance in the original sense, and vice versa. 

Therefore, any Type-I gauge theory with action functional $S[\phi]$ may be viewed as a non-gauge, BRST-symmetric theory on an extendend space of field histories where ghost fields appear, with action functional $\bar{S}[\phi,\chi,\psi]$ given by \eqref{full_act}.

\begin{remark}
	In the general case when no assumption on the bosonic nature of the indices from the first part of the Greek alphabet, ${\rm sdet } (\underline{N}^\alpha _\beta)^{-1}$ appears in place of ${\rm det } (\underline{N}^\alpha _\beta)^{-1}$; therefore (see \cite{berazin2012method} and \cite{dewitt1992supermanifolds}) if the indices from the first part of the Greek alphabet are bosonic, one obtains
	\begin{equation}
	{\rm sdet } (\underline{N}^\alpha _\beta)^{-1} = {\rm det } (\underline{N}^\alpha _\beta)^{-1} = \int [d\chi] \int [d\psi] \; e^{i\chi_\alpha \underline{\mathcal{M}}^\alpha _\beta \psi^\beta},
	\end{equation}
	where $\chi_\alpha,\psi^\beta$ are fermionic fields, as shown in the previous section.
	On the other hand, if the indices from the first part of the Greek alphabet are fermionic, then 
	\begin{equation}
	{\rm sdet } (\underline{N}^\alpha _\beta)^{-1} = {\rm det } (\underline{N}^\alpha _\beta) = \int [d\chi] \int [d\psi] \; e^{i\chi_\alpha \underline{\mathcal{M}}^\alpha _\beta \psi^\beta},
	\end{equation}
	where, in this case, $\chi_\alpha,\psi^\beta$ are bosonic fields.
	
	Hence, we infer that the fermionic nature of the ghost fields is \textit{always} opposite to the one suggested by their indices, and this is another clue that ghost fields do not represent physical particles.
\end{remark}
\subsection{A few words on the measure functional}
The measure functional was introduced as a device for correcting the possible failure of chronological ordering to yield Hermitian (or skew-Hermitian) operator field equations. It arose from the noncommutativity (or nonanticommutativity) of field operators and hence is a purely quantum construct. But the measure functional plays a far deeper role, and we shall briefly outline the reason in this section.

As is well known, the main tool to evaluate transition amplitudes in an interacting field theory is renormalized perturbation theory; in order to obtain renormalized observables, one has to choose a renormalization scheme and has to deal with divergent Feynman diagrams, up to a chosen order. Consider now one-loop perturbation theory for some field in Minkowski space-time; in Minkowski space-time one can use the Fourier transform and pass to momentum space; therefore the task is to evaluate a graph consisting of a single closed loop with $r$ external prongs. Let the momenta assigned to the internal lines all have the same orientation around the loop. Then, making use of the so-called \lq\lq Feynman's trick'' to combine the factors contributed by the internal lines, i.e., by the propagators, and appropriately shifting the integration zero point, one finds for the Feynman-propagator contribution to the value of the graph an expression having the general form 
\begin{eqnarray}
I(C) &=& {\rm constant} \; \cdot \; \int d^{r-1}y \int d^n k \;\frac{P_m (y,k,p)}{[k^2 - i\epsilon + Q_m (y,k,p)]^r} \label{Wick_1} \\
&=& {\rm constant} \; \cdot \; \int d^{r-1}y \int_C d^n k \; \frac{P_m (y,k,p)}{[k^2 + Q_m (y,k,p)]^r} \label{Wick_2}
\end{eqnarray}
in which external space-time and/or spinor indices have been suppressed. Here \lq\lq $y$'' denotes the parameters $y_{1} , y_{2} ,...,y_{r-1}$ needed to implement \lq\lq Feynman's trick'' and  $\int d^{r-1} y$ is a schematic symbol for the integrations in which these parameters are involved. The incoming momenta at the external prongs are $(-p_{2} - p_{3} -...- p_{r} ), p_{2} ,...,p_{r}$. $Q_m$ is a quadratic function of these momenta, which also depends on the $y$'s and on the masses $m$ associated with the internal lines. $P_m$ is a polynomial in the $k$'s and $p$'s, which depends on the $y$'s and $m$'s. $C$ in \eqref{Wick_2} denotes the contour in the complex plane of the time component $k^0$ of the $k$-variable which is appropriate to the Feynman propagator: this contour runs from $-\infty$ to $0$ below the negative real axis (in the complex $k^0$-plane) and from $0$ to $+\infty$ above the positive real axis: it is the same as integrating on the real line with the $i\epsilon$ prescription used in \eqref{Wick_1}. \textit{If the integral were convergent}, the contour could be rotated so that it would run along the imaginary axis. One would set $k^0 = ik^n$, and \eqref{Wick_2} would become an integral over Euclidean momentum-$n$-space. Generically, however, this rotation, which is known as \textit{Wick rotation}, is not legitimate. Contributions from arcs at infinity, which themselves diverge or are nonvanishing, have to be included. \textit{These contributions cannot be handled by dimensional regularization}.

When the measure is included it contributes to the generic one-loop graph an amount equal to the negative of the integral
\begin{eqnarray}
I(C^+) &=& {\rm constant} \; \cdot \; \int d^{r-1}y \int d^n k \; \frac{P_m (y,k,p)}{[-(k^0 - i\epsilon)^2 + (\vec{k})^2 + Q_m (y,k,p)]^r} \nonumber \\
&=& {\rm constant} \; \cdot \; \int d^{r-1}y \int_{C^+} d^n k \; \frac{P_m (y,k,p)}{[k^2 + Q_m (y,k,p)]^r} \label{Wick_3}
\end{eqnarray} 
where $C^+$ is the contour (in the complex $k^0$-plane) appropriate to the advanced Green function: it runs from $-\infty$ to $+\infty$ below the real axis; it is the same as integrating on the real line with the $i\epsilon$ prescription used on the first line of the previous equation. These two contributions, taken together, yield $I(C)-I(C^+)$ as the correct value of the graph. This corresponds to taking a contour that runs from $+\infty$ to $0$ below the positive real axis and then back to $+\infty$ again above the positive real axis, and yields an integral that \textit{can} be handled by dimensional regularization.

The remarkable fact is that $I(C)-I(C^+)$ is equal precisely to the value that is obtained by Wick rotation. This means that \textit{the measure justifies
	the Wick-rotation procedure}. Although it has never been proved, one may speculate that the exact measure functional, whatever it is, will justify the Wick rotation to all orders and will establish a rigorous connection between quantum field theory in Minkowski space-time and its corresponding euclideanized version.

\section{Green's Functions: Neutral Scalar Meson}
\subsection{Integral representations in Minkowski space-time}
Green's functions have been introduced in section 2 : in gauge theories, they are the negative inverses of the differential operator $\tensor[_i]{F}{_j}$, while in field theories with no gauge transformations, they are the negative inverses of the non-singular operator $\tensor[_1]{S}{_1}$. The prototypes of the Green's functions of interest in quantum field theory are those of the neutral scalar meson in Minkowski space-time; by spin-statistics theorem, it has to be a bosonic particle, and its Lagrangian is
\begin{equation}
L = -\tfrac{1}{2} \left(\tensor{\phi}{_,^\mu} \phi _{,\mu} + m^2 \phi^2 \right).
\end{equation} 
Therefore 
\begin{equation}
\tensor{S}{_1} = \tensor[_1]{S}{} = \frac{\delta}{\delta \phi(x)} S =  \left( \tensor{\phi}{_, ^\mu _\mu} - m^2 \phi \right),
\end{equation}
and 
\begin{eqnarray}
\tensor[_1]{S}{_1} &=& \tensor[_2]{S}{} \nonumber \\
&=& \frac{\delta}{\delta \phi(x')}\frac{\delta}{\delta \phi(x)} S \nonumber \\
&=&  \left( \tensor{\delta(x,x')}{_, ^\mu _\mu} - m^2 \delta(x,x') \right) \nonumber \\
&=&  \left( \partial_\mu \partial^\mu - m^2 \right)\delta(x,x').
\end{eqnarray}
Hence \eqref{def_Gr} takes the form
\begin{eqnarray}
&& \int dy \; \left[\left( \partial_\mu \partial^\mu - m^2 \right)\delta(x,y)\right] G(y,x') \nonumber  \\ &&=  \left( \partial_\mu \partial^\mu - m^2 \right)G(x,x') = -\delta(x,x').
\end{eqnarray}
This equation is most easily solved in \lq\lq momentum'' space, i.e., using the Fourier transform; as is well known, it is a linear, invertible operator which turns derivative operators into multiplication operators, i.e., the Fourier transform of a linear, differential equation for a function is a linear, algebraic equation which is easily solved; therefore the desired function can be obtained by inverse Fourier transform; in our case, the result is 
\begin{equation}
G(x,x') = \frac{1}{(2\pi)^4} \int dp \; \frac{e^{ip(x-x')}}{m^2 + p^2}. \label{scalar_Gr}
\end{equation}
In the previous equation the contours in the $p^1, p^2, p^3$ planes are confined to the real axis and the choice of Green's function is determined by selecting a contour in the $p^0$ plane which passes in an appropriate fashion around the poles at $\pm E$ where 
\begin{eqnarray}
E &\equiv & \sqrt{m^2 + \vec{p}^2} \equiv \omega, \\
\vec{p} & \equiv & (p^1, p^2, p^3), \\
(\vec{p})^2 &\equiv & (p^1)^2 + (p^2)^2 + (p^3)^2 .
\end{eqnarray}
The most important contours are shown in Fig. \ref{figure}. From these contours the following relations between the various Green's functions are easily established: 
\begin{eqnarray}
\bar{G} &=& \tfrac{1}{2} \left(G^+  + G^- \right) = \tfrac{1}{2} \tilde{G} +  G^- = -\tfrac{1}{2} \tilde{G} + G^+ ,  \\ 
\tilde{G} &=& G^+  - G^- = G^{(+)}  + G^{(-)}, \label{useful_1}\\
G^{(1)} &=& i \left(G^{(+)}  - G^{(-)}\right), \\
G &=& \bar{G} + \tfrac{1}{2}i G^{(1)} = G^- +  G^{(-)} = G^+ - G^{(+)}, \\
G^* &=& \bar{G} - \tfrac{1}{2}i G^{(1)} = G^- +  G^{(+)} = G^+ - G^{(-)}. 
\end{eqnarray}
\begin{figure}[htbp]
	\includegraphics[scale=0.45]{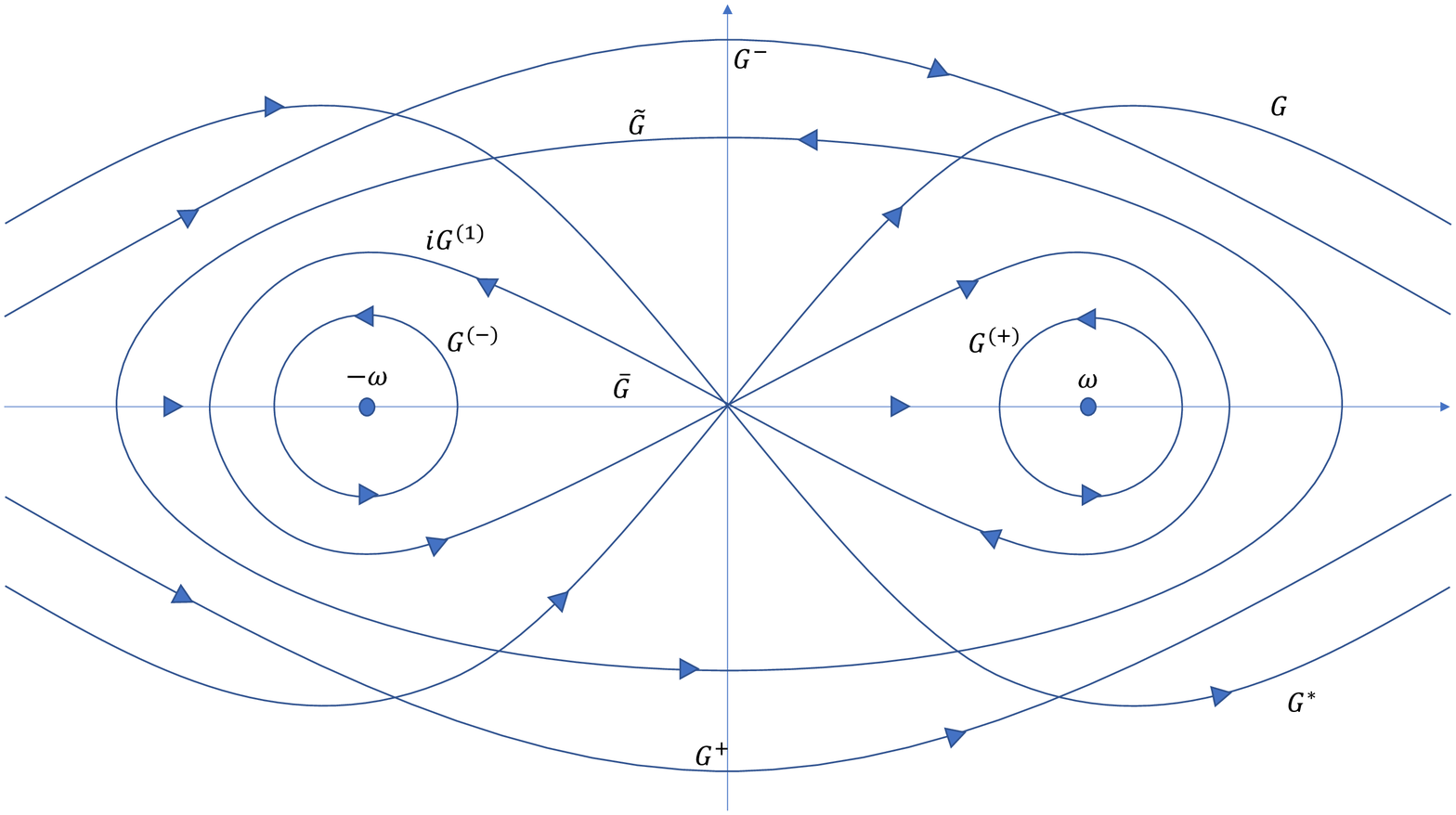}
	\vspace*{8pt}
	\caption{Contours in the complex $p^0$-plane for the integral representation of the Green's function of the neutron scalar meson.}
	\label{figure}
\end{figure}

By closing the contours for $G^+$ and $G^-$ at infinity it is easy to see that these functions satisfy the kinematical conditions \eqref{ret_adv_Gr} and hence are the advanced and retarded Green's functions. Their uniqueness is also evident. We may therefore write the further relations 
\begin{eqnarray}
G^+ (x,x') &=& 2\theta(x',x) \bar{G} (x,x')= \theta(x',x)\tilde{G} , \label{useful_f1}\\
G^- (x,x') &=& 2\theta(x,x') \bar{G} (x,x')= -\theta(x,x')\tilde{G}  ,\label{useful_f2}\\
\tilde{G}(x,x') &=& -2\epsilon(x,x')\bar{G} (x,x'), \label{useful_2} \\
\bar{G} (x,x') &=& -\tfrac{1}{2}\epsilon(x,x')\tilde{G}(x,x'), \label{useful_3}
\end{eqnarray}
where $\theta(x,x')$ and $\epsilon(x,x')$ are the step functions, defined by 
\begin{eqnarray}
\theta(x,x') &\equiv &\begin{cases}
1 \; \; {\rm for } \; \;x > x', \\
0  \;  \;{\rm for } \; \; x < x',
\end{cases} \\
&=& 1- \theta(x',x) = \tfrac{1}{2} \left[ 1 + \epsilon(x,x') \right],  \label{useful_4}\\
\epsilon(x,x') &=& \theta(x,x') - \theta(x',x) = \begin{cases}
1  \; \;{\rm for }   \; x > x', \\ -1  \; \;{\rm for }  \;  x < x', \\
\end{cases}
= - \epsilon(x',x).
\end{eqnarray}
We also have 
\begin{eqnarray}
G(x,x') &=& -\theta(x,x')G^{(+)}(x,x')+\theta(x',x)G^{(-)}(x,x'), \label{useful_def_Feynman}\\
G^*(x,x') &=& \theta(x',x)G^{(+)}(x,x')-\theta(x,x')G^{(-)}(x,x'), \\
G^{(+)}(x,x') &=&  -\theta(x,x')G(x,x') + \theta(x',x)G^* (x,x'), \\
G^{(-)}(x,x') &=&  \theta(x',x)G(x,x') - \theta(x,x')G^* (x,x'), 
\end{eqnarray}
which follow from \eqref{useful_1}, \eqref{useful_2}, \eqref{useful_3}, \eqref{useful_4}, and the identities 
\begin{eqnarray}
\theta(x,x')\theta(x',x) &=& 0 \\
\left[ \theta(x,x') \right]^2 &=&  \theta(x,x') \\
\left[ \epsilon(x,x') \right]^2 &=&  1
\end{eqnarray}
Care should be exercised in the use of the step functions. Strictly speaking, all the equations where $\theta(x,x')$, $\epsilon(x,x')$ appear and their corollaries can be inferred to hold only when one of the two points $x, x'$ is clearly to the future or the past of the other. When the two points are separated by a space-like interval further investigation is needed. We shall see presently that the functions $G^\pm(x, x')$ vanish for finite space-like  
separations, and hence the investigation reduces to a study of the behavior of the Green's functions when $x'$ is in the immediate neighborhood of $x$. The study is complicated by the fact that the Green's functions are actually \textit{distributions} rather than ordinary functions. It turns out, in the present case, that the above relations are in fact valid everywhere. Analogous relations, for the Green's functions of systems more complicated than the neutral scalar meson, however, do not always similarly hold when $x = x'$. In this work we shall avoid this difficulty by using the step functions only when $x \neq x'$. We may also remark that there will never be any ambiguity about the Green's functions themselves. In the present case they are well defined by the integral representation \eqref{scalar_Gr}, once the contour is chosen. 

\subsection{Symmetries of Green's functions}
For the neutral scalar meson, the reciprocity relations \eqref{recipr_rel} and \eqref{symm_supercomm} read
\begin{eqnarray}
G^\pm (x,x') &=& G^\mp (x',x), \\
\tilde{G}(x,x') &=& -\tilde{G}(x',x), \\
\bar{G}(x,x') &=& \bar{G}(x',x).
\end{eqnarray}
The contour for the function $\bar{G}$ corresponds to performing a \textit{principal value} integration along the real axis. In light of the reality of all the integration variables in this case, and because of the symmetry (in $p$) of the denominator of the integrand of \eqref{scalar_Gr}, we may infer the reality of $\bar{G}$:
\begin{equation}
\bar{G}^* = \bar{G}.
\end{equation}
Similarly, by performing the transformation $p \mapsto - p$, paying attention to the contour and using the previous relations, we may infer 
\begin{eqnarray}
G(x,x')&=&G(x',x), \\
G^{(1)}(x,x')&=& G^{(1)}(x',x), \\
G^{(\pm)}(x,x')&=& -G^{(\mp)}(x',x),
\end{eqnarray}
and 
\begin{eqnarray}
G^{\pm *} &=& G^{\pm }, \\
\tilde{G}^* &=& \tilde{G}, \\
G^{(1)*}&=&G^{(1)}, \\
G^{(\pm)*}&=& - G^{(\pm)},
\end{eqnarray}
i.e., $G^\pm,  \tilde{G}, G^{(1)}$ are all real, while $G^{(\pm)}$ is imaginary; it follows that $G^*$, as defined above, is the complex conjugate of $G$.

We note, finally, the differential equations satisfied by the various functions: 
\begin{eqnarray}
\left( \partial_\mu \partial^\mu - m^2 \right)G(x,x') &=&  \left( \partial_\mu \partial^\mu - m^2 \right)\bar{G}(x,x') \nonumber \\ &=& \left( \partial_\mu \partial^\mu - m^2 \right)G^\pm(x,x')\nonumber \\&=&-\delta(x,x'), \\
\left( \partial_\mu \partial^\mu - m^2 \right)\tilde{G}(x,x') &=&  \left( \partial_\mu \partial^\mu - m^2 \right)G^{(1)}(x,x') \nonumber \\&=& \left( \partial_\mu \partial^\mu - m^2 \right)G^{(\pm)}(x,x')\nonumber \\&=&0.
\end{eqnarray}

\subsection{The Feynman propagator}
In harmony with section 2, $\tilde{G}$ is known as the \textit{commutator function}, and $G^{(+)}$, $G^{(-)}$ are called its \textit{positive and negative frequency parts}, respectively. $G^{(1)}$ is known as \textit{Hadamard's elementary function}, and $G$ is called the Feynman propagator. From the relations given above it may be seen that all of the functions which we have introduced may be obtained from the Feynman propagator by splitting it into its real, 
imaginary, advanced, and retarded parts, and recombining these parts in various ways. It suffices therefore to evaluate the Feynman propagator in order to obtain all the rest. From Fig. \ref{figure} it is not hard to see that the contour for the Feynman propagator may be displaced to the real axis provided we give to the mass $m$ in eq. \eqref{scalar_Gr} an infinitesimal negative imaginary part, i.e., we move up the negative pole and move down the positive one in the complex $p^0$-plane. We therefore write 
\begin{equation}
G(x,x') = \frac{1}{(2\pi)^4}\int dp \; \frac{e^{ip(x-x')}}{m^2 - i\epsilon + p^2}, \; \; \; \epsilon > 0.
\end{equation}
with the understanding that the limit $\epsilon \to 0$ has to be taken at the end of all integrations. It is important to stress that the Feynman propagator can be obtained by analytic continuation from the \textit{unique} Green's function which the operator $(\partial _\mu \partial ^\mu - m^2)$ possesses when the $x$-manifold has a positive definite metric. This fact is responsible for many of the remarkable properties which characterize the Feynman propagator, and the analytic continuation method is often employed to obtain it. 

Making use of the integral identities
\begin{eqnarray}
&&\int_0 ^{+\infty} ds \; e^{-is(\xi-i\epsilon)} = \frac{1}{i(\xi -i\epsilon)}, \; \; \; \epsilon > 0 \\
&&\int_{-\infty} ^{+\infty} dx \;  e^{iax^2} = \sqrt{(\pi / |a|)} e^{i {\rm sgn}(a)(\pi /4)}, \; \; \; a \; {\rm real},
\end{eqnarray}
(where {\rm sgn} is the signum function: ${\rm sgn}(a)\equiv a/|a|$), one obtains
\begin{eqnarray}
G(x,x') &=& \frac{i}{(2\pi)^4}\int_0 ^{+\infty} ds \int dp \;  e^{ip(x-x')} e^{-is(m^2 + p^2 -i\epsilon)} \nonumber \\
&=& \frac{i}{(2\pi)^4}\int_0 ^{+\infty} ds \int dp \; e^{-i\left[(m^2 + p^2 -i\epsilon)s  -p(x-x')      )        \right]} \nonumber \\
&=& \frac{i}{(2\pi)^4}\int_0 ^{+\infty} ds \; e^{-is(m^2 -i\epsilon)} \int dp \;  e^{-isp^2 +ip(x-x')} \nonumber \\
&=& \frac{i}{(2\pi)^4}\int_0 ^{+\infty} ds \; e^{-is(m^2 -i\epsilon)} \int dp \;  e^{-isp^2 +ip(x-x') - i(\tfrac{x-x'}{2\sqrt{s}})^2} e^{i(\tfrac{x-x'}{2\sqrt{s}})^2} \nonumber \\
&=& \frac{i}{(2\pi)^4}\int_0 ^{+\infty} ds \; e^{-is(m^2 -i\epsilon)}e^{i\left(\tfrac{x-x'}{2\sqrt{s}}\right)^2} \int dp \; \bullet \nonumber \\ && \; \;\bullet  \; \;  e^{-i\left(sp^2 -p(x-x') + \left(\tfrac{x-x'}{2\sqrt{s}}\right)^2\right)} \nonumber \\
&=& \frac{i}{(2\pi)^4}\int_0 ^{+\infty} ds \; e^{-is(m^2 -i\epsilon)+i\left(\tfrac{x-x'}{2\sqrt{s}}\right)^2} \int dp \;  e^{-i\left(\sqrt{s}p - \tfrac{x-x'}{2\sqrt{s}}\right)^2} \nonumber \\
&=& \frac{i}{(2\pi)^4}\int_0 ^{+\infty} ds \; e^{-is(m^2 -i\epsilon)+i\left(\tfrac{x-x'}{2\sqrt{s}}\right)^2} \int dp \;  e^{-is\left(p - \tfrac{x-x'}{2s}\right)^2} \nonumber \\
&=& \frac{i}{(2\pi)^4}\int_0 ^{+\infty} ds \; e^{-is(m^2 -i\epsilon)+i\left(\tfrac{x-x'}{2\sqrt{s}}\right)^2} \;\bullet \nonumber \\ && \; \bullet\left(\sqrt{(\pi / s)}\right)^4  \left(e^{i {\rm sgn}(-s)(\pi /4)}\right)^2 \nonumber \\
&=& \frac{i}{(2\pi)^4}\int_0 ^{+\infty} ds \; e^{-is(m^2 -i\epsilon)+i\left(\tfrac{x-x'}{2\sqrt{s}}\right)^2} \frac{\pi^2}{s^2} e^{-i(\pi /4)2} \nonumber \\
&=& \frac{i}{(2\pi)^4}\int_0 ^{+\infty} ds \; e^{-is(m^2 -i\epsilon)+i\left(\tfrac{x-x'}{2\sqrt{s}}\right)^2} \frac{\pi^2}{s^2} e^{-i\pi/2 } \nonumber \\
&=& \frac{-i^2}{(4\pi)^2}\int_0 ^{+\infty} ds \; \frac{1}{s^2} e^{-is(m^2 -i\epsilon)+i\left(\tfrac{x-x'}{2\sqrt{s}}\right)^2} \nonumber \\
&=& \frac{1}{(4\pi)^2}\int_0 ^{+\infty} ds \; \frac{1}{s^2} e^{-is(m^2 -i\epsilon)+i\left(\tfrac{x-x'}{2\sqrt{s}}\right)^2} \nonumber \\
&=& \frac{1}{(4\pi)^2}\int_0 ^{+\infty} ds \; \frac{1}{s^2} e^{-i\left(m^2 s-\tfrac{(x-x')^2}{4s}\right)} . \label{Feyn_Gr}
\end{eqnarray}
In the final form the negative imaginary part $-i\epsilon$ attached to $m^2$ has been dropped, with the understanding that $G(x,x')$ has to be regarded as the \textit{boundary value} (on the real axis) of a function of $m^2$ and $(x-x')^2$ which is analytic in the lower half $m^2$ plane and in the upper half $(x-x')^2$ plane. The fact that $G(x,x')$ depends on $x$ and $x'$ only through the combination $(x-x')^2$ is a consequence of Lorentz invariance and the homogeneity of flat space-time. We shall see later that in a curved space-time the dependence of $G(x,x')$ on $x$ and $x$' will not be so simple.

When $(x-x')^2 <0$ it is convenient to introduce the new variables 
\begin{eqnarray}
&& z^2 = -m^2 (x-x')^2 >0, \; \; \; z>0 \\
&& u = -2im^2 \tfrac{s}{z},
\end{eqnarray}
which convert \eqref{Feyn_Gr} to 
\begin{eqnarray}
G(x,x') &=& \frac{1}{16\pi^2} \int_{0} ^{-i\infty} du \;  \left( -\frac{z}{2im^2}\right) \left(-\frac{4m^4}{u^2 z^2} \right) e^{-i\left(-m^2 \tfrac{uz}{2im^2}+\tfrac{im^2(x-x')^2}{2uz}\right) } \nonumber \\
&=& \frac{1}{16\pi^2} \int_{-i\infty} ^{0} du \; \left( \frac{iz}{2m^2}\right) \left(\frac{4m^4}{u^2 z^2} \right) e^{\left( \tfrac{uz}{2}+\tfrac{m^2(x-x')^2}{2uz}\right) } \nonumber \\
&=& \frac{1}{16\pi^2} \frac{2im^2}{z} \int_{-i\infty} ^{0} du \; \frac{1}{u^2} e^{\tfrac{z}{2}\left( u-\tfrac{1}{u}\right) } \nonumber \\
&=& \frac{im^2}{8\pi^2} \frac{1}{z} \int_{-i\infty} ^{0} du \; \frac{1}{u^2} e^{\tfrac{z}{2}\left( u-\tfrac{1}{u}\right) }
\end{eqnarray}
The contour of integration may be deformed in the manner shown in Fig. \ref{figure2}, and in virtue of the well known integral representation
\begin{equation}
H_1 ^{(2)}(z) = \frac{1}{i\pi} \int_C du \; \frac{1}{u^2} e^{\tfrac{z}{2}\left( u-\tfrac{1}{u}\right) }
\end{equation}
of the Hankel function of the second kind, of order $1$, we finally have
\begin{equation}
G(x,x') = -\frac{m^2}{8\pi} \frac{H_1 ^{(2)}(z)}{z}. \label{Feynman_Hankel}
\end{equation} 
\begin{figure}[htbp]
	\includegraphics[scale=0.40]{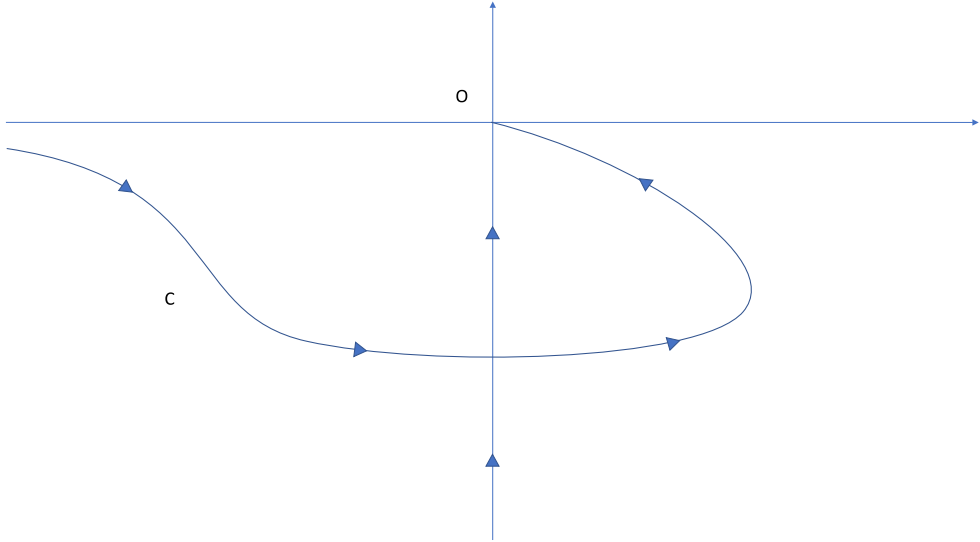}
	\vspace*{8pt}
	\caption{Contour for the Hankel function $H_1 ^{(2)}(z)$.}
	\label{figure2}
\end{figure}

\subsection{Series expansions and singularities}
For small values of $z$ (i.e., near the \textit{light cone}) it is convenient to use the power series expansions 
\begin{equation}
H_1 ^{(2)}(z) = J_1(z) -iY_1(z),
\end{equation}
\begin{eqnarray}
J_1(z) &=& \frac{z}{2} - \frac{z^3}{2^2 4} + \frac{z^5}{2^2 4^2 6} + ..., \\
Y_1(z) &=& \frac{2}{\pi} \bigg[-\frac{1}{z} J_0 (z) + (\gamma + {\rm log} \tfrac{z}{2}) J_1(z) \nonumber \\&& -\frac{z}{2} +\frac{z^3}{2^2 4}(1+\tfrac{1}{2})-\frac{z^5}{2^2 4^2 6}(1+\tfrac{1}{2} +\tfrac{1}{3})+...\bigg], \\
&& J_0 (z)= 1- \frac{z^2}{2^2}+\frac{z^4}{2^2 4^2} - ..., \\
&& \gamma = 0,5772... \; .
\end{eqnarray}
Remembering that analytic continuation should be performed in the lower half $z^2$ plane, and making use of the identities, which hold  $\forall z^2 \in \mathcal{R}$:
\begin{eqnarray}
\lim_{\epsilon \to 0^+} \frac{1}{z^2 - i \epsilon} &=& \frac{1}{z^2} + i\pi \delta(z^2), \\
\lim_{\epsilon \to 0^+}{\rm log}(z^2 - i \epsilon) &=& {\rm log} |z^2| -i\pi \theta(-z^2),
\end{eqnarray} 
we find, on splitting the Feynman propagator into its real and imaginary parts 
\begin{eqnarray}
&&\bar{G}(x,x')={\rm Re}(G)(x,x') = \frac{1}{4\pi}\delta((x-x')^2)  \nonumber \\&&\; \; - \frac{m^2}{8\pi}\theta(-(x-x')^2)\left[ \tfrac{1}{2}+\frac{m^2 (x-x')^2}{2^2 4} + \frac{m^4 (x-x')^4}{2^2 4^2 6}+... \right], 	\label{exp_real_F_flat}\\
&&G^{(1)}(x,x')= 2 {\rm Im}(G)(x,x') = \frac{m^2}{2\pi^2}\bigg\{ \frac{1}{m^2 (x-x')^2} \nonumber \\&&\; \;+ \left[\gamma - {\rm log}2 + {\rm log}m +\tfrac{1}{2}{\rm log}|(x-x')^2| \right] \left[\tfrac{1}{2}+\frac{m^2 (x-x')^2}{2^2 4} + ... \right]  \nonumber \\ &&\; \; -\tfrac{1}{4}-\frac{m^2 (x-x')^2}{2^2 4}(1+\tfrac{1}{4})+\frac{m^4 (x-x')^4}{2^2 4^2 6}(1+\tfrac{1}{2} +\tfrac{1}{6})- ... \bigg\}. \label{exp_imag_F_flat}
\end{eqnarray}
The Green's function $G$ (and hence also $G^+$ and $G^-$) is seen to have a $\delta$-distribution type singularity on the light cone $[(x - x')^2 = 0]$ and to vanish outside the light cone $[(x - x')^2 > 0]$. It also vanishes inside the light cone when $m = 0$\footnote{This property no longer holds when space-time is curved.}. In this case we have 
\begin{eqnarray}
\bar{G}(x,x') &=& \frac{1}{4\pi}\delta((x-x')^2),\\
G^{(1)}(x,x') &=&  \frac{1}{2\pi^2  (x-x')^2},
\end{eqnarray}
whence, in virtue of eqs. \eqref{useful_f1} and \eqref{useful_f2} and the identity 
\begin{equation}
\delta(\xi^2 - a^2) = \frac{1}{2a}\left[\delta(\xi-a) + \delta(\xi+a)\right], \; \; \; a>0 \\
\end{equation}
we obtain
\begin{equation}
G^\pm (x,x') = \frac{1}{4\pi} \frac{1}{|\vec{x}-\vec{x'}|}\delta(x^0 - x'^0 \pm |\vec{x}-\vec{x'}|).
\end{equation}

\subsection{Curved space-time formalism}
Consider now the scalar field theory obtained by applying the \textit{minimal coupling} to the gravitational field to the field theory examined in the previous sections:
\begin{itemize}
	\item[1.] Replace the Minkowski metric $\eta_{\mu\nu}$ by $g_{\mu\nu}$. 
	\item[2.] Replace ordinary space-time derivatives by the covariant derivatives associated to the (unique) Levi-Civita connection determined by the metric.
	\item[3.] Multiply the Lagrange function by $|\mathrm{g}|^{1/2}$, where $\mathrm{g} = {\rm det}(g_{\mu\nu})$. 
\end{itemize}
Then the action functional for the theory is
\begin{equation}
S[\phi] = -\tfrac{1}{2} \int dx \; |\mathrm{g}|^{1/2} \left(\phi_;^\mu \phi _{;\mu} + m^2 \phi^2 \right).
\end{equation}
Hence
\begin{equation}
S_1 [\phi]=  |\mathrm{g}|^{1/2} \left(\tensor{\phi}{_; ^\mu _\mu}  - m^2 \phi \right),
\end{equation}
and 
\begin{eqnarray}
\tensor[_1]{S}{_1} &=& |\mathrm{g}(x)|^{1/2} \left(\tensor{\delta(x,x')}{_; ^\mu _\mu}  - m^2 \delta(x,x') \right) \nonumber \\
&=& |\mathrm{g}(x)|^{1/2} \left(\nabla^\mu \nabla_\mu  - m^2  \right)\delta(x,x').
\end{eqnarray}
Therefore the equation for the Green's functions is
\begin{equation}
|\mathrm{g}(x)|^{1/2} \left(\nabla^\mu \nabla_\mu  - m^2  \right)G(x,x') = -\delta(x,x'). \label{scalar_Gr_eq_curved}
\end{equation}
A ve\-ry el\-eg\-ant me\-th\-od for so\-lvi\-ng th\-is eq\-ua\-ti\-on ex\-ist\-s, wh\-ich is due to Sch\-win\-ger. One re\-gar\-ds the Green's function as the matrix element of an operator $G$ in an abstract (nonphysical) Hilbert space: 
\begin{equation}
G(x,x') = \left\langle x | G | x' \right\rangle, 
\end{equation}
the basis vectors $| x' \rangle$ being eigenvectors of a commuting set of Hermitian operators $x^\mu$
\begin{equation}
x^\mu | x' \rangle =  x'^\mu | x' \rangle, \; \; \; \left\langle x'' | x' \right\rangle = \delta(x'',x').
\end{equation}
The differential eq. \eqref{scalar_Gr_eq_curved} may then be recast in the operator form 
\begin{equation}
(p_\mu |\mathrm{g}|^{1/2} g^{\mu\nu}p_\nu + m^2 |\mathrm{g}|^{1/2}) G = 1, \label{scalar_Gr_eq_curved_operator}
\end{equation}
where the $p_\mu$ are Hermitian operators which satisfy the commutation relations
\begin{equation}
[x^\mu, p_\nu] = i \underline{\delta}^\mu _\nu, \; \; \; [p_\mu,p_\nu]=0.
\end{equation} 

\subsection{General definition of the Feynman propagator}
In order to solve the operator equation \eqref{scalar_Gr_eq_curved} we must first decide which Green's function we want. As in the previous sections, we shall choose the Feynman propagator as the basic Green's function of interest. However, this immediately begs the question of what we \textit{mean} by the Feynman propagator when space-time is curved and non-empty. In a flat empty space-time the Feynman propagator can be defined as that Green's function which propagates positive frequencies into the future and negative frequencies into the past (see eq. \eqref{useful_def_Feynman}). The same definition can be used when space-time is curved provided it becomes asymptotically flat at large space-like and time-like distances and the words \lq\lq future'' and \lq\lq past'' are replaced by \lq\lq remote future'' and \lq\lq remote past'' respectively. Under these circumstances the same variational law holds for the Feynman propagator as well as the retarded and advanced Green's functions: 
\begin{equation}
\delta G^{ij} = G^{ik} \; \tensor[_k]{\delta F}{_l} \;  G^{lj}, \label{good_var_law}
\end{equation}
for this law immediately permits the expansion about the flat-empty-space-time values, $\tensor[^0]{F}{}$ and $\tensor[^0]{G}{}$, of the operators $F$ and $G$, 
\begin{eqnarray}
G - \tensor[^0]{G}{} &=& \delta G = ( \tensor[^0]{G}{} + \delta G )(F- \tensor[^0]{F}{})( \tensor[^0]{G}{} + \delta G) \nonumber \\
&=&  \tensor[^0]{G}{} \; U \; \tensor[^0]{G}{} + \delta G \; F \; \tensor[^0]{G}{} + \tensor[^0]{G}{} \; F \;  \delta G + \delta G \; F \; \delta G \nonumber \\
&=&  \tensor[^0]{G}{} \; U \; \tensor[^0]{G}{} + \tensor[^0]{G}{} \; U \; \tensor[^0]{G}{} U \; \tensor[^0]{G}{} + ...
\end{eqnarray}
where $U \equiv F- \tensor[^0]{F}{}$; hence one obtains
\begin{eqnarray}
G &=& \tensor[^0]{G}{} +  \tensor[^0]{G}{} \; U \; \tensor[^0]{G}{} + \tensor[^0]{G}{} \; U \; \tensor[^0]{G}{} \; U \; \tensor[^0]{G}{} + ... \nonumber \\
&=&  \tensor[^0]{G}{} \; (1- U \tensor[^0]{G}{})^{-1} \\
&=& (1- \tensor[^0]{G}{} U )^{-1} \; \tensor[^0]{G}{}.
\end{eqnarray}
We see that the first $\tensor[^0]{G}{}$ standing on the left and the last $\tensor[^0]{G}{}$ standing on the right, in each term of the expansion, do indeed ensure that ultimately only pure positive frequencies are found in the remote future and pure negative frequencies in the remote past, owing to the effectively limited domain over which $U$ is non-vanishing. 

A word is perhaps in order at this point regarding the very special properties the Feynman propagator possesses. When $F$ is symmetric (and we have always assumed it is) the Feynman propagator is symmetric. Since it also satisfies the variational law \eqref{good_var_law} \textit{it is the only Green's function which, when regarded as a continuous matrix, obeys all the rules of finite matrix theory}. In a certain sense it may therefore be regarded as \textit{the} inverse of the matrix $(-\tensor[_i]{F}{_j})$. In flat space-time its special properties stem from the fact (already noted) that it may be obtained by analytic continuation from the \textit{unique} inverse which $(-\tensor[_i]{F}{_j})$ possesses in a Euclidean space. When space-time is curved these properties may themselves be used to define the Feynman propagator \textit{even when space-time is not asymptotically flat}. 

\subsection{Integral representation (I)}
From the results of the previous sections, we shall obtain the Feynman propagator, in curved space-times as well as flat, simply by giving the mass parameter $m$ an infinitesimal negative imaginary part. This has the effect of rendering the operator in \eqref{scalar_Gr_eq_curved_operator} nonsingular so that inverses may be taken in a simple and direct fashion. It also emphasizes once again that Green's functions are boundary values of analytic functions. Multiplying equation \eqref{scalar_Gr_eq_curved_operator} on the left by $|\mathrm{g}|^{-1/4}$ and on the right by $|\mathrm{g}|^{1/4}$, we obtain 
\begin{eqnarray}
&&|\mathrm{g}|^{-1/4}(p_\mu |\mathrm{g}|^{1/2} g^{\mu\nu}p_\nu + m^2 |\mathrm{g}|^{1/2}) G |\mathrm{g}|^{1/4}= |\mathrm{g}|^{-1/4}|\mathrm{g}|^{1/4}, \nonumber \\
&&|\mathrm{g}|^{-1/4}p_\mu |\mathrm{g}|^{1/2} g^{\mu\nu}p_\nu G |\mathrm{g}|^{1/4} +  m^2 |\mathrm{g}|^{1/4} G |\mathrm{g}|^{1/4} = 1, \nonumber \\
&&|\mathrm{g}|^{-1/4}p_\mu |\mathrm{g}|^{1/2} g^{\mu\nu}p_\nu |\mathrm{g}|^{-1/4}|\mathrm{g}|^{1/4} G |\mathrm{g}|^{1/4} +  m^2 |\mathrm{g}|^{1/4} G |\mathrm{g}|^{1/4} = 1, \nonumber \\
&&(|\mathrm{g}|^{-1/4}p_\mu |\mathrm{g}|^{1/2} g^{\mu\nu}p_\nu |\mathrm{g}|^{-1/4}+  m^2)|\mathrm{g}|^{1/4} G |\mathrm{g}|^{1/4} = 1, \nonumber \\
&&(H+m^2)|\mathrm{g}|^{1/4} G |\mathrm{g}|^{1/4} = 1,
\end{eqnarray}
where 
\begin{equation}
H \equiv |\mathrm{g}|^{-1/4}p_\mu |\mathrm{g}|^{1/2} g^{\mu\nu}p_\nu |\mathrm{g}|^{1/4}.
\end{equation}
Therefore, with the correct prescription for the Feynman propagator:
\begin{equation}
|\mathrm{g}|^{1/4} G |\mathrm{g}|^{1/4} = \frac{1}{H+m^2 -i\epsilon} = i\int _0 ^{+\infty} ds \; e^{-is(H+m^2)}.
\end{equation}
Taking matrix elements of the previous equation we obtain 
\begin{eqnarray}
|\mathrm{g(x')}|^{1/4} G |\mathrm{g}(x'')|^{1/4} &=& i \int _0 ^{+\infty} ds \; \left\langle x' |e^{-isH}|x'' \right\rangle e^{-ism^2}, \nonumber \\
|\mathrm{g'}|^{1/4} G |\mathrm{g}''|^{1/4}  &=&  i \int _0 ^{+\infty} ds \; \left\langle x',s |x'',0 \right\rangle e^{-ism^2}, \label{corr_prescr_F}
\end{eqnarray}
with 
\begin{equation}
\left\langle x',s |x'',0 \right\rangle \equiv \left\langle x' |e^{-isH}|x'' \right\rangle.
\end{equation}
Thus we are led to an associated dynamical problem governed by the \lq\lq Hamiltonian'' $H$. 

The \lq\lq transition amplitude'' $\left\langle x',s |x'',0 \right\rangle$ satisfies the Schr\"{o}dinger equation 
\begin{equation}
i \frac{\partial}{\partial s} \left\langle x',s |x'',0 \right\rangle = \left\langle x',s |H|x'',0 \right\rangle = -\tensor{\left\langle x',s |x'',0 \right\rangle}{_{;\mu'}^{\mu'}} \label{curved_ampl}
\end{equation} 
and the boundary condition
\begin{equation}
\left\langle x',0 |x'',0 \right\rangle = \delta(x',x''). \label{curved_ampl_in}
\end{equation} 
In flat empty space-time, this equation is solved by 
\begin{equation}
\left\langle x',s |x'',0 \right\rangle _{{\rm Minkowski}} = \frac{-i}{16 \pi^2} \frac{1}{s^2} e^{i \frac{(x'-x'')^2}{4s}}, \label{flat_ampl}
\end{equation}
which agrees with \eqref{Feyn_Gr}.

In order to discuss the generalization to curved space-time, some insight on auxiliary geometric quantities is necessary.

\subsection{Auxiliary geometric quantities}
\subsubsection{$k$-point tensors}
As is well known, a $(r,s)$-tensor field $T$ on a manifold $M$ is a map which assigns to every point $p \in M$ an element from the direct product of the tangent space $T_p M$, taken $r$ times, and the cotangent space $T_p ^* M$, taken $s$ times:
\begin{eqnarray}
&& T: M \rightarrow \left( TM \right)^r \otimes  \left( T^* M \right)^s ,  \nonumber \\
&& p \mapsto T(p) = \tensor{T}{^{\mu _1...\mu _r}_{\nu _1...\nu _s}}(p)\frac{\partial}{\partial x^{\mu _1}}\bigg |_p \otimes ... \otimes \frac{\partial}{\partial x^{\mu _r}}\bigg |_p \otimes dx^{\nu _1}\bigg |_p \otimes ... \otimes  dx^{\nu _s}\bigg |_p. \nonumber \\
&& 
\end{eqnarray}
The tensor field concept can be generalized in this way: we shall call $k$-point tensor on a manifold $M$ a map which assigns to every $k$-tuple of points in $M$ an element from the direct product of the tensor spaces built upon those points:
\begin{eqnarray}
&& T: M^k \rightarrow \left[\left( TM \right)^{r_1} \otimes  \left( T^* M \right)^{s_1} \right] \otimes ... \otimes \left[\left( TM \right)^{r_k} \otimes  \left( T^* M \right)^{s_k} \right]  ,\nonumber \\ \nonumber \\
&& (p^{(1)},p^{(2)},...,p^{(k)})  \mapsto \nonumber \\
&& \; \; \; \; \tensor{T}{^{\mu^{(1)} _1...\mu^{(1)} _r}_{\nu^{(1)} _1...\nu^{(1)} _s} ^{...} _{...} ^{\mu ^{(k)} _1...\mu^{(k)} _{r^{k}}} _{\nu^{(k)} _1...\nu^{(k)} _{s^{(k)}}}  }(p^{(1)},p^{(2)},...,p^{(k)}) \cdot \nonumber \\ && \; \; \; \; \; \; \cdot \tfrac{\partial}{\partial x^{\mu^{(1)} _1}}\bigg |_{p^{(1)}} \otimes ... \otimes \tfrac{\partial}{\partial x^{\mu^{(1)} _{r^{(1)}}}}\bigg |_{p^{(1)}} \otimes dx^{\nu^{(1)} _1}\bigg |_{p^{(1)}} \otimes ... \otimes  dx^{\nu^{(1)} _{s^{(1)}}}\bigg |_{p^{(1)}} \otimes ... \nonumber \\ && \; \; \; \; \; \; \; \;    \otimes \tfrac{\partial}{\partial x^{\mu^{(k)} _1}}\bigg |_{p^{(k)}} \otimes ... \otimes \tfrac{\partial}{\partial x^{\mu^{(k)} _{r^{(k)}}}}\bigg |_{p^{(k)}} \otimes dx^{\nu ^{(k)} _1}\bigg |_{p^{(k)}} \otimes ... \otimes  dx^{\nu^{(k)} _{s^{(k)}}}\bigg |_{p^{(k)}} . \nonumber \\
\end{eqnarray} 
Roughly speaking, whenever $k-1$ points (however chosen) are held fixed, a $k$-point tensor becomes a tensor field.

\subsubsection{Geodesics}
For a detailed discussion on geodesics on a Riemannian manifold, see  \cite{milnor1969morse}; for a Lorentzian manifold, see \cite{hawking1973large}. Here we will only introduce tools necessary for later treatise.
As is well known, given a connection $\nabla$ on a manifold $M$, there is exactly one parallel transport on $M$, i.e., exactly one way to parallel transport a given vector along any curve; one shall define \textit{geodesic} (associated to that connection) any curve whose tangent vector is parallel transported along the curve. If $M$ is a (pseudo-) Riemannian manifold with metric tensor $g$, and $\nabla$ is the unique Levi-Civita connection associated to that metric, then, given two close enough points, the curve for which the length functional (associated to the metric $g$) is stationary (on the curves which connect those points) is a geodesic.

In fact the equations for a geodesic $x(\tau)$ with affine parameter, 
\begin{equation}
\ddot{x}^\mu (\tau) + \Gamma^\mu _{\rho\sigma}(x(\tau))\dot{x}^\rho (\tau) \dot{x}^\sigma (\tau) =0 
\end{equation}
are precisely the Euler-Lagrange equations associated to the functional $$\underline{S}[x(\tau)] = {\rm Length }[x(\tau)] =\int d\tau \; \underline{L}(x(\tau),\dot{x}(\tau) =\int d\tau \; \left[\pm\dot{x}^\rho g_{\rho\sigma} \dot{x}^\sigma\right]^{1/2}$$:
\begin{eqnarray}
&&\frac{\delta \underline{S}}{\delta x^\mu} = \frac{\partial \underline{L}}{\partial x^\mu} - \frac{d}{d\tau} \frac{\partial \underline{L}}{\partial \dot{x}^\mu} = \nonumber \\
&=& \tfrac{1}{2} \left[\pm\dot{x}^\rho g_{\rho\sigma} \dot{x}^\sigma\right]^{-1/2} \left[\pm g_{\rho\sigma ,\mu}\dot{x}^\rho\dot{x}^\sigma \mp  \frac{d}{d\tau} \left( 2 g_{\mu\sigma} \dot{x}^\sigma  \right)\right] \nonumber \\
&=& \mp \frac{1}{\underline{L}} \left[\frac{d}{d\tau} \left( g_{\mu\sigma} \dot{x}^\sigma    \right)  -  \tfrac{1}{2}   g_{\rho\sigma ,\mu}\dot{x}^\rho\dot{x}^\sigma  \right]\nonumber \\
&=& \mp \frac{1}{\underline{L}} \left( g_{\mu\sigma} \ddot{x}^\sigma +\frac{d}{d\tau} \left( g_{\mu\sigma} \right) \dot{x}^\sigma - \tfrac{1}{2}   g_{\rho\sigma ,\mu}\dot{x}^\rho\dot{x}^\sigma \right)\nonumber \\
&=& \mp \frac{1}{\underline{L}}\left(g_{\mu\sigma} \ddot{x}^\sigma + g_{\mu\sigma ,\rho} \dot{x}^\rho \dot{x}^\sigma -\tfrac{1}{2}   g_{\rho\sigma ,\mu}\dot{x}^\rho\dot{x}^\sigma \right) \nonumber \\
&=& \mp \frac{1}{\underline{L}}\left( g_{\mu\sigma} \ddot{x}^\sigma + \tfrac{1}{2}g_{\mu\sigma ,\rho} \dot{x}^\rho \dot{x}^\sigma + \tfrac{1}{2}g_{\mu\rho ,\sigma} \dot{x}^\rho \dot{x}^\sigma - \tfrac{1}{2}   g_{\rho\sigma ,\mu}\dot{x}^\rho\dot{x}^\sigma \right) = 0, \nonumber \\
&&g_{\mu\sigma} \ddot{x}^\sigma + \tfrac{1}{2}g_{\mu\sigma ,\rho} \dot{x}^\rho \dot{x}^\sigma + \tfrac{1}{2}g_{\mu\rho ,\sigma} \dot{x}^\rho \dot{x}^\sigma - \tfrac{1}{2}   g_{\rho\sigma ,\mu}\dot{x}^\rho\dot{x}^\sigma  = 0.
\end{eqnarray}
Raising the index $\mu$, we obtain
\begin{equation}
\ddot{x}^\mu + \tfrac{1}{2} g^{\mu\nu}\left(g_{\nu\sigma ,\rho} +g_{\nu\rho ,\sigma} -   g_{\rho\sigma ,\nu}\right)\dot{x}^\rho \dot{x}^\sigma = 0,
\end{equation}
which are exactly the geodesic equations, being
\begin{equation}
\Gamma ^{\mu}_{\sigma\rho} = \tfrac{1}{2} g^{\mu\nu}\left(g_{\nu\sigma ,\rho} +g_{\nu\rho ,\sigma} -  g_{\rho\sigma ,\nu}\right).
\end{equation}
As is straightforward to verify, the same equations are obtained considering the functional $S[x(\tau)] = \int d\tau \; L(x(\tau),\dot{x}(\tau))=\int d\tau \; \tfrac{1}{2} \dot{x}^\rho g_{\rho\sigma} \dot{x}^\sigma$.

Since a variational principle has been introduced, the theory of geodesics may be viewed \textit{formally} as a dynamical theory, and all the results of Hamilton-Jacobi theory can be immediately applied to it. The \lq\lq conjugate momenta'' and \lq\lq Hamiltonian'' are given by
\begin{eqnarray}
p_\mu &\equiv &\frac{\partial L}{\partial \dot{x}^\mu} =  g_{\mu\nu} \dot{x}^\nu  \equiv  \dot{x}_\mu, \\
H &\equiv & p_\mu \dot{x}^\mu - L = \dot{x}_\mu \dot{x}^\mu - \tfrac{1}{2} \dot{x}^\mu g_{\mu\nu} \dot{x}^\nu = L. 
\end{eqnarray}  
Hence we are led to the \lq\lq energy integral'' for the geodesics:
\begin{equation}
H = L = \tfrac{1}{2} \dot{x}^\rho g_{\rho\sigma} \dot{x}^\sigma = \tfrac{1}{2} \left( \frac{ds}{d\tau}\right)^2 = {\rm const},
\end{equation}
where $s$ is the arc length defined on the curve; the action functional on a solution, i.e., a geodesic, whose endpoints are $x(\tau) \equiv x$, $x(\tau ') \equiv x'$, reduces to
\begin{eqnarray}
S(x,\tau |x',\tau ') &=& \int_{\tau '} ^{\tau} d\tau '' \; \tfrac{1}{2} \dot{x}^\rho g_{\rho\sigma} \dot{x}^\sigma  \nonumber \\
&=& \int_{\tau '} ^{\tau} d\tau '' \; \tfrac{1}{2} \left( \frac{ds}{d\tau ''}\right)^2 \nonumber \\
&=& \tfrac{1}{2} \int_{x'} ^{x} ds \frac{d\tau ''}{ds} \; \left( \frac{ds}{d\tau ''}\right)^2 \nonumber \\
&=& \tfrac{1}{2} \int_{x'} ^{x} ds \; \left( \frac{ds}{d\tau ''}\right) \nonumber \\
&=& \tfrac{1}{2} \left( \frac{ds}{d\tau }\right) \int_{x'} ^{x} ds \nonumber \\
&=& \tfrac{1}{2} \left( \frac{ds}{d\tau }\right) \left(s(x)-s(x')\right) \nonumber \\
&=& \tfrac{1}{2} \left( \frac{s(x)-s(x')}{\tau -\tau '}\right) \left(s(x)-s(x')\right) \nonumber \\
&=& \frac{\sigma (x,x')}{\tau -\tau '},
\end{eqnarray}
where the \textit{bi-scalar} $\sigma (x,x')$, which we shall call \textit{geodetic interval} or \textit{world function}, is equal to one half the square of the distance along the geodesic between $x$ and $x'$. 

The bi-scalar of geodetic interval satisfies an important differential equation which follows immediately from the Hamilton-Jacobi equation for the action $S$; we have
\begin{eqnarray}
p_\mu &=& \frac{\partial S}{\partial x^\mu} = \frac{\sigma _{;\mu}}{\tau -\tau '}, \\
0 &=& \frac{\partial S}{\partial \tau} + H = -\frac{\sigma (x,x')}{(\tau -\tau ')^2} + \tfrac{1}{2} p_\mu p^\mu ,
\end{eqnarray}
where $p_\mu$ is now the \lq\lq momentum'' at $x$ corresponding to the geodesic defined by the endpoints $x$, $x'$; therefore the world function is the solution of the Cauchy problem
\begin{equation}
\begin{cases} 
\tfrac{1}{2} \sigma _{;\mu}\sigma ^{;\mu} = \sigma,  \label{ham_jac_eq}\\
\sigma (x',x') = 0.
\end{cases}
\end{equation}
Obviously, the Hamilton-Jacobi equation holds on the other endpoint too; then
\begin{equation}
\tfrac{1}{2}\sigma _{;\mu}\sigma ^{;\mu} =\tfrac{1}{2}\sigma _{;\mu '}\sigma ^{;\mu '}= \sigma.
\end{equation}
In other words, $\sigma _{;\mu}$ is a vector of length equal to the distance along the geodesic between $x$ and $x'$, tangent to the geodesic at $x$, and oriented in the direction $x' \rightarrow x$, while $\sigma _{;\mu'}$ is a vector of equal length, tangent to the geodesic at $x'$, and oriented in the opposite direction. The geodetic interval itself is obviously a symmetric function of $x$ and $x'$: 
\begin{equation}
\sigma(x,x')=\sigma(x',x).
\end{equation}

\subsubsection{Caustic surfaces}
In a general Riemannian manifold the geodetic interval is not single-valued, except when $x$ and $x'$ are sufficiently close to one another. The geodesics emanating from a given point will often, beyond a certain distance, begin to cross over one another. The locus of points at which the onset of overlap occurs forms an envelope of the family of geodesics, known as a \textit{caustic surface}. The equation for the caustic surface relative to a given point can be expressed in terms of the quantity ${\rm det } (\sigma_{;\mu\nu'})$.

In fact, a geodesic can be specified by means of its endpoints or by means of one of its endpoints together with a tangent vector at that point having a length equal to the length of the geodesic. Therefore we can vary $\sigma_{;\mu'}$ helding $x'$ fixed, and evaluate the resulting variation in $x$; it is straightforward to obtain
\begin{equation}
\delta \sigma_{;\mu'} =  \sigma_{;\mu '\nu} \delta x^\nu;
\end{equation}
therefore
\begin{equation}
\delta x^\mu = - D^{-1 \mu \nu'} \delta \sigma_{;\nu '}, \label{caust_eq}
\end{equation}
where $D^{-1 \mu \nu '}$ is the inverse transpose of the finite matrix having the elements $D_{\rho \nu '} = - \sigma_{;\nu '\rho}$, i.e.,
\begin{equation}
D^{-1 \mu \nu'} D_{\rho \nu '} = - D^{-1 \mu \nu'} \sigma_{;\nu '\rho} = \underline{\delta}^{\mu}_{\rho}.
\end{equation}
When $D^{-1 \mu \nu '}$ is a singular matrix, it is possible to choose a variation in $\sigma_{;\mu '}$ which produces no variation in the $x$. The point $x$ then lies on the caustic surface relative to $x'$, and the condition for this is evidently $D^{-1} = 0$, where 
\begin{equation}
D = -{\rm det } (D_{\mu\nu'}),
\end{equation}
the minus sign expressing a convention appropriate to the metric of space-time. In $4$-dimensional space-time the caustic surface will usually be a $3$-dimensional hypersurface, but degenerate forms having fewer dimensions, including zero (focal points) can occur. It will be noted that variations of $\sigma_{;\mu '}$ which leave $x$ unchanged must be orthogonal to $\sigma_{;\mu '}$; that is, the length of the geodesic itself must remain unchanged. This may be inferred by taking the derivative of the Hamilton-Jacobi equation \eqref{ham_jac_eq}:
\begin{equation}
\sigma^{;\mu }\sigma_{;\mu\nu ' }= \sigma_{;\nu '}, \label{sec_der_sigma}
\end{equation} 
and recasting it in the form 
\begin{equation}
-D^{-1 \mu \nu '}\sigma_{;\nu '} = \sigma^{;\mu } \neq 0,
\end{equation}
which shows, together with \eqref{caust_eq} that changing the length of $\sigma^{;\mu ' }$ without changing its direction necessarily shifts $x$ a proportional distance: in fact, by taking $ \delta \sigma_{;\mu '} = \epsilon \sigma_{;\mu '}$, one obtains
\begin{eqnarray}
\delta x^\mu &=& - D^{-1 \mu \nu'} \delta \sigma_{;\nu '} \nonumber \\
&=& - \epsilon D^{-1 \mu \nu'} \sigma_{;\nu '} \nonumber \\
&=& - \epsilon \sigma^{;\mu }.
\end{eqnarray}
\subsubsection{Divergence of geodesics}
The determinant $D$ is a bi-density, of unit weight at both $x$ and $x'$. Not surprisingly it plays a fundamental role in the description of the \textit{rate} at which geodesics emanating from fixed points diverge from or converge toward one another. If we differentiate eq. \eqref{sec_der_sigma} with respect to $x^\rho$ and we note that the indices $\mu$ and $\rho$ commute, we get 
\begin{eqnarray}
\tensor{\sigma}{_; ^\mu _\rho} \sigma_{;\mu\nu ' } + \sigma^{;\mu }\sigma_{;\mu\nu '\rho } &=& \sigma_{;\nu '\rho}, \nonumber \\
\tensor{\sigma}{_; ^\mu _\rho} \sigma_{;\mu\nu ' } + \sigma^{;\mu }\sigma_{;\rho\nu '\mu } &=& \sigma_{;\nu '\rho}, \nonumber \\
-\tensor{\sigma}{_; ^\mu _\rho} D_{ \mu \nu'} -\sigma^{;\mu }D_{\nu' \rho;\mu }&=& - D_{\rho\nu '}, \nonumber \\
D_{\rho\nu '} &=& \tensor{\sigma}{_; ^\mu _\rho} D_{ \mu \nu'} + \sigma^{;\mu }D_{\nu' \rho;\mu },
\end{eqnarray} 
which, on multiplication by $D^{-1 \rho \nu'}$, gives 
\begin{eqnarray}
D_{\rho\nu '}D^{-1 \rho \nu'} &=& \tensor{\sigma}{_; ^\mu _\rho} D_{ \mu \nu'}D^{-1 \rho\nu'} + \sigma^{;\mu }D_{\nu' \rho;\mu }D^{-1 \rho \nu'} \nonumber \\
-\underline{\delta}^\rho _\rho &=&  -\tensor{\sigma}{_; ^\mu _\rho} \underline{\delta}^\rho _\mu   -   \sigma^{;\mu } D D_{;\mu} \nonumber \\
4 &=& \tensor{\sigma}{_; ^\mu _\mu} + \sigma^{;\mu } D^{-1} D_{;\mu}  \label{div_geod_0}\\
D^{-1} \left(D  \sigma^{;\mu }\right)_{;\mu} &=& 4. \label{div_geod}
\end{eqnarray}
The significance of this equation may be made transparent by first replacing $D$ with the bi-scalar 
\begin{equation}
\Delta \equiv |{\rm g}| ^{-1/2}  D  |{\rm g'}| ^{-1/2} 
\end{equation}
and observing that the operator $\sigma^{;\mu }\partial _{\mu}$ gives the derivative of any function along the geodesic from $x'$. Thus
\begin{equation}
\sigma^{;\mu }\partial _{\mu} f = (\tau-\tau')\dot{f}
\end{equation}
where $f$ is any scalar. Arbitrarily setting $\tau' = 0$, we may recast eq. \eqref{div_geod} in the form 
\begin{equation}
\tensor{\sigma}{_; ^\mu _\mu} = 4 - \frac{d(\log \Delta)}{d (\log \tau)}.  \label{div_geod2}
\end{equation}
In fact 
\begin{eqnarray}
\sigma^{;\mu } D^{-1} D_{;\mu} &=&  \sigma^{;\mu } \left( |{\rm g}| ^{1/2}  \Delta  |{\rm g'}| ^{1/2} \right)^{-1} \left( |{\rm g}| ^{1/2}  \Delta  |{\rm g'}| ^{1/2} \right)_{;\mu} \nonumber \\
&=& \sigma^{;\mu } |{\rm g}| ^{-1/2}  \Delta^{-1}  |{\rm g'}| ^{-1/2}  |{\rm g}| ^{1/2}  \Delta_{;\mu}  |{\rm g'}| ^{1/2} \nonumber \\
&=& \sigma^{;\mu } \Delta^{-1}\Delta_{;\mu}  \nonumber \\
&=& \sigma^{;\mu } ({\rm log} \Delta) _{;\mu} \nonumber \\
&=& \tau \frac{d}{d\tau }({\rm log} \Delta) \nonumber \\
&=& \frac{d(\log \Delta)}{d (\log \tau)}.
\end{eqnarray}
From \eqref{div_geod2} it follows immediately that $\Delta$ increases or decreases along each geodesic from $x'$ according as the rate of divergence of the neighboring geodesics from $x'$, which is measured by $\tensor{\sigma}{_; ^\mu _\mu}$, is less than or greater than $4$, the rate in flat space-time. If the divergence rate becomes negatively infinite a caustic surface develops and $\Delta$ blows up. 

\subsubsection{Geodetic parallel displacement}
Another geometrical quantity of fundamental importance is the \textit{geodetic parallel displacement bi-vector}, $g_{\mu\nu'}$, which is defined by the differential equations 
\begin{equation}
\tensor{\sigma}{_; ^\eta} g_{\mu\nu';\eta} = 0 \label{par_displ}
\end{equation}
together with the boundary condition 
\begin{equation}
{\rm lim}_{x' \rightarrow x} g_{\mu\nu'} = g_{\mu\nu}. \label{par_displ_limit}
\end{equation}
The bi-vector $g_{\mu\nu'}$ gets its name from the fact that the result of applying it, for example, to a local contravariant vector $A^{\mu'}$ at $x'$, is to obtain the covariant form or the vector which results from displacing $A^{\mu'}$ in a parallel fashion along the geodesic from $x'$ to $x$. This follows from the defining eq. \eqref{par_displ}, which requires the covariant derivative of $g_{\mu\nu'}$ to vanish in directions tangent to the geodesic: in fact
\begin{eqnarray}
\frac{d}{d\tau}\left(g_{\mu\nu'} A^{\nu'} \right) &=& \frac{\tensor{\sigma}{_{;} ^{\rho}}}{\tau}\left(g_{\mu\nu'} A^{\nu'} \right)_{;\rho} \nonumber \\
&=& \tfrac{1}{\tau}\tensor{\sigma}{_{;} ^{\rho}}g_{\mu\nu';\rho}A^{\nu'}  \nonumber \\
&=&  0,
\end{eqnarray}
and 
\begin{eqnarray}
{\rm lim}_{x\rightarrow x'} g^{\mu\rho} g_{\rho\nu'} A^{\nu'}  &=& g^{\mu'\rho'} g_{\rho'\nu'} A^{\nu'} \nonumber \\
&=& \underline{\delta}_{\nu'}^{\mu'}A^{\nu'} \nonumber \\
&=&A^{\mu'}.
\end{eqnarray}
From its geometrical significance and the fact that tangents to a geodesic are self-parallel the following properties of $g_{\mu\nu'}$ are obvious:
\begin{eqnarray}
\tensor{g}{_\mu ^{\nu'}} \tensor{\sigma}{_{;\nu'}} = - \tensor{\sigma}{_{;\mu}}, \; \; \;&& \; \; \; \tensor{g}{^\nu _{\mu'}} \tensor{\sigma}{_{;\nu}} = - \tensor{\sigma}{_{;\mu'}}, \\
\tensor{\sigma}{_; ^{\eta'}} g_{\mu\nu';\eta'} &=& 0, \\
\tensor{g}{_\mu _{\nu'}} &=& \tensor{g}{ _{\nu'}  _\mu}, \\
\tensor{g}{_\mu _{\rho'}}\tensor{g}{_{\nu'} ^{\rho'} } = \tensor{g}{_\mu _\nu},  \; \; \;&& \; \; \; \tensor{g}{_\rho _{\mu'} }\tensor{g}{^\rho _{\nu'}} = \tensor{g}{_{\mu'} _{\nu'}}, \\
{\rm det}(-g_{\mu\nu'})&=&|{\rm g}|^{1/2} |{\rm g'}|^{1/2}.
\end{eqnarray}
In a similar manner one may define a \textit{geodetic parallel displacement bi-spinor} $I(x,x')$ which satisfies 
\begin{eqnarray}
\tensor{\sigma}{_; ^\mu} I_{;\mu} &=& 0, \label{par_displ_sc} \\
{\rm lim}_{x' \rightarrow x} I &=& {\rm unity \; matrix},  \label{par_displ_sc2}
\end{eqnarray}
and which transforms like $\psi \equiv \psi(x)$ at $x$ and like $\psi' \equiv \psi(x')$ at $x'$. 

\subsection{Integral representation (II)}
Now we are ready to propose an ansatz to solve \eqref{curved_ampl}, \eqref{curved_ampl_in}: in \eqref{flat_ampl}, replace
\begin{eqnarray}
\frac{1}{s^2} &\mapsto & \frac{D(x,x')^{1/2}}{s^2} \\
(x-x')^2 &\mapsto &  2 \sigma (x,x')
\end{eqnarray}
and multiplicate by a power series in $(is)$:
\begin{eqnarray}
&&\sum_{n=0}^{+\infty} a_n (x,x') (is)^n, \\
&&{\rm lim}_{x' \rightarrow x} a_0 (x,x') = 1,
\end{eqnarray}
Then the ansatz is
\begin{eqnarray}
&& \left\langle x',s |x'',0 \right\rangle = \frac{-i}{16 \pi^2} \frac{D(x,x')^{1/2}}{s^2} e^{i \frac{\sigma (x,x')}{2s}}\sum_{n=0}^{+\infty} a_n (x,x') (is)^n, \label{ansatz}\\
&&{\rm lim}_{x' \rightarrow x} a_0 (x,x') = 1.
\end{eqnarray}
A requirement for the ansatz to be meaningful is the existence of a recurrence relation for the unknown $a_n$'s. Inserting \eqref{ansatz} in \eqref{curved_ampl}, the lhs is:
\begin{eqnarray}
&& i \frac{\partial}{\partial s}\left\langle x',s |x'',0 \right\rangle = i \frac{\partial}{\partial s}\left(\frac{-i}{16 \pi^2} \frac{D^{1/2}}{s^2} e^{i \frac{\sigma}{2s}}\sum_{n=0}^{+\infty} a_n(is)^n\right) \nonumber \\
&&= \frac{D^{1/2}}{16 \pi^2} \frac{\partial}{\partial s} \left(\frac{1}{s^2}e^{i \frac{\sigma }{2s}}\sum_{n=0}^{+\infty} a_n (is)^n\right) \nonumber \\
&&= \frac{D^{1/2}}{16 \pi^2} e^{i \frac{\sigma }{2s}} \bigg( -\frac{2}{s^3}\sum_{n=0}^{+\infty} a_n (is)^n  +\frac{1}{s^2}  \frac{-i\sigma }{2s^2}\sum_{n=0}^{+\infty} a_n (is)^n\nonumber \\
&& \; \; \; \; \; +\frac{1}{s^2}\sum_{n=1}^{+\infty}in a_n (is)^{n-1} \bigg) \nonumber \\
&&=\frac{D^{1/2}}{16 \pi^2} e^{i \frac{\sigma}{2s}} \bigg(\sum_{n=0}^{+\infty} (-2i^{3}a_n )(is)^{n-3} +\sum_{n=0}^{+\infty} (-ii^4\sigma a_n /2)(is)^{n-4}  \nonumber \\ 
&& \; \; \; \; \;+\sum_{n=1}^{+\infty}(ii^2n a_n )(is)^{n-3}\bigg) \nonumber \\
&&= \frac{iD^{1/2}}{16 \pi^2} e^{i \frac{\sigma }{2s}}  \bigg(\sum_{n=0}^{+\infty}(2a_n )(is)^{n-3}+\sum_{n=0}^{+\infty} (-\sigma a_n /2)(is)^{n-4}  \nonumber \\
&& \; \; \; \; \; +\sum_{n=1}^{+\infty}(-n a_n)(is)^{n-3}\bigg),
\end{eqnarray} 
while the rhs is 
\begin{eqnarray}
&&\tensor{-\left\langle x',s |x'',0 \right\rangle}{_{;\mu}^{\mu}} = -\left( \frac{-i}{16 \pi^2} \frac{D^{1/2}}{s^2} e^{i \frac{\sigma }{2s}}\sum_{n=0}^{+\infty} a_n  (is)^n\right)\tensor{}{_{;\mu}^{\mu}}\nonumber \\
&&= \frac{i}{16 \pi^2}\frac{1}{s^2}\left( D^{1/2}e^{i \frac{\sigma }{2s}}\sum_{n=0}^{+\infty} a_n  (is)^n\right)\tensor{}{_{;\mu}^{\mu}} \nonumber \\
&&= \frac{i}{16 \pi^2}\frac{1}{s^2}e^{i \frac{\sigma }{2s}} \bigg(D^{1/2}\tensor{}{_{;\mu}^{\mu}}\sum_{n=0}^{+\infty} a_n  (is)^n \nonumber \\
&&\; \; \; \; \; +D^{1/2} \frac{i}{2s} \tensor{\sigma}{_{;\mu}^{\mu}}\sum_{n=0}^{+\infty} a_n  (is)^n +D^{1/2} \left(\frac{i}{2s}\right)^2\tensor{\sigma}{_{;\mu}}\tensor{\sigma}{^{;\mu}}\sum_{n=0}^{+\infty} a_n  (is)^n \nonumber \\
&&\; \; \; \; \; +D^{1/2}\sum_{n=0}^{+\infty} \tensor{a}{_n _{;\mu}^{\mu}} (is)^n   +2D^{1/2}\tensor{}{_{;\mu}}\left(\frac{i}{2s}\right)\tensor{\sigma}{^{;\mu}}\sum_{n=0}^{+\infty} a_n  (is)^n  \nonumber \\
&&\; \; \; \; \; +2D^{1/2}\tensor{}{_{;\mu}}\sum_{n=0}^{+\infty} \tensor{a }{_n^{;\mu}} (is)^n   +2D^{1/2}\left(\frac{i}{2s}\right)\tensor{\sigma}{_{;\mu}}\sum_{n=0}^{+\infty}\tensor{a }{_n^{;\mu}} (is)^n \bigg) \nonumber \\
&&= -\frac{iD^{1/2}}{16 \pi^2}e^{i \frac{\sigma }{2s}} \bigg( \sum_{n=0}^{+\infty} \frac{D^{1/2}\tensor{}{_{;\mu}^{\mu}}}{D^{1/2}}  a_n  (is)^{n-2}\nonumber \\  
&&\; \; \; \; \; +\sum_{n=0}^{+\infty} \frac{-\tensor{\sigma}{_{;\mu}^{\mu}}}{2}a_n  (is)^{n-3}  +\sum_{n=0}^{+\infty} \frac{\tensor{\sigma}{_{;\mu}}\tensor{\sigma}{^{;\mu}}}{4}a_n(is)^{n-4} \nonumber \\
&&\; \; \; \; \; +\sum_{n=0}^{+\infty} \tensor{a}{_n _{;\mu}^{\mu}} (is)^n   +\sum_{n=0}^{+\infty}\frac{-D^{1/2}\tensor{}{_{;\mu}}}{D^{1/2}}\tensor{\sigma}{^{;\mu}} a_n(is)^{n-3} \nonumber \\
&&\; \; \; \; \; +\sum_{n=0}^{+\infty} \frac{2D^{1/2}\tensor{}{_{;\mu}}}{D^{1/2}} \tensor{a }{_n^{;\mu}} (is)^{n-2}  +\sum_{n=0}^{+\infty} (-\tensor{\sigma}{_{;\mu}}\tensor{a}{_n ^{;\mu}})(is)^{n-3}\bigg); \nonumber \\
\end{eqnarray}
exploiting \eqref{ham_jac_eq} and \eqref{div_geod_0}, one obtains
\begin{eqnarray}
&&\tensor{-\left\langle x',s |x'',0 \right\rangle}{_{;\mu}^{\mu}}= - \frac{iD^{1/2}}{16 \pi^2}e^{i \frac{\sigma }{2s}} \bigg( \sum_{n=0}^{+\infty} \frac{D^{1/2}\tensor{}{_{;\mu}^{\mu}}}{D^{1/2}}  a_n  (is)^{n-2} \nonumber \\ 
&&\; \; \; \; \; +\sum_{n=0}^{+\infty} (-2 a_n (is)^{n-3} +\sum_{n=0}^{+\infty} \frac{\sigma}{2}a_n(is)^{n-4}  \nonumber \\
&&\; \; \; \; \; +\sum_{n=0}^{+\infty} \tensor{a}{_n _{;\mu}^{\mu}} (is)^{n-2} +\sum_{n=0}^{+\infty} \frac{2D^{1/2}\tensor{}{_{;\mu}}}{D^{1/2}} \tensor{a }{_n^{;\mu}} (is)^{n-2} \nonumber \\
&& \; \; \; \; \;+\sum_{n=0}^{+\infty} (-\tensor{\sigma}{_{;\mu}}\tensor{a}{_n ^{;\mu}})(is)^{n-3}\bigg). 
\end{eqnarray}
Finally, equating lhs and rhs one obtains
\begin{eqnarray}
0&=&\frac{iD^{1/2}}{16 \pi^2} e^{i \frac{\sigma }{2s}} \bigg(\sum_{n=0}^{+\infty}(-n a_n)(is)^{n-3} +\nonumber \\
&&  \sum_{n=0}^{+\infty}\Delta^{-1/2}(\Delta^{1/2} a_n)\tensor{}{_{;\mu}^{\mu}}(is)^{n-2} +\nonumber \\
&&  \sum_{n=0}^{+\infty}(-\tensor{\sigma}{_{;\mu}}\tensor{a}{_n ^{;\mu}})(is)^{n-3}\bigg).
\end{eqnarray}
The necessary and sufficient condition for this equation to be satisfied for every $s$ is the multiplicative coefficient of every monomial $(is)^{k}$ be zero; therefore
\begin{eqnarray}
&&(is)^{-3}:   \tensor{\sigma}{_{;\mu}}\tensor{a}{_0 ^{;\mu}} = 0, \label{rec_rel1} \\
&&(is)^{k},\; k>-3:  \nonumber \\
&& \;\;\;\;\;    \tensor{\sigma}{_{;\mu}}\tensor{a}{_{n+1} ^{;\mu}}+(n+1) a_{n+1} = \Delta^{-1/2}(\Delta^{1/2} a_n)\tensor{}{_{;\mu}^{\mu}}. \label{rec_rel2}
\end{eqnarray}
In view of \eqref{par_displ_sc} and \eqref{par_displ_sc2} the equations for $a_0$ are solved by 
\begin{equation}
a_0(x,x') = I(x,x'),
\end{equation}
while the recurrence relation \eqref{rec_rel2} may be solved by integrating along each geodesic emanating from $x'$; in fact, multiplying the lhs by $\tau^{ '' n}$, where $\tau '' $ is the parameter labeling a point $x''$ on the geodesic between $x' (\tau '' = 0)$  and $x (\tau '' =\tau)$, one obtains
\begin{eqnarray}
&&\tau^{''n}\tensor{\sigma}{_{;\mu}}\tensor{a}{_{n+1} ^{;\mu}}(x(\tau''),x')+(n+1)\tau^{'' n} a_{n+1}(x(\tau''),x') = \nonumber \\
&=&    \tau^{'' n+1}\frac{d}{d\tau''}a(x(\tau''),x') + \left(\frac{d}{d\tau''} \tau^{'' n+1}\right)a_{n+1}(x(\tau''),x') \nonumber \\
&=& \frac{d}{d\tau''}\left(\tau^{'' n+1}  a_{n+1}(x(\tau''),x') \right),
\end{eqnarray} 
therefore
\begin{equation}
a_{n+1}(x,x') = \tau^{-n-1} \int_0 ^\tau d\tau'' \;  \tau^{'' n}\Delta^{''-1/2}(\Delta^{''1/2} a_n '')\tensor{}{_{;\mu''}^{\mu''}}
\end{equation}

\subsection{Series expansions}
Inserting \eqref{ansatz} in \eqref{corr_prescr_F}, we now get
\begin{eqnarray}
G(x,x') &=& \frac{\Delta^{1/2}}{(4\pi)^2}\int_0 ^{\infty}ds \; \frac{1}{s^2} e^{-i\left(m^2 s - \frac{\sigma}{2s}\right)} \sum_{n=0}^{+\infty} a_n (is)^n \nonumber \\
&=&  \frac{\Delta^{1/2}}{(4\pi)^2}  \sum_{n=0}^{+\infty} a_n \left(-\frac{\partial}{\partial m^2} \right)^n \int_0 ^{\infty}ds \;  e^{-i\left(m^2 s - \frac{\sigma}{2s}\right)}.
\end{eqnarray}
The latter integral has already been evaluated (eqs. \eqref{Feyn_Gr} and \eqref{Feynman_Hankel}). Breaking the Feynman propagator into its real and imaginary parts and making use of the expansions \eqref{exp_real_F_flat} and \eqref{exp_imag_F_flat}, we obtain, upon carrying out the differentiations with respect to $m^2$, 
\begin{eqnarray}
G(x,x') &=& \bar{G}(x,x')+\frac{1}{2}iG^{(1)}(x,x'), \nonumber \\
\bar{G}(x,x') &=& \frac{\Delta^{1/2} I }{8 \pi}\delta(\sigma)-\frac{\Delta^{1/2}  }{8 \pi}\theta(-\sigma) \bigg[\tfrac{1}{2}(m^2 I - a_{1})\nonumber \\
&&+\frac{2\sigma}{2^2 4}(m^4 I -2m^2 a_1 +2a_2)  \nonumber \\ 
&& +\frac{(2\sigma)^2}{2^2 4^2 6}(m^6 I -3m^4 a_1 +6m^2 a_2 -6a_3) + ... \bigg],  \\
G^{(1)}(x,x') &=& \frac{\Delta^{1/2} I }{4 \pi^2 \sigma}  \nonumber \\ 
&& +\frac{\Delta^{1/2}  }{2 \pi^2}(\gamma -\tfrac{1}{2}{\rm log }2 + \tfrac{1}{2}{\rm log }|2m^2 \sigma |)\bullet \nonumber \\ && \; \bullet\bigg[\tfrac{1}{2}(m^2 I -a_1)+ \frac{2\sigma}{2^2 4}(m^4 I -2m^2 a_1 + 2 a_2) + ... \bigg]  \nonumber \\
&&  -\frac{\Delta^{1/2}  }{2 \pi^2}\bigg[\tfrac{1}{4}m^2 I + \frac{2\sigma}{2^2 4}(\tfrac{5}{4}m^4 I-2m^2 a_1 + a_2)  \nonumber \\
&& \;  +\frac{(2\sigma)^2}{2^2 4^2 6} \left( \frac{5}{3}m^6 I -\frac{9}{2}m^4 a_1 +\frac{15}{2}m^2 a_2 -\frac{9}{2}a_3 \right) + ... \bigg] \nonumber  \\
&&  +\frac{\Delta^{1/2}  }{2 \pi^2} \bigg[ \left( \frac{a_2}{m^2}+\frac{a_3}{4m^4}+\frac{a_4}{2m^6}+... \right)  \nonumber \\
&& \; -\frac{2\sigma}{2^2 4}\left(\frac{a_3}{m^2}+\frac{a_4}{m^4}+ ... \right) + ... \bigg]. \label{exp_G1_curved}
\end{eqnarray}

Several comments must be made about these expansions. First, there is the obvious remark that they are useful only for small values of $\sigma$. However, this is precisely the domain in which we are often interested, particularly in renormalization theory. We note that the Feynman propagator has, at $\sigma = 0$, the same types of singularity in the presence of a gravitational field as it has in a flat empty space-time. The Green's function $\bar{G}$, which can be split into the advanced and retarded Green's functions, has a $\delta$-distribution singularity on the light cone and vanishes outside. We note, however, that when $m = 0$, it no longer vanishes inside the light cone as it does when space-time is flat and empty. Instead, we have 
\begin{equation}
\bar{G}=\frac{\Delta^{1/2} I }{8 \pi}\delta(\sigma)+\frac{\Delta^{1/2}  }{16 \pi}\theta(-\sigma)\left(a_1 -\frac{\sigma}{2}a_2 + \frac{\sigma^2}{2 \cdot  4} a_3 - ...\right).
\end{equation}
lt is important to observe in this connection that although the expansion in terms of the $a$'s can be used for $\bar{G}$ when $m = 0$, it cannot be used for $G^{(1)}$. This may be seen from the last line of \eqref{exp_G1_curved} which shows that an expansion in inverse powers of $m^2$ is involved. When $m$ is vanishing, alternative methods, based either on special properties of the fields or on perturbation theory, must be found for evaluating the Feynman propagator. 

\section{Conclusions}
This work has pedagogical purposes, hence the particular care in calculations, most of which are explicitly shown. A powerful formalism for gauge field theories has been described and has been used to obtain a \textit{manifesly covariant quantization} of such theories, even in curved space-time; an important application is the evaluation of physical observables, such as the \textit{stress-energy tensor} through point-split regularization: in any theory of interacting fields the set of currents that describe the interaction is of fundamental importance; in General Relativity, these currents are the components of the stress-energy tensor, therefore the \textit{main problem in developing a quantum field theory in curved space-time is precisely to understand the stress-energy tensor} (see \cite{dewitt1975quantum1}). 

One fundamental result is that it can always be expressed through Feynman's Green function and its derivatives (see \cite{bimonte2003photon},  \cite{christensen1976vacuum}, \cite{christensen1978regularization}, \cite{dewitt1975quantum1}) but the actual task is to give it meaning by some subtraction process. A regularization, or subtraction, process conventionally makes use of the vacuum state, but in a curved space-time, this notion is not trivial, as already stressed. Compared to the flat space-time case, in a curved background the resulting \textit{renormalized stress-energy tensor} is covariantly conserved, of course, but it \textit{possesses a state-independent anomalous trace} (see \cite{kay2006}).

Currently the author is working on the study of the stress-energy tensor for Maxwell's theory, along the lines of Christensen's work, which is a cornerstone in this area (see \cite{christensen1976vacuum} and \cite{christensen1978regularization}).
Last, it is important to emphasize that similar calculations appear interestingly also in the context of effective action in curved space-time, whose divergent part is essential to discuss renormalization group equations for the Newton constant and the cosmological constant (see \cite{giacchini2020}).

\section*{Acknowledgments}
The author is grateful to G. Esposito for his guidance and support; his love of physics is deeply inspiring.

\end{document}